\documentclass[preprint,12pt,authoryear]{elsarticle}
\usepackage{graphicx} 

\usepackage{amsthm}
\usepackage[english]{babel}
\usepackage{amsfonts}
\usepackage[psamsfonts]{eucal}
\usepackage[T1]{fontenc}
\usepackage[utf8]{inputenc}
\usepackage{amssymb,amsbsy,amsmath}

\numberwithin{equation}{section} 
\newtheorem{theorem}{Theorem}[section]
\newtheorem{corollary}{Corollary}[theorem]

\newtheorem{proposition}[theorem]{Proposition}

\usepackage[round]{natbib}
\usepackage{url}
\usepackage{hyperref}
\usepackage{verbatim}
\usepackage{booktabs}
\usepackage{color}
\usepackage{makeidx}
\usepackage{lmodern}
\usepackage{float}
\usepackage{physics}
\usepackage{multirow}
\usepackage{xcolor}
\usepackage{caption}
\usepackage{subfigure}
\usepackage{comment}
\usepackage{empheq}
\usepackage{mathrsfs, mathtools}
\usepackage{latexsym}
\usepackage{amsfonts}
\usepackage[psamsfonts]{eucal}
\usepackage{version}
\usepackage{fancyhdr}
\usepackage{verbatim}
\mathtoolsset{showonlyrefs}
\usepackage{bm}
\usepackage{rotating}

\usepackage{stackengine}

\usepackage{amsmath}
\usepackage{ulem}

\journal{}

\begin{document}

\begin{frontmatter}
\title{Hedging market risk and uncertainty via a robust portfolio approach}

\author[]{Adele Ravagnani$^{(1)}$, Mattia Chiappari$^{(2)}$, Andrea Flori$^{(2)}$, Piero Mazzarisi$^{(1)}$, Marco Patacca$^{(3)}$\\
\medskip
$^{(1)}${\it \small Department of Economics and Statistics, University of Siena}\\
$^{(2)}${\it \small Department of Management, Economics and Industrial Engineering, Politecnico di Milano}\\
$^{(3)}${\it \small Department of Economics, University of Perugia}
}

\begin{abstract}
Shorting for hedging exposes to risk when the market dynamics is uncertain. Managing uncertainty and risk exposure is key in portfolio management practice. This paper develops a robust framework for dynamic minimum-variance hedging that explicitly accounts for forecast uncertainty in volatility and covariance estimation to achieve empirical stability and reduced turnover, further improving other standard performance metrics. The approach combines high-frequency realized variance and covariance measures, autoregressive models for multi-step volatility forecasting, and a box-uncertainty robust optimization scheme. We derive a closed-form solution for the robust hedge ratio, which adjusts the standard minimum-variance hedge by incorporating variance and covariance forecast uncertainty. Using a diversified sample of equity, bond, and commodity ETFs over 2016-2024, we show that robust hedge ratios are more stable and entail lower turnover than standard dynamic hedges. While overall variance reduction is comparable, the robust approach improves downside protection and risk-adjusted performance, particularly when transaction costs are considered. Bootstrap evidence supports the statistical significance of these gains.

\end{abstract}




\end{frontmatter}

\section{Introduction}

Hedge ratios are fundamental tools in financial risk management, as they represent the optimal short position in a hedging instrument required to minimize portfolio variance while holding a long position equivalent to one unit of the hedged asset. Time variation and instability in hedge ratios represent a concrete economic problem for both institutional and retail investors. Although dynamic models, such as multivariate DCC-GARCH frameworks \citep{kroner1993multivariate,engle2002dynamic}, improve in-sample variance minimization, they often generate highly volatile hedge ratios that require frequent portfolio rebalancing. Research also documents that the optimal hedge ratio varies across asset classes and across market regimes, with sizeable implications for hedging practice \citep{sadorsky2012correlations,basher2016hedging,ahmad2018optimal,antonakakis2018oil,antonakakis2020oil,dutta2020assessment,chiappari2025hedging}. At the same time, multiscale and wavelet-based analyses show that hedge effectiveness depends on the investment horizon, suggesting that a one-size-fits-all hedge may misalign with investors’ time preferences \citep{lien2007wavelet,conlon2011multiscale}. In practice, this instability translates into substantial transaction costs, market impact, and operational complexity, potentially offsetting the theoretical gains from dynamic hedging. Moreover, hedge ratios are extremely sensitive to volatility and covariance estimation error, especially during periods of market stress when correlations spike and forecast uncertainty increases \citep{fabozzi2012}.  Empirical evidence shows that small perturbations in variance–covariance inputs can lead to large shifts in optimal portfolio weights, a phenomenon commonly referred to as estimation risk or error maximization \citep{yin2021}. This instability not only reduces net hedge effectiveness after costs, but may also expose investors to unintended risk concentrations precisely when protection is most needed. Consequently, designing hedge ratios that explicitly account for forecast uncertainty and reduce excessive turnover is not merely a statistical refinement, but an economically meaningful improvement in real-world risk management.

In this paper, we propose a tractable methodology for dynamic, minimum-variance (MV) hedging that explicitly incorporates forecast uncertainty, is tailored to realized measures of volatility and covariance, and is flexible for portfolio rebalancing at a generic time scale. Our approach builds on three pillars. First, we exploit high-frequency realized variance and covariance estimators to obtain more responsive and lower-noise inputs for hedge construction \citep{andersen2003,andersen2009}. Second, we model and forecast these realized measures with parsimonious autoregressive specifications (including HAR-like variants) that capture persistence and long-memory features while remaining computationally light \citep{corsi}. Third, we embed forecasts into a robust optimization framework with box uncertainty sets for variances and covariances, producing hedge ratios that minimize worst-case portfolio variance over a plausible range of risk inputs. This formulation reduces sensitivity to estimation errors, limits overreaction to transient volatility spikes, and allows for solutions at a generic time scale, while preserving analytical tractability.

Empirically, we evaluate the proposed robust hedge ratios on a diversified basket of equity, fixed-income, and commodity Exchange Traded Funds (ETFs), comparing them to standard benchmarks under realistic transaction-cost specifications. We demonstrate that robustness produces smoother hedge ratio dynamics, which translates into lower portfolio turnover and reduced transaction costs. At the same time, we show that including uncertainty preserves the core benefits of standard dynamic hedging models in terms of variance reduction, while delivering more consistent performance across different market environments. Our empirical results indicate that ``robust hedging'' either improves or maintains key risk-adjusted financial metrics, including profitability, Sharpe and Omega ratios, and measures of downside risk, thereby offering investors a more reliable and cost-efficient tool for managing market risk without compromising performance.

The remainder of the paper is organized as follows. Section \ref{sec:method} presents the robust hedging methodology and its econometric implementation; Section \ref{sec:data} describes the ETF data and realized measure construction; Section \ref{sec:results} reports empirical findings; Section \ref{sec:concl} concludes with implications and avenues for further research. Finally, \ref{sec:app1} and \ref{sec:app2} discuss theoretical aspects, while \ref{app:plots}, \ref{app:HAR}, and \ref{app: performance} present additional figures and results.

\section{Methods}\label{sec:method}
Time-varying hedge ratios require frequent rebalancing, which is typically exposed to fluctuations and various forms of market instability, with a significant impact on costs. To tackle this issue, we propose a novel methodology to dynamically estimate minimum-variance (MV) hedge ratios \citep{johnson1960}, based on robust portfolio optimization \citep{fabozzi2012, yin2021} and autoregressive models for realized volatility and covariance.

\subsection{Robust hedge ratio} Given an instrument $S$ and the hedging instrument $F$, let $\Sigma_{t+\tau}$ be the covariance matrix between the returns of the two instruments over the portfolio time horizon $\tau\geq 1$, with $t$ the time at which the hedged portfolio is built. For notational simplicity and without the risk of confusion, we remove redundant notation and indicate $\Sigma=\Sigma_{t+\tau}$. We consider a robust portfolio optimization based on box uncertainty for the variances of the two instruments and the covariance between them. In particular, the variance of $S$ is $\tilde{\sigma}^2_{S}\in [\sigma^2_{S}-\Theta_{S},\sigma^2_{S}+\Theta_{S}]$ for some $\sigma_S^2>0$ with uncertainty value $\Theta_{S}\geq 0$, the variance of the hedging instrument is $\tilde{\sigma}^2_{F}\in [\sigma^2_{F}-\Theta_{F},\sigma^2_{F}+\Theta_{F}]$ for some $\sigma_F^2>0$ with uncertainty value $\Theta_{F}\geq 0$, and the covariance is $\tilde{\sigma}_{SF}\in [\sigma_{SF}-\Theta_{SF},\sigma_{SF}+\Theta_{SF}]$ for some $\sigma_{SF}>0$ with uncertainty value $\Theta_{SF}\geq 0$. As such, the box uncertainty set is defined as 
\begin{equation}
\begin{split}
\mathcal{U}:=\{(\tilde{\sigma}_S^2, \tilde{\sigma}_F^2, \tilde{\sigma}_{SF}): & \:\tilde{\sigma}_S^2\in[\sigma^2_{S}-\Theta_{S},\sigma^2_{S}+\Theta_{S}],\:\tilde{\sigma}_F^2\in [\sigma^2_{F}-\Theta_{F},\sigma^2_{F}+\Theta_{F}], \\
& \:\tilde{\sigma}_{SF}\in [\sigma_{SF}-\Theta_{SF},\sigma_{SF}+\Theta_{SF}] 
\}.
\end{split}
\end{equation}
The robust portfolio optimization problem to solve in order to find the optimal hedge ratio can be stated as the following {\it unconstrained} min-max problem
\begin{equation}\label{eq:1}
\min_{h \in \mathbb{R}}
\;
\max_{(\tilde\sigma_S^2,\tilde\sigma_F^2, \tilde{\sigma}_{SF})\in\mathcal U}
w^{\top}\Sigma w
\end{equation}

where $w:=(1,-h)$ represents the portfolio weights and $h$ is the hedge ratio, i.e., the quantity of the hedging instrument to short over the horizon~$\tau$. Under this parameterization, the objective function, namely the variance of the hedged portfolio, reads as
$$
w^{\top}\Sigma w= \sigma_S^2+h^2\sigma_F^2-2h\sigma_{SF}.
$$
\begin{proposition}[Robust hedge ratio with box uncertainty]\label{prop:1}
Let $\Theta_S,\Theta_F, \Theta_{SF}\ge 0$ and let us assume $\sigma_F^2+\Theta_F>0$. Consider the robust optimization (min-max) problem
\[
\min_{h\in\mathbb{R}}
\;
\max_{(\tilde\sigma_S^2,\tilde\sigma_F^2, \tilde\sigma_{SF})\in\mathcal U}
\Big(
\tilde\sigma_S^2
+ h^2 \tilde\sigma_F^2
- 2 h \tilde\sigma_{SF}
\Big).
\]
Then, the problem admits a unique global minimizer given by
\begin{equation}\label{eq:h_cov}
h^*=\frac{sgn(\sigma_{SF}) \Big(|\sigma_{SF}| - \Theta_{SF}\Big)^+}{\sigma_F^2+\Theta_F}.
\end{equation}
\end{proposition}
\begin{proof} See \ref{sec:app1}.\end{proof}

As it will emerge from empirical results, this approach turns out to be highly conservative therefore, we will focus on a sub-case of this problem where the covariance is consider fixed. This amounts to assume $\Theta_{SF} = 0$ in Proposition \ref{prop:1}.
\begin{corollary}
[Robust hedge ratio with box uncertainty only on variances]\label{corollary:1}
Let us consider Proposition \ref{prop:1} with $\sigma_{SF}\in\mathbb{R}$ fixed and so, $\Theta_{SF} = 0$. The unique global minimizer of the robust optimization problem is
\begin{equation}\label{eq:h}
h^*=\frac{\sigma_{SF}}{\sigma_F^2+\Theta_F}.
\end{equation}
\end{corollary}

\subsection{Modeling volatility and uncertainty} To explicitly model the squared volatility and the covariance over the horizon $\tau$ and the corresponding uncertainty, we consider an autoregressive model of order $p$, namely an AR(p) model
$$
y_{t+1}=\phi_0+\sum_{j=0}^{p-1}\phi_{i+1}y_{t-j}+\eta_{t+1}
$$
where $\phi_1,\ldots,\phi_p$ are the parameters of the model, and $\eta_{t+1}$ is a white noise accounting for both innovation and measurement errors, with the assumption of zero mean and finite second-order moment.\footnote{The parameters of the AR(p) process are estimated via maximum likelihood and, for simplicity, the corresponding estimation error is assumed to be small when compared with other error sources in the analysis below.} Finally, $y_t$ represents a general realized measure of either variance or covariance. Notice that the order of the process is able to capture the long memory of volatility for some proper $p\geq 1$. If $y_t$ is intended as a measure of the realized variance over a unit horizon at time $t$, i.e., $\hat{\sigma}_t^2$, it follows that the estimated volatility over a horizon $\tau$ is the square root of the integrated variance over $\tau$ steps, namely $\hat{\sigma}_{t+\tau}^2 = \sum_{j=1}^\tau y_{t+j}$ where $\hat{y}_{t+j}$ is the $j$-th step-ahead forecast of the AR(p) process. Then, the uncertainty interval $\Theta_{\tau}$ associated with $\hat{\sigma}_{t+\tau}^2$ is described by the standard deviation of the forecast error $e_\tau\equiv\left(\sum_{j=1}^{\tau}y_{t+j}\right)-\left(\sum_{j=1}^\tau\hat{y}_{t+j}\right)$, which can be computed in closed-form for linear autoregressive processes. The full derivation is provided in \ref{sec:app2}, together with a comparison between the closed-form expression and its empirical counterpart. Analogously, we can derive the integrated covariance over $\tau$ steps and its uncertainty interval. Modeling volatility with AR processes is motivated by the seminal paper \citep{corsi}, proposing a simple AR-type model for volatility that considers aggregations over daily, weekly, and monthly time scales, each one associated with a different autoregressive coefficient. Despite its simplicity, the so-called HAR model can reproduce the stylized facts, in particular the long-memory behavior, as evidenced by excellent forecasts. In our work, we employ a model specification similar to HAR for robustness checks. However, we show that simple AR processes are better suited to the general scope of the paper.

Restoring time indexing in Eq. \eqref{eq:h_cov}, the optimal hedge ratio can be expressed in terms of the variance of $F$ and the covariance between $S$ and $F$, that is
\begin{equation}\label{HR_cov_complete}
  h^*_t= \frac{sgn(\sigma_{SF, t+\tau}) \Big(|\sigma_{SF, t+\tau}| - \Theta_{SF, \tau}\Big)^+}{\sigma_{F, t+\tau}^2+\Theta_{F, \tau}}
\end{equation}
with uncertainty intervals $\Theta_{F, \tau}$ and $ \Theta_{SF, \tau}$ over $\tau$ computed as in \ref{sec:app2}. In particular, $\sigma_{F,t}^2$ and $\sigma_{SF,t}$ are measured as the realized variance (RV) and realized covariance (RCV) estimators \citep{andersen2003, andersen2009}. For a given day $t$, it is
\begin{equation}
    \begin{split}
        RV_{x, t} &= \frac{M}{M_{x,t}} \sum_{i=1}^{M_{x,t}} r^2_x(i) \\
        RCV_{xy, t} &= \frac{M}{M_{xy,t}}\sum_{i=1}^{M_{xy,t}} r_x(i)r_y(i) \\
    \end{split}
\end{equation}
where
$$r_x(i) = \log\Bigg(\frac{p_x^{close}(i)}{p_x^{close}(i-1)}\Bigg),$$
with $x$ and $y$ denoting two instruments, $i$ the 5-minutes intervals on day $t$, $M_{x,t}$ ($M_{xy,t}$) the total number of these intervals with available return data for $x$ (for both $x$ and $y$) on $t$, $M$ the maximum number of 5-minute intervals,\footnote{In our estimations, we drop the first and last half hours of the trading day. As a result, we consider the time window from 10:00 AM to 3:30 PM. This gives: $M = 66$.} and $p_x^{close}(i)$ the closing price of $x$ in the interval $i$. The factors $M/M_{x,t}$ and $M/M_{xy,t}$ allow to account for missing observations throughout the trading day, ensuring that the realized estimators remain comparable across days with different data availability.

\section{Data}\label{sec:data}
In order to test our methodology, we focus on a diversified basket of ETFs related to the equity and bond markets, and to commodities; they are listed in Table \ref{tab:etfs}. 

The realized variances and covariances, as well as the returns, of the selected ETFs, are estimated by relying on historical intraday data\footnote{Source: \href{https://www.kibot.com/}{https://www.kibot.com/}.} related to the period 2016-2024. In Table \ref{tab:etfs_stat}, we report some summary statistics of our data set. Finally, Fig.~\ref{fig:correlations_hist} in \ref{app:plots} presents the histogram of pairwise return correlations, showing substantial dispersion across instrument pairs.

\begin{table}[]
    \centering
    \small
    \begin{tabular}{c|c|c}
         Symbol & Description & Class\\
         \hline\hline
         IVV	&ISHARES S\&P 500 INDEX & Equity\\
ICLN	& ISHARES GLOBAL CLEAN ENERGY & Equity \\
QQQM&	INVESCO NASDAQ 100 ETF & Equity \\
ASHR	&DEUTSCHE X-TRACKERS HARVEST CSI300 CHNA & Equity\\
EWH &	ISHARES MSCI HONG KONG INDEX FUND & Equity\\
IEV	& ISHARES S\&P EUROPE 350 INDEX FUND & Equity\\	
CORP	&PIMCO INVESTMENT GRADE CORPORATE BD ETF & Bond\\
IGOV&	ISHARES INTERNATIONAL TREASURY BOND & Bond\\
GOVT	&ISHARES US TREASURY BOND & Bond\\	
BNO	& UNITED STATES BRENT OIL FUND LP ETV & Commodity\\
UNG	& UNITED STATES NATURAL GAS & Commodity\\
AAAU &	GOLDMAN SACHS PHYSICAL GOLD ETF SHARES & Commodity\\
GSG	& ISHARES S\&P GSCI COMMODITY-INDEXED TRUST & Commodity\\

    \end{tabular}
    \caption{Basket of ETFs in the data set.}
    \label{tab:etfs}
\end{table}

\begin{table}[]
    \centering
    \small
    \begin{tabular}{c|c|c|c|c}
         ETF & Start day & N. days& Avg. daily n. & Avg. daily volume \\
         & & in 2016-2024 & 5-minute intervals & ($\times10^3$) \\
         \hline\hline
         IVV & 11/30/2005 & 2264 & 65.72 $\pm$ 0.06 & 2260.03 \\ 
ICLN & 6/25/2008 & 2264 & 51.33 $\pm$ 0.41 & 1516.71 \\ 
QQQM & 9/10/2012 & 1061 & 63.13 $\pm$ 0.21 & 514.27 \\ 
ASHR & 11/6/2013 & 2264 & 64.97 $\pm$ 0.08 & 1961.17 \\ 
EWH & 11/30/2005 & 2264 & 65.55 $\pm$ 0.06 & 2679.14 \\ 
IEV & 11/30/2005 & 2264 & 62.04 $\pm$ 0.12 & 313.98 \\ 
CORP & 9/21/2010 & 2264 & 34.36 $\pm$ 0.23 & 46.53 \\ 
IGOV & 1/30/2009 & 2264 & 43.47 $\pm$ 0.26 & 131.53 \\ 
GOVT & 2/24/2012 & 2264 & 64.33 $\pm$ 0.09 & 3993.48 \\ 
BNO & 6/2/2010 & 2264 & 59.84 $\pm$ 0.19 & 469.20 \\ 
UNG & 4/18/2007 & 2264 & 65.62 $\pm$ 0.06 & 4524.23 \\ 
AAAU & 8/15/2018 & 1604 & 48.82 $\pm$ 0.50 & 672.59 \\ 
GSG & 7/21/2006 & 2264 & 59.12 $\pm$ 0.19 & 502.81 \\ 
    \end{tabular}
    \caption{Summary statistics of the ETF basket.}
    \label{tab:etfs_stat}
\end{table}

\section{Results}\label{sec:results}
After estimating the realized variances and covariances of the financial instruments of interest, we assess the stationarity of the resulting time series by means of Augmented Dickey–Fuller tests. Then, we run the robust methodology described in Section \ref{sec:method}. In order to ensure positivity of the realized variances, we take their logs before fitting AR models. By following \cite{demetrescu2019}, in order to correct the bias which emerges from the reverse transformation from log
forecasts to realized variances, we rely on the variance-based correction. The latter consists in the following approximation
\begin{equation*}
    \mathbb{E}[RV_{t+j}|\mathcal{F}_t] = \exp{\log\mathbb{E}[RV_{t+j}|\mathcal{F}_t]}\simeq \exp{\mathbb{E}[\log RV_{t+j}|\mathcal{F}_t] + \frac{1}{2}Var(\epsilon_{j})}
\end{equation*}
where $\epsilon_{j}$ is the $j$-step-ahead forecast error estimated in-sample, and $\mathcal{F}_t$ represents the information set available at time $t$. 

The AR(p) processes are estimated on the train set, i.e., in-sample, and then, the hedge ratios are computed out-of-sample. We quantify the hedged portfolio effectiveness by means of the following metrics
\begin{equation}\label{eq_def_HE}
    \begin{split}
        HE & = 1 -\frac{Var(R_h)}{Var(R_S)} \\
        HE_C & = 1 -\frac{Var(R_h|R_S 
 < \delta)}{Var(R_S| R_S < \delta)} \\
 HE_R & = \frac{\mathbb{E}[R_h| R_S < \delta]}{\mathbb{E}[R_S| R_S < \delta]}
    \end{split}
\end{equation}
where $\delta$ is a given threshold, $R_{x,t}$ is the return of $x$ on day $t$, and $Var(R_S)$ and $Var(R_h)$ are the variances of the daily returns of the unhedged and hedged portfolio, respectively. In particular, at time $t + \tau$ the latter is $R_{h, t + \tau} = R_{S, t + \tau} - h_t R_{F, t + \tau}$. 

Hedge effectiveness ($HE$) is a standard measure of hedging performance in this context \citep{chen2024hedge}, quantifying the reduction in the hedged portfolio's variance relative to the long position on investment $S$ (unhedged portfolio). We introduce here a threshold-based modification of $HE$, namely {\it conditioned} hedge effectiveness ($HE_C$), to measure hedge effectiveness for returns below a given threshold, with a focus on left-sided $S$ movements. This is to better capture the role of robustness in the context of negative market events. Finally, $HE_R$ is introduced with a similar objective, but with a focus on expected returns instead of variances, to measure the extent to which the loss of the hedged portfolio is reduced relative to the long position during left-tail movements in $S$.

\subsection{The robust hedge ratio is almost always zero}
The robust hedge ratio we propose in Eq. \eqref{HR_cov_complete} is derived by introducing uncertainty intervals for the variances of the two instruments and the covariance between them. However, when it is derived for the couples of ETFs in our dataset, we find that this approach is extremely conservative. The uncertainty on the covariance is usually higher than its absolute value and so the term
$$\Big(|\sigma_{SF, t + \tau}| - \Theta_{SF, \tau}\Big)^+$$ in Eq. \eqref{HR_cov_complete} drives the hedge ratio to $0$. In Table \ref{tab:theta_ratios} the distributions of the ratios between the uncertainty intervals, $\Theta_{SF, \tau}$ and $\Theta_{F, \tau}$, and the corresponding mean values are compared for three prediction horizons $\tau$ and two AR(p) models. We notice that the ratios related to the covariance tend to assume higher values than those of the variance and they are mainly higher than $1$. This second issue drives the robust hedge ratio $h_t^*$ of Eq. \eqref{HR_cov_complete} to $0$, as it is shown in Table \ref{tab:nonzeroHR}. Here details about the distribution of the fraction of days such that $h_t^*$ is different from $0$ are provided and it is evident that in the vast majority of days the portfolio weight of the hedging instrument is $0$. Therefore, in the following we consider a less conservative robust approach by assuming that the covariance is fixed and focusing on the uncertainty on squared volatility. This corresponds to the hedge ratio defined in Eq. \eqref{HR_cov_complete} with $\Theta_{SF, \tau} = 0$, i.e., 
\begin{equation}\label{HR_cov}
    h^*_t= \frac{\sigma_{SF, t+\tau}}{\sigma_{F, t+\tau}^2+\Theta_{F, \tau}},
\end{equation}
which is the time-varying counterpart of Eq. \eqref{eq:h} in Corollary \ref{corollary:1}.

\begin{table}[]
    \centering
    \small
    \begin{tabular}{c|c|c|c|c||c|c|c}
         & $\tau$ & $q_{0.25}(K_{SF})$ & $q_{0.50}(K_{SF})$ & $q_{0.75}(K_{SF})$ & $q_{0.25}(K_{S})$ & $q_{0.50}(K_{S})$ & $q_{0.75}(K_{S})$\\
        \hline
        $p=1$ & 1 &  4.2 &  6.2 &  11.8 & 1.6 &  2.4 &  3.7 \\ 
        & 5 & 2.5 &  3.3 &  6.1
& 1.0 &  1.4 &  2.2 \\ & 10 & 2.1 &  2.5 &  5.0
& 0.9 &  1.3 &  1.8 \\
        \hline\hline 
        $p=5$ & 1 &  4.8 &  7.6 &  11.9
& 1.7 &  2.7 &  4.6 \\
& 5 &  2.5 &  3.8 &  6.4
& 1.0 &  1.4 &  2.6 \\
& 10 &  2.1 &  3.1 &  5.2
& 0.9&  1.3 &  2.4\\
    \end{tabular}
    \caption{Details about the distributions of $K_{SF} = \Theta_{SF, \tau}/|\langle\sigma_{SF, t + \tau}\rangle_t|$, $K_S = \Theta_{S, \tau}/\langle\sigma^2_{S, t + \tau}\rangle_t$.
    }
    \label{tab:theta_ratios}
\end{table}

\begin{table}[]
    \centering
    \small
    \begin{tabular}{c|c|c|c|c|c|c}
         &$\tau$ & mean($f_{h^*}$) & std($f_{h^*}$) & $q_{0.9}(f_{h^*})$ &$q_{0.95}(f_{h^*})$ & $q_{0.99}(f_{h^*})$ \\
        \hline
          $p=1$ & 1 &  0.0019 &  0.0055 &  0.0046 &  0.0066 &  0.0301\\
         & 5 &  0.0025 &  0.0133 &  0.0055 &  0.011 &  0.0584\\
 & 10 &  0.0641 &  0.2445&  0.0000 &  1.0000 &  1.0000\\
        \hline\hline
         $p=5$  & 1 &  0.0044 &  0.0084 &  0.0136 &  0.0158 &  0.0443 \\
 & 5 &  0.0131 &  0.0324 &  0.0234 &  0.0459 &  0.1778 \\
 & 10 &  0.0206&  0.0558&  0.0378 &  0.1071&  0.2809 \\
    \end{tabular}
    \caption{Details about the distribution of $f_{h^*}$, that is the fraction of days such that the robust hedge ratio $h_t^*$ of Eq. \eqref{HR_cov_complete} is different from $0$.
    }
    \label{tab:nonzeroHR}
\end{table}

\subsection{Robust vs. standard approach}

\begin{figure}
\includegraphics[width=0.5\linewidth]{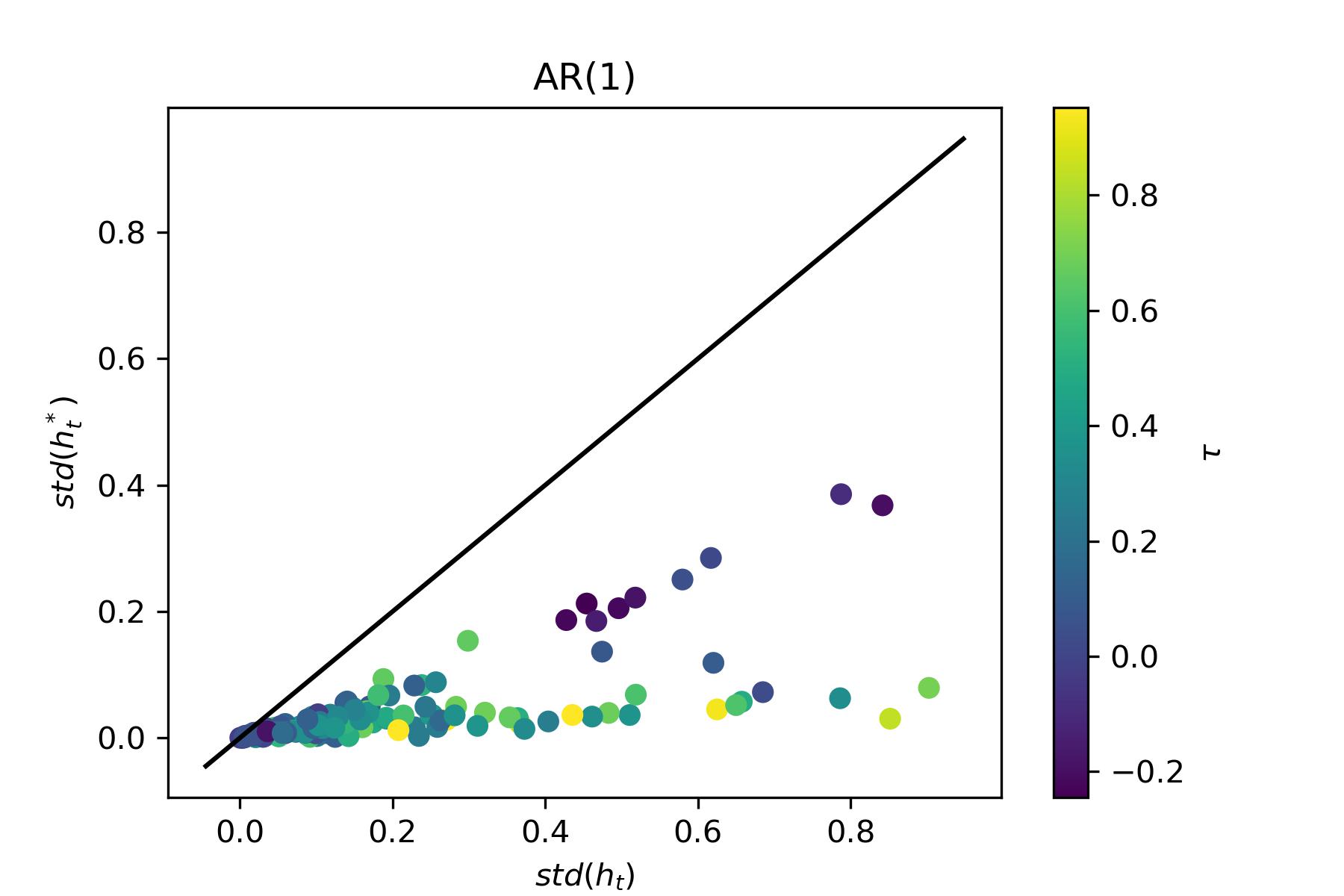} \includegraphics[width=0.5\linewidth]{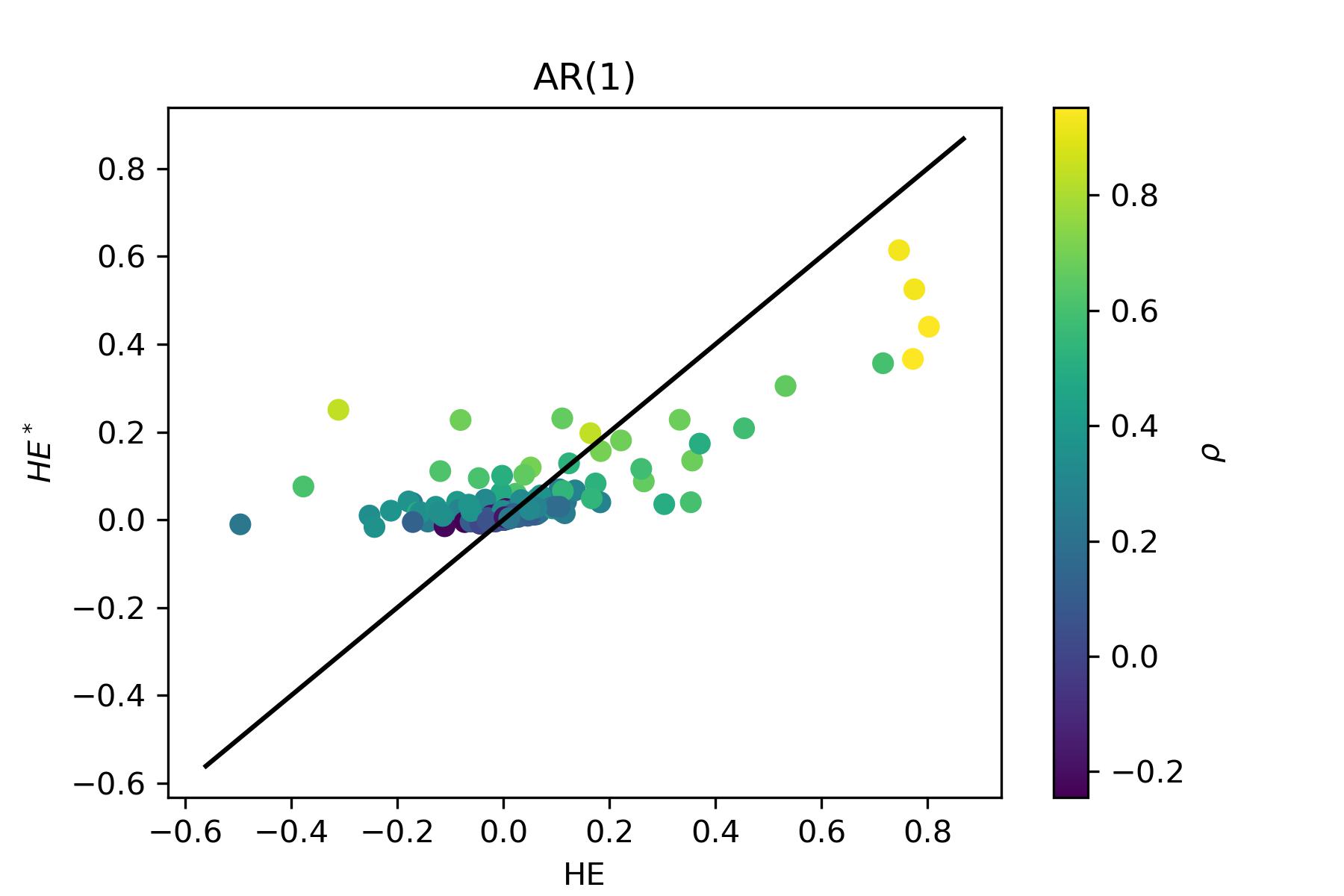}
\includegraphics[width=0.5\linewidth]{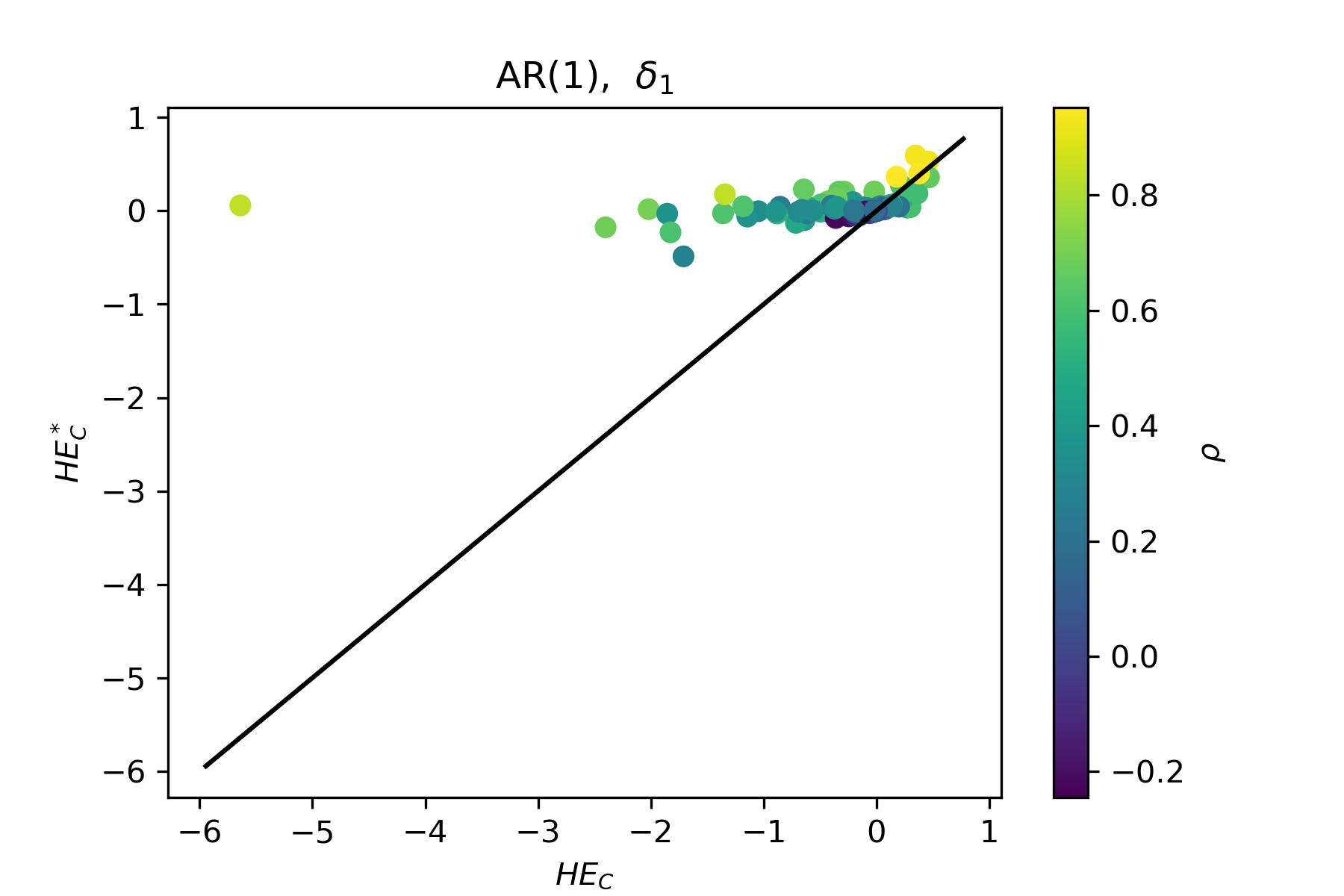}
\includegraphics[width=0.5\linewidth]{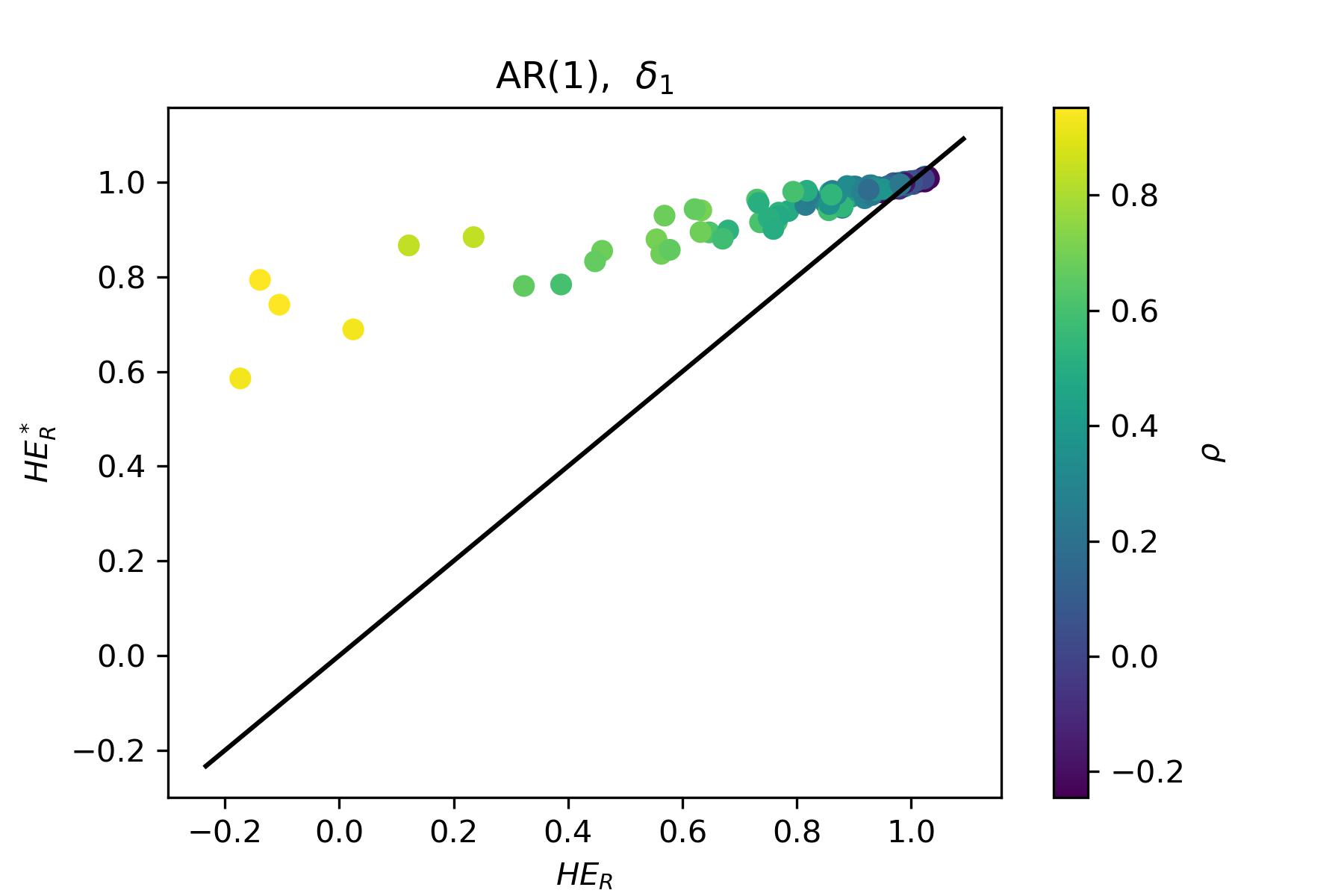}
    \caption{Comparison between the outputs of the standard and robust methodologies when realized variances and covariances are fitted by AR(1) processes. Each point is associated with a pair of instruments and its color is the return correlation $\rho$ of the pair. The threshold $\delta_1$ is the first quartile of the asset returns. The dark line corresponds to the bisector.}
    \label{fig:AR1}
\end{figure}

One crucial finding of our paper is that, being more conservative, the robust approach yields lower standard
deviations of the hedge ratio $h_t^*$ (Eq. \eqref{HR_cov}) compared to the non-robust one $h_t$ (Eq. \eqref{HR_cov} with $\Theta_{F, \tau} = 0$), as it is represented in the top left panel of Fig. \ref{fig:AR1}, when $\tau= 1$ and the variances and covariances are fitted with AR(1) processes. This result is extremely relevant for practitioners since it allows to limit the rebalancing process and transaction costs. 

By computing the effectiveness metrics of Eq. \eqref{eq_def_HE}, we observe that our robust methodology does not always outperform the standard one in terms of $HE$. However, even when the standard method yields negative hedge effectiveness, the robust approach produces values that are approximately zero or positive. This advantage is illustrated in the top right panel of Fig. \ref{fig:AR1}. Additionally, as it is shown in the bottom panels of Fig. \ref{fig:AR1}, it is clear that the proposed robust methodology outperforms the standard one in the worst-case scenarios, i.e., when the conditioned metrics $HE_C$ and $HE_R$ are considered.\footnote{In the plots of Fig. \ref{fig:AR1}, the conditioned metrics have been computed with a threshold that is the first quartile of the asset returns. Consistent results can be obtained with other choices of the threshold, e.g., $0$, and they are available upon request.} 

In the effectiveness plots, we can also observe that the pairs of instruments with high (low) correlation are associated with higher (lower) discrepancy between the standard and robust values, being farer from the bisector. Analogous conclusions can be drawn if the color of the heatmaps is a variable representing the class co-occurence between instruments of each pair, as it is represented in Fig. \ref{fig:pairtype} of \ref{app:plots}. This is explained by the low values of the hedge ratios that we obtain when the two instruments are low-correlated.  As a consequence, the hedged portfolio does not differ consistently from its unhedged counterpart resulting in hedge effectiveness metrics approximately equal to zero. In Fig. \ref{fig:meanHR_compare} of \ref{app:plots}, we compare the histograms of the hedge ratios for four instrument pairs characterized by different levels of return correlation. Finally, we observe that all the findings outlined above also hold when AR processes with longer memory, e.g., AR(5), model the variances and covariances. This is shown in Fig. \ref{fig:AR5} of \ref{app:plots}.

\subsection{Predictions using longer horizons}

\begin{figure}
\includegraphics[width=0.5\linewidth]{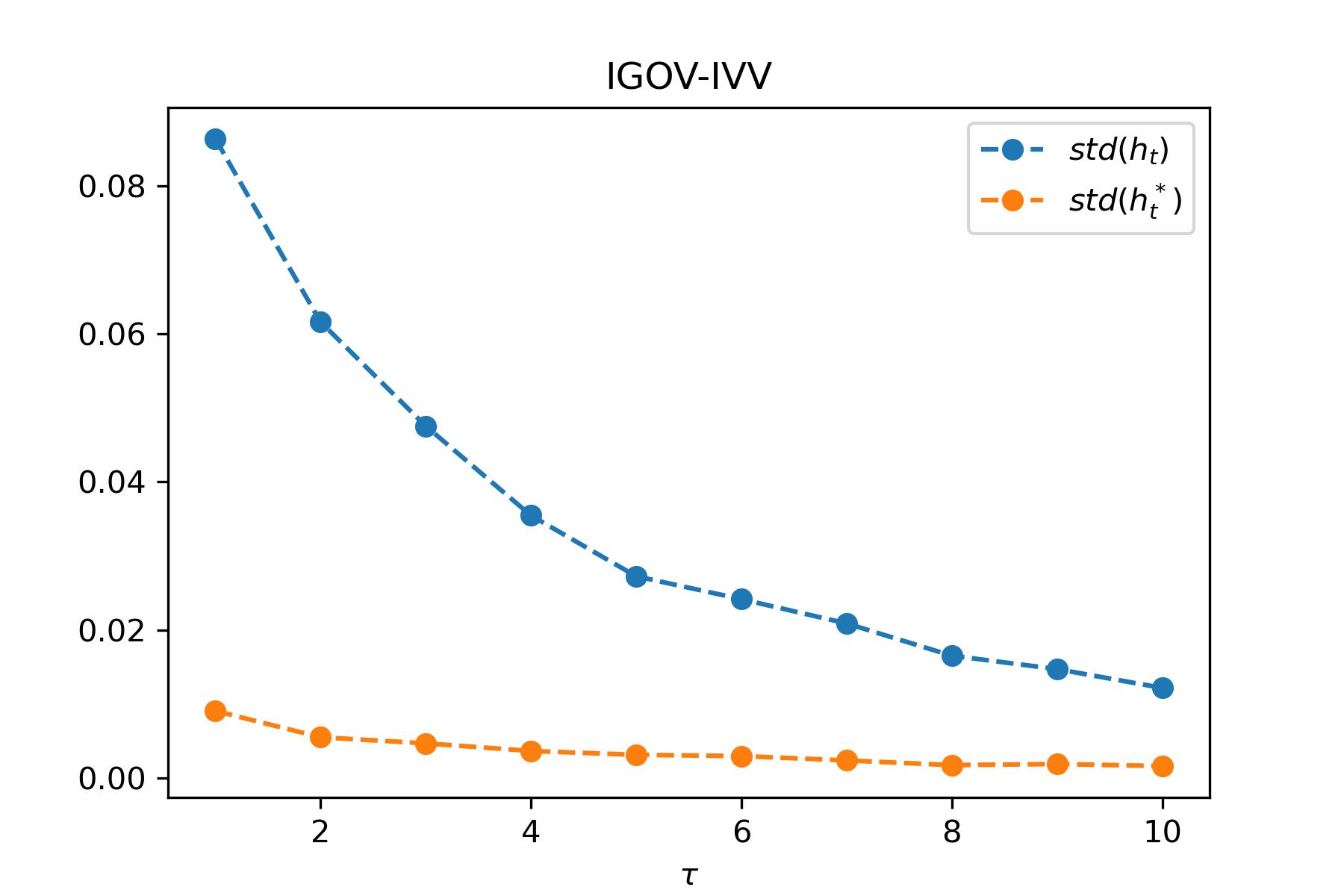} \includegraphics[width=0.5\linewidth]{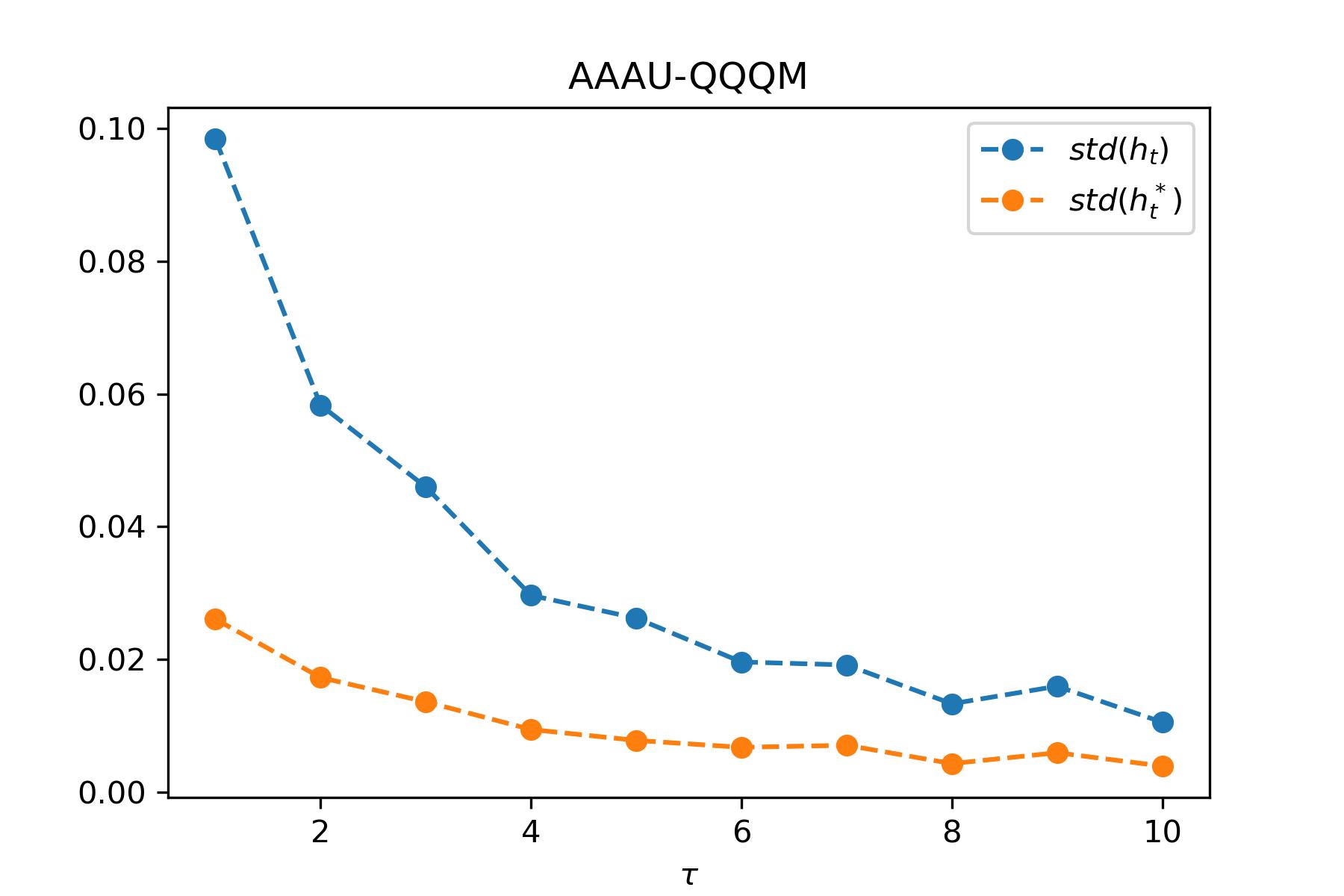}
    \caption{Standard deviation of the hedge ratios as a function of the prediction horizon $\tau$ employed for predictions. Two different instrument pairs are considered: in the title, the left label refers to the hedged asset and the right to the hedging instrument.}
    \label{fig:stdHR_vs_tau}
\end{figure}

In the previous subsection, we have seen that the robust approach has lower standard deviations of the hedge ratio than the standard method and this could be crucial in order to limit transaction costs. However, it is straightforward to see that limiting portfolio rebalancing could also be achieved by predicting hedge ratios with longer horizons $\tau$. Let us shed light on this issue. In Fig. \ref{fig:stdHR_vs_tau}, we plot the standard deviation of hedge ratios as a function of $\tau$ for both the standard and robust methods, and two pairs of instruments in our dataset\footnote{Similar outputs are observed for the other pairs in our dataset.}. A decreasing trend is evident especially for the standard method while the robust approach is associated to more stability. Moreover, concerning the performances of the strategies in terms of the effectiveness metrics, employing longer prediction horizons leads to higher values of $HE^*$, slightly higher values of $HE_C^*$, and lower values of $HE_R^*$. 
These effects are particularly evident for pairs with high correlation of the returns, as shown in Fig. \ref{fig:HEr_tau1_vs_tau10}, where we compare the metrics for $\tau=1$ and $\tau = 10$, and AR(1).\footnote{Analogous results are obtained when AR(5) is employed to fit realized variances and covariances.} 
Therefore, when realized variances and covariances are predicted over longer horizons, the hedge ratio exhibits less variability, and the variance reduction in portfolio returns obtained through hedging is larger. However, this comes at the expense of lower expected returns in worst-case scenarios.

\begin{figure}
\includegraphics[width=0.5\linewidth]{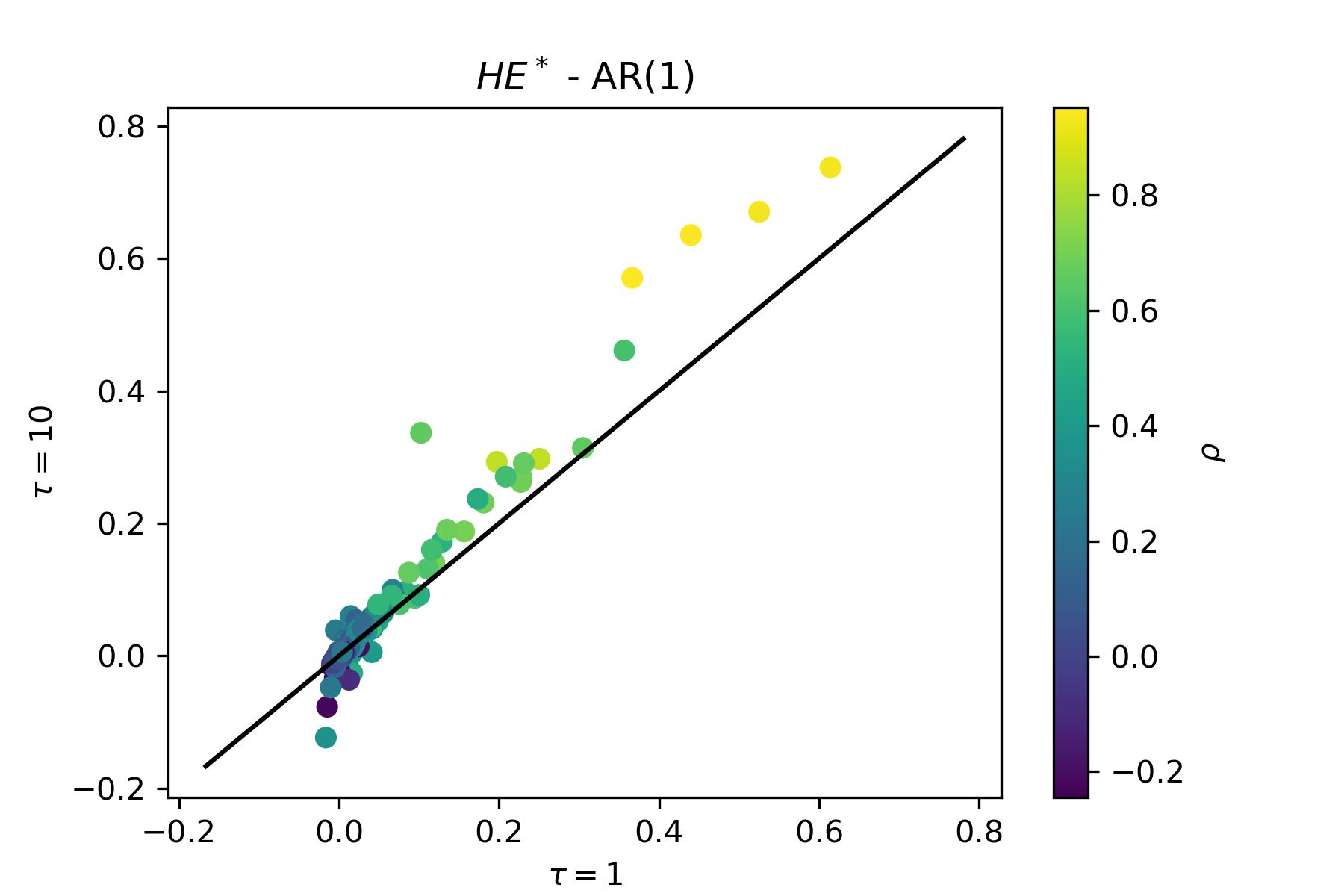}
\includegraphics[width=0.5\linewidth]{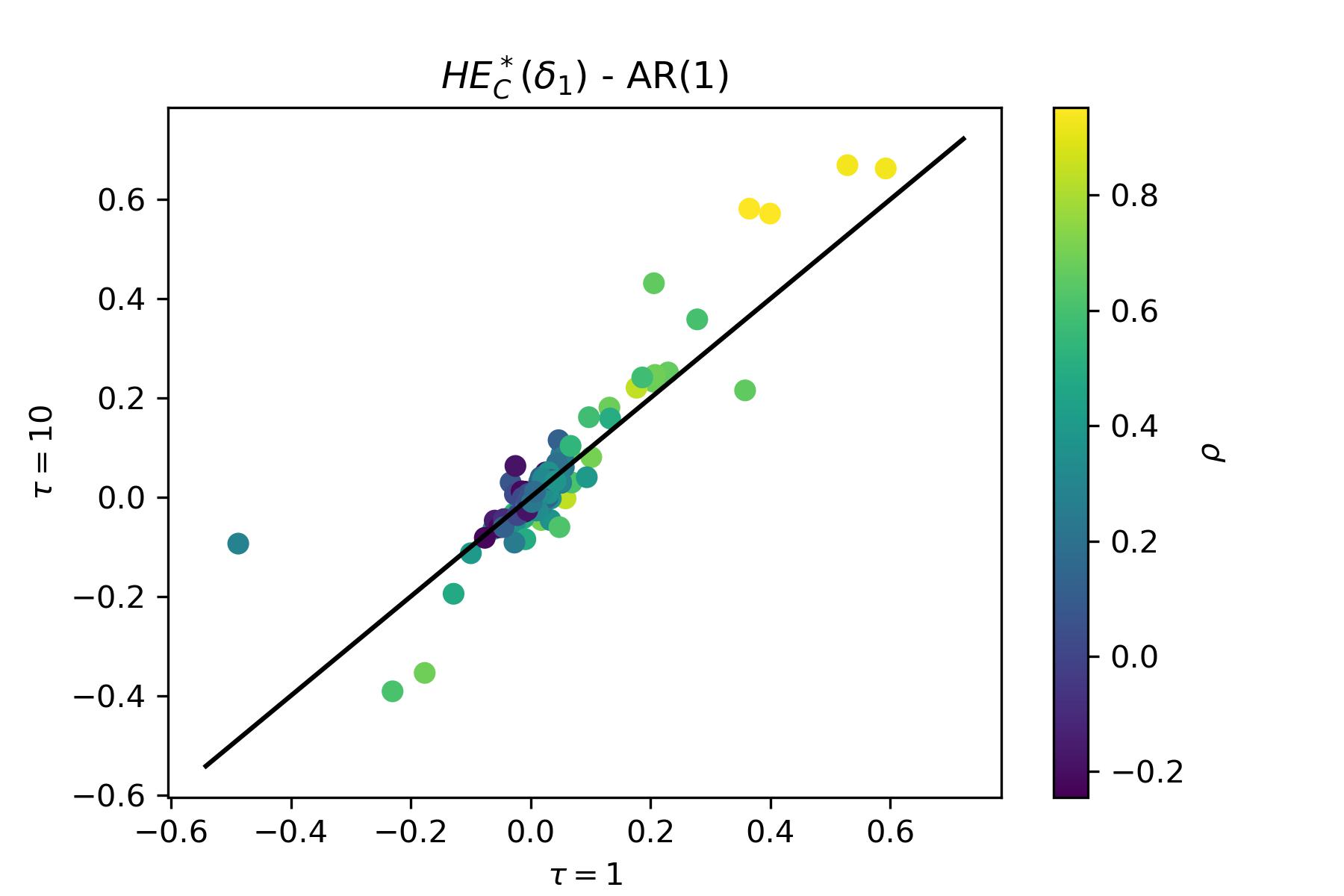} \includegraphics[width=0.5\linewidth]{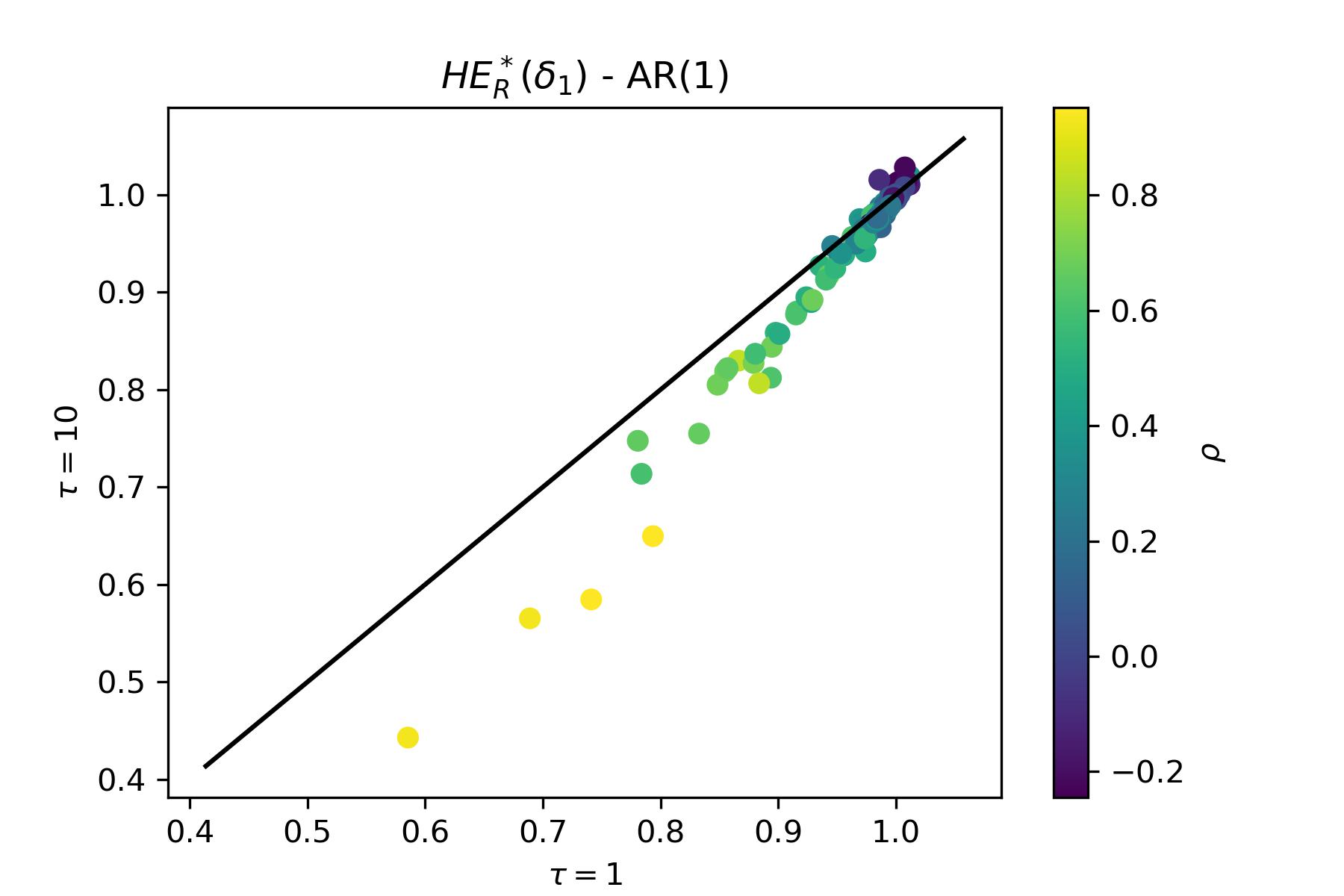}
    \caption{Comparison between the hedge effectiveness metrics of the robust methodology when realized variances and covariances are fitted by AR(1) processes, and two different prediction horizons are considered, i.e., $\tau = 1$ and $\tau = 10$. Each point is associated with a pair of instruments and its color is the return correlation $\rho$ of the pair. The threshold $\delta_1$ is the first quartile of the asset returns. The dark line corresponds to the bisector.}
    \label{fig:HEr_tau1_vs_tau10}
\end{figure}

\subsection{The role of memory}\label{subsec:AR1_vs_AR5}
A natural question is the impact of employing a longer memory in fitting the variances and covariances. 
In Fig. \ref{fig:stdHR_AR1_vs_AR5}, we compare the standard deviation of the hedge ratios obtained using AR(1) and AR(5) models for $\tau=1$ and $\tau = 10$. We observe that a larger order of the autoregressive process leads to higher values of $std(h_t)$ and $std(h_t^*)$, and the effect is stronger as the prediction horizon $\tau$ increases. Therefore, estimating the realized variances and covariances with higher precision would require more frequent portfolio rebalancing, without substantial differences in terms of the effectiveness metrics, as it is represented in Fig. \ref{fig:HEr_AR1_vs_AR5}. The reason at the root of the lower hedge ratio variability with AR(5) compared to AR(1) is precisely the improved forecasting performance that comes with a higher-order AR specification. In Fig. \ref{fig:DeltaIRF_AR1_vs_AR5}, we plot the following quantity for two instrument pairs in our dataset\footnote{Similar results are obtained for other pairs.}
\begin{equation}\label{eq_def_DeltaIRF}
    \Delta_{\psi_h} = \frac{\sigma_{SF, \infty} + \psi_h^{SF}}{\sigma^2_{F, \infty} + \psi_h^{F}} - \frac{\sigma_{SF, \infty}}{\sigma^2_{F, \infty}}
\end{equation}
where $\sigma_{\cdot, \infty}$ are the equilibrium values of the AR processes and $\psi_h$ the impulse response function. The metric $\Delta_{\psi_h} $ represents how the hedge ratio dynamically fluctuates around its equilibrium value in response to a shock. We observe that, for the AR(5), the response is more persistent, exhibits longer memory, and displays more bumps compared to the AR(1). This leads to higher variability of the covariance and variance predictions, and, as a consequence, of the hedge ratio. 

Finally, as a robustness check, in \ref{app:HAR}, we report the results that we obtain by describing realized variances and covariances with an HAR-type model \citep{corsi}. They are consistent with the findings outlined above.

\begin{figure}
    \includegraphics[width=0.5\linewidth]{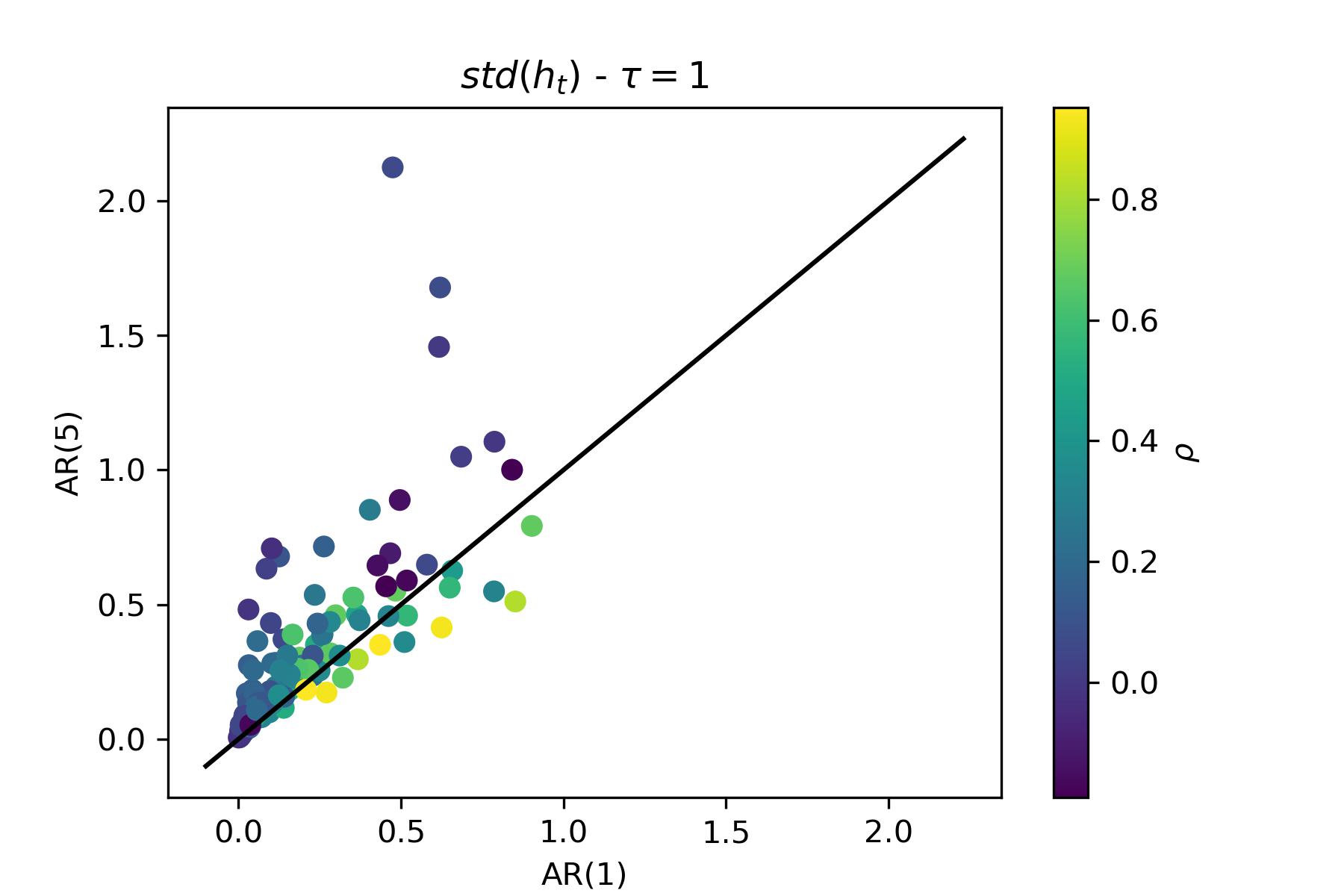}
    \includegraphics[width=0.5\linewidth]{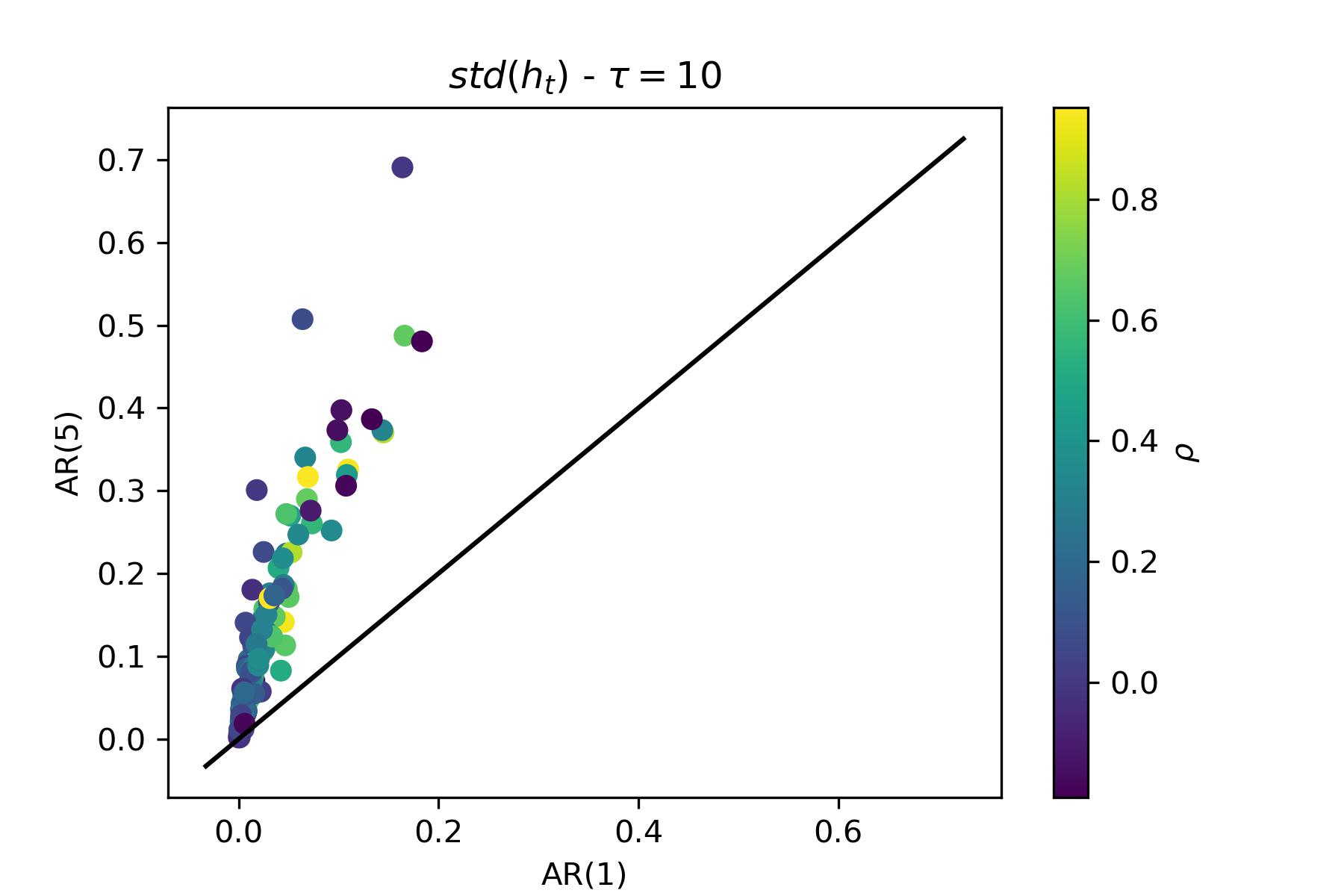}
    \includegraphics[width=0.5\linewidth]{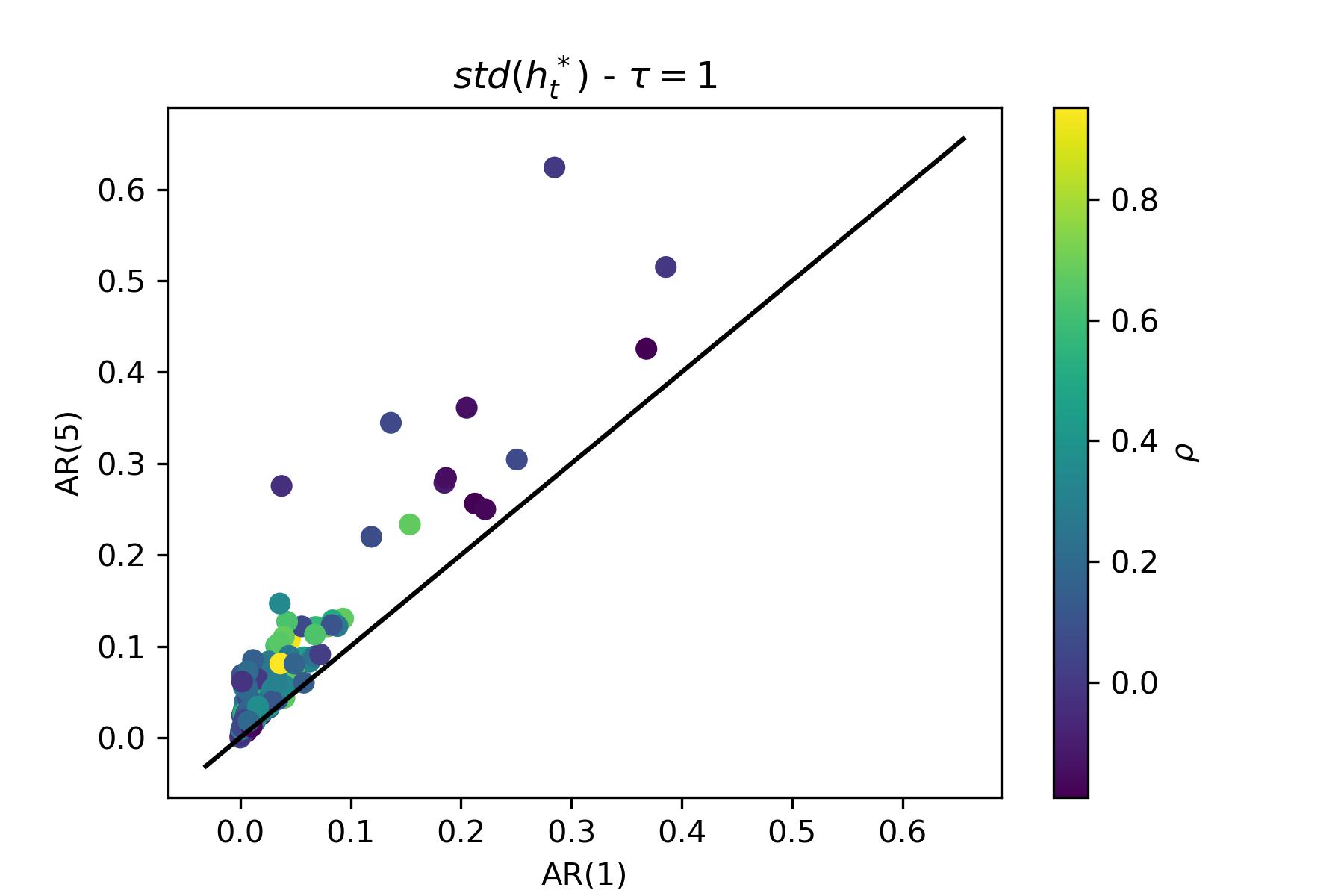}
    \includegraphics[width=0.5\linewidth]{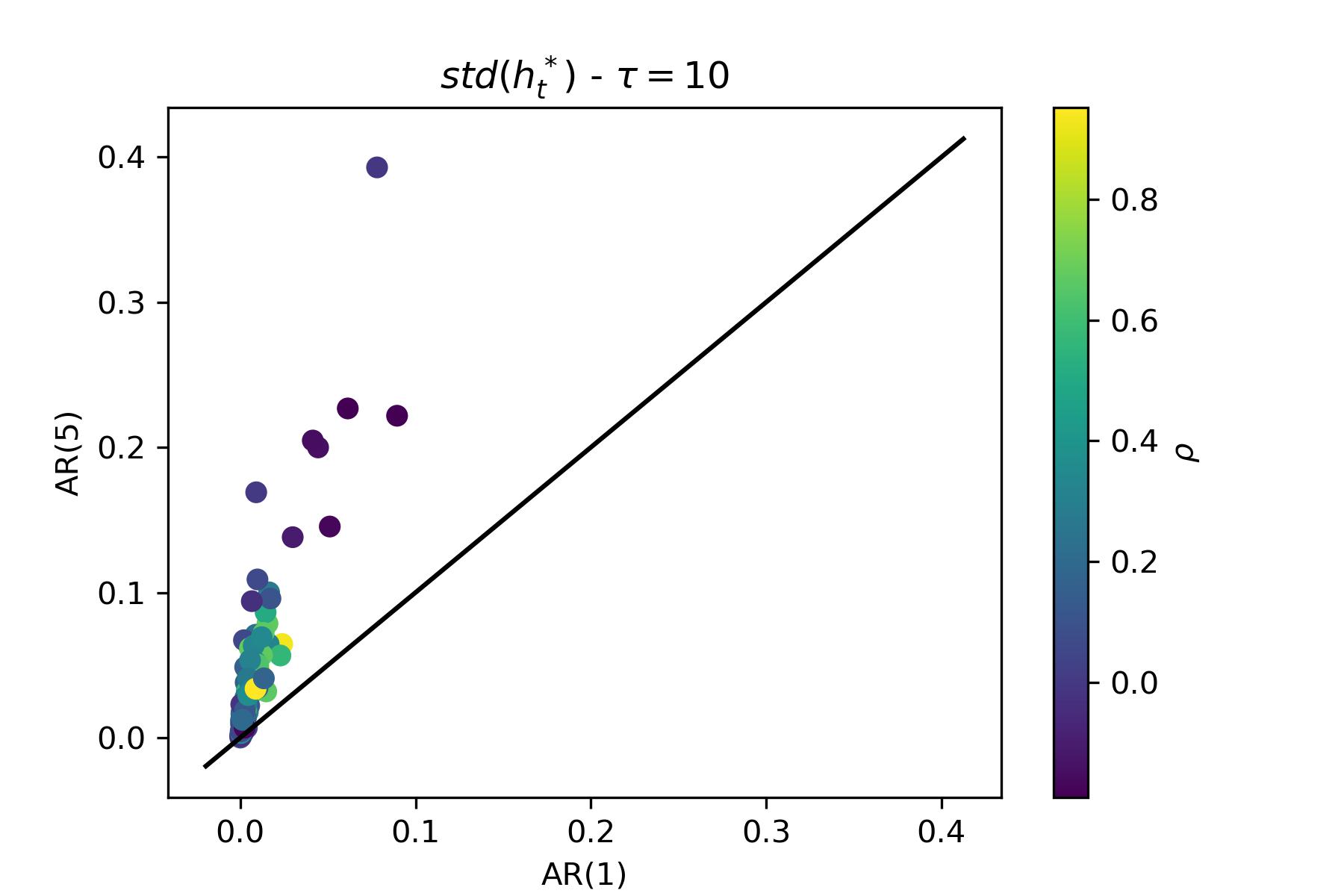}
    \caption{Comparison between the standard deviations of the hedge ratios that are obtained when realized variances and covariances are fitted by AR(1) and AR(5) processes. Two different prediction horizons are considered, i.e., $\tau = 1$ and $\tau = 10$. Each point is associated with a pair of instruments and its color is the return correlation $\rho$ of the pair. The dark line corresponds to the bisector.}
    \label{fig:stdHR_AR1_vs_AR5}
\end{figure}

\begin{figure}
\includegraphics[width=0.5\linewidth]{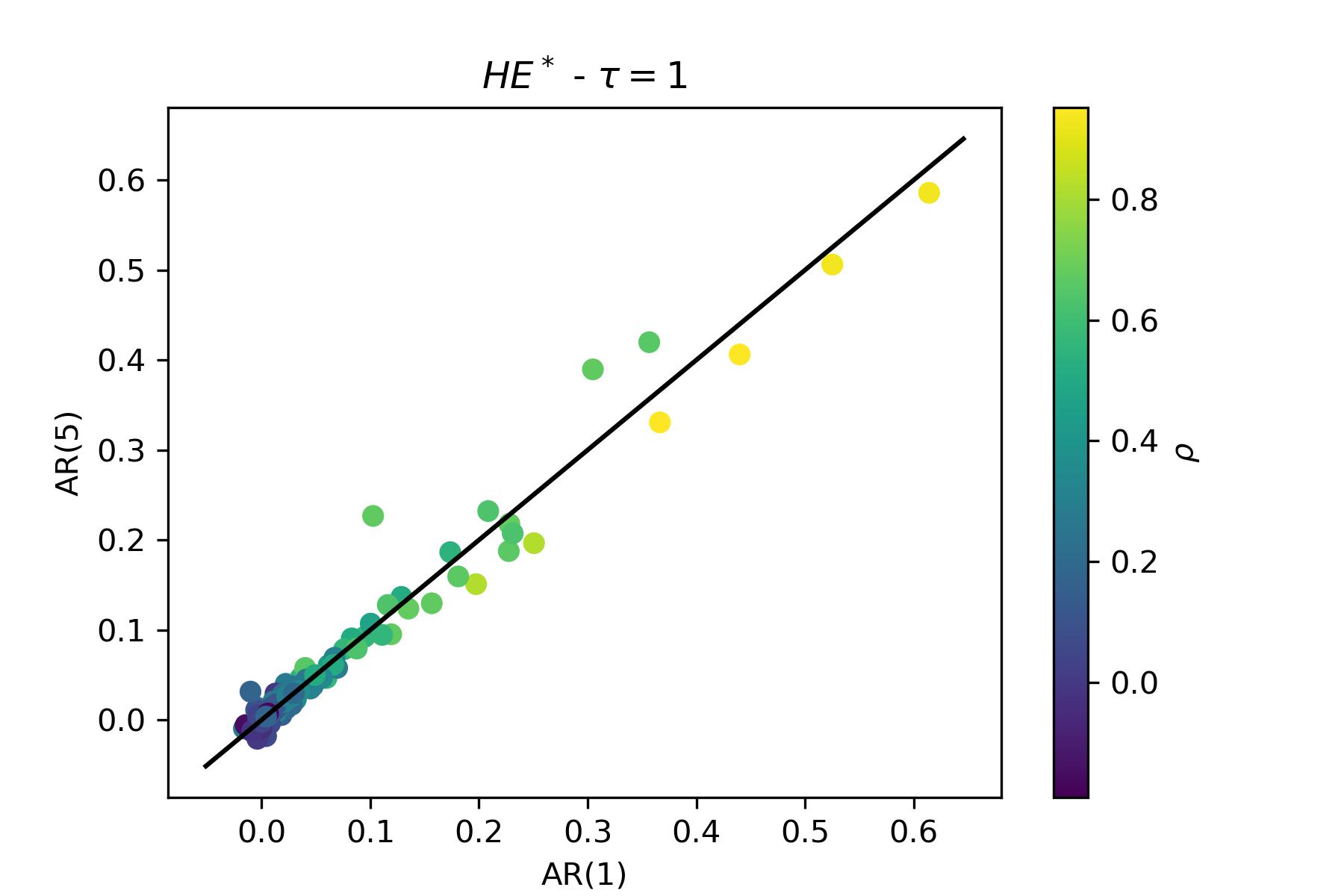}
\includegraphics[width=0.5\linewidth]{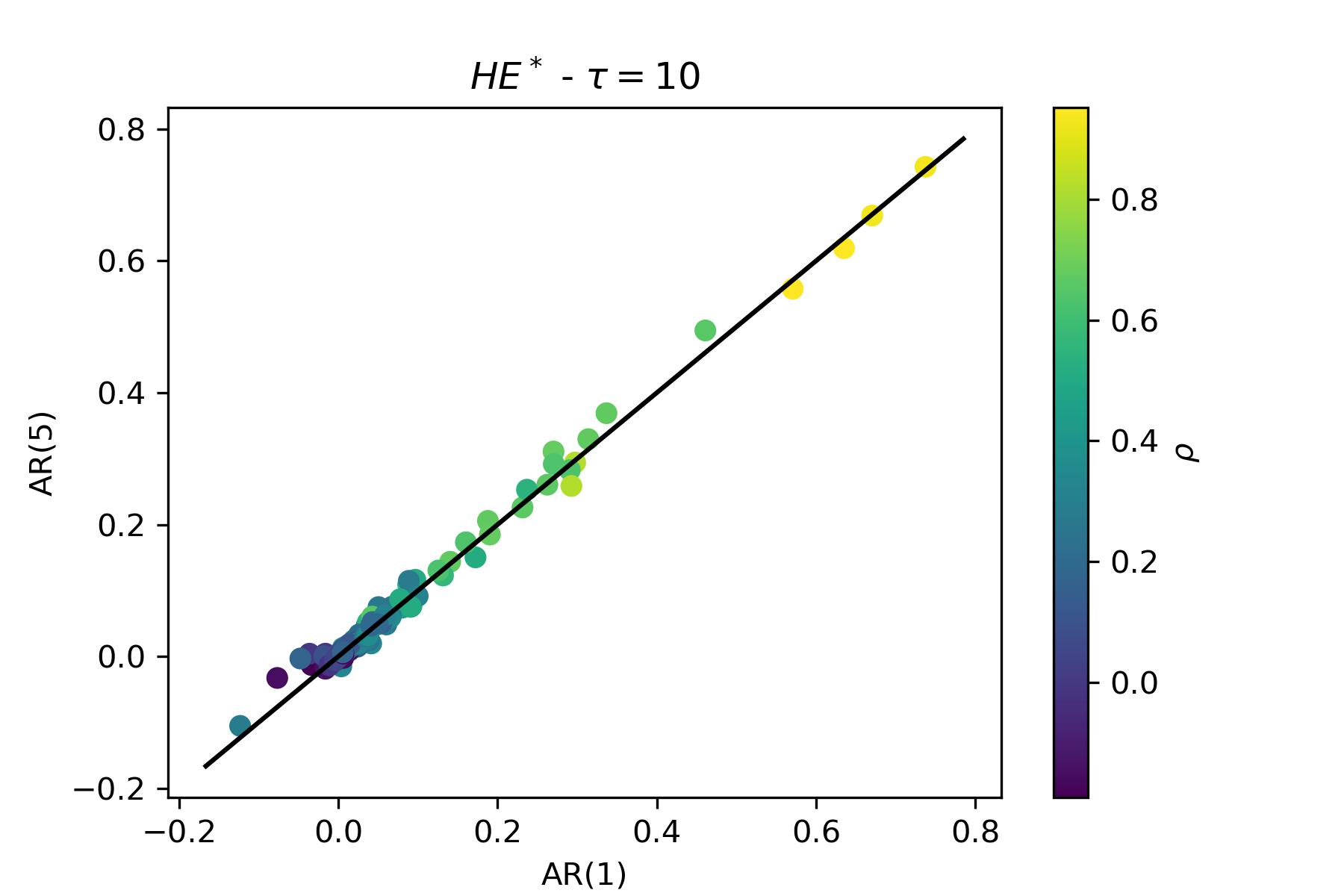}
\includegraphics[width=0.5\linewidth]{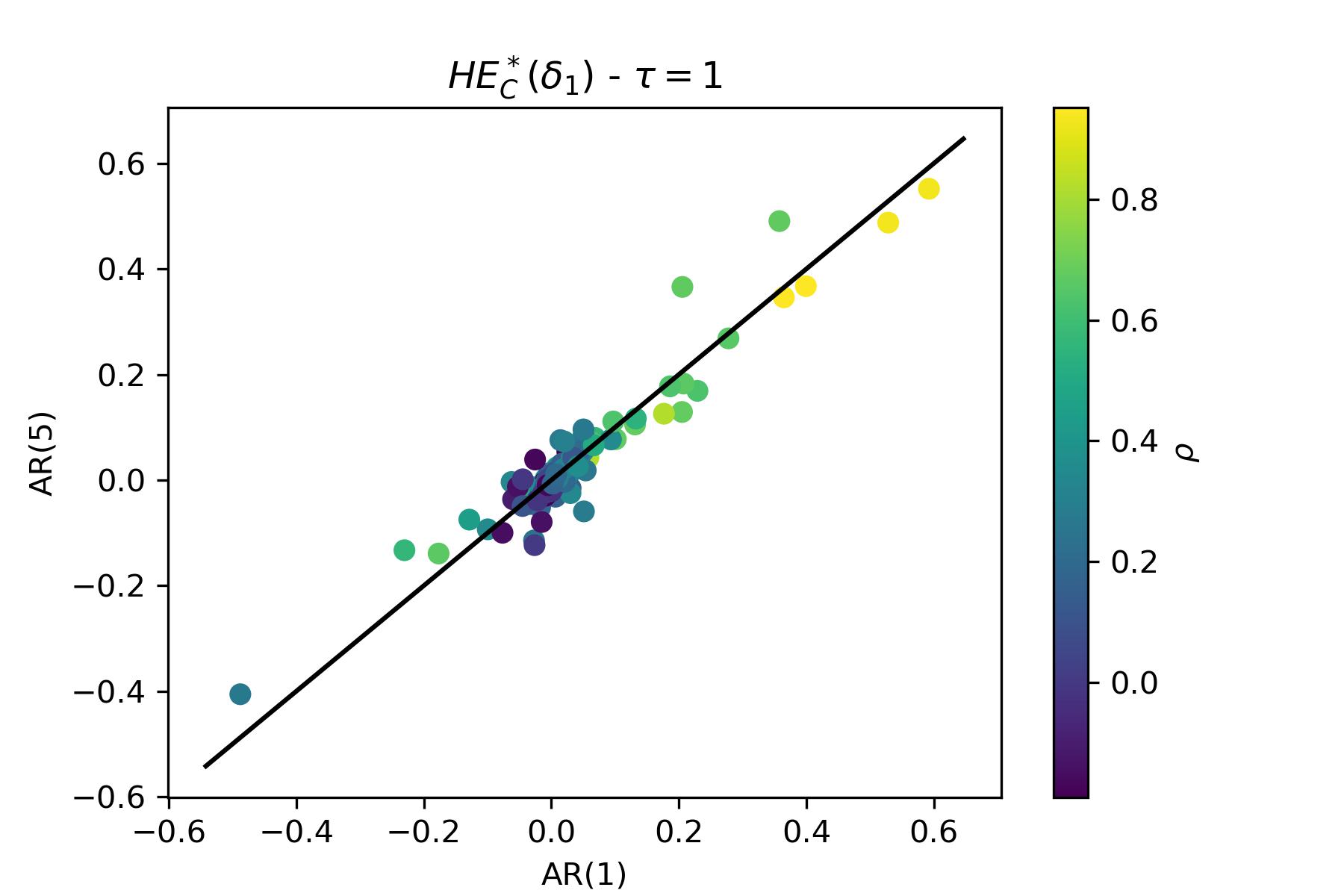}
\includegraphics[width=0.5\linewidth]{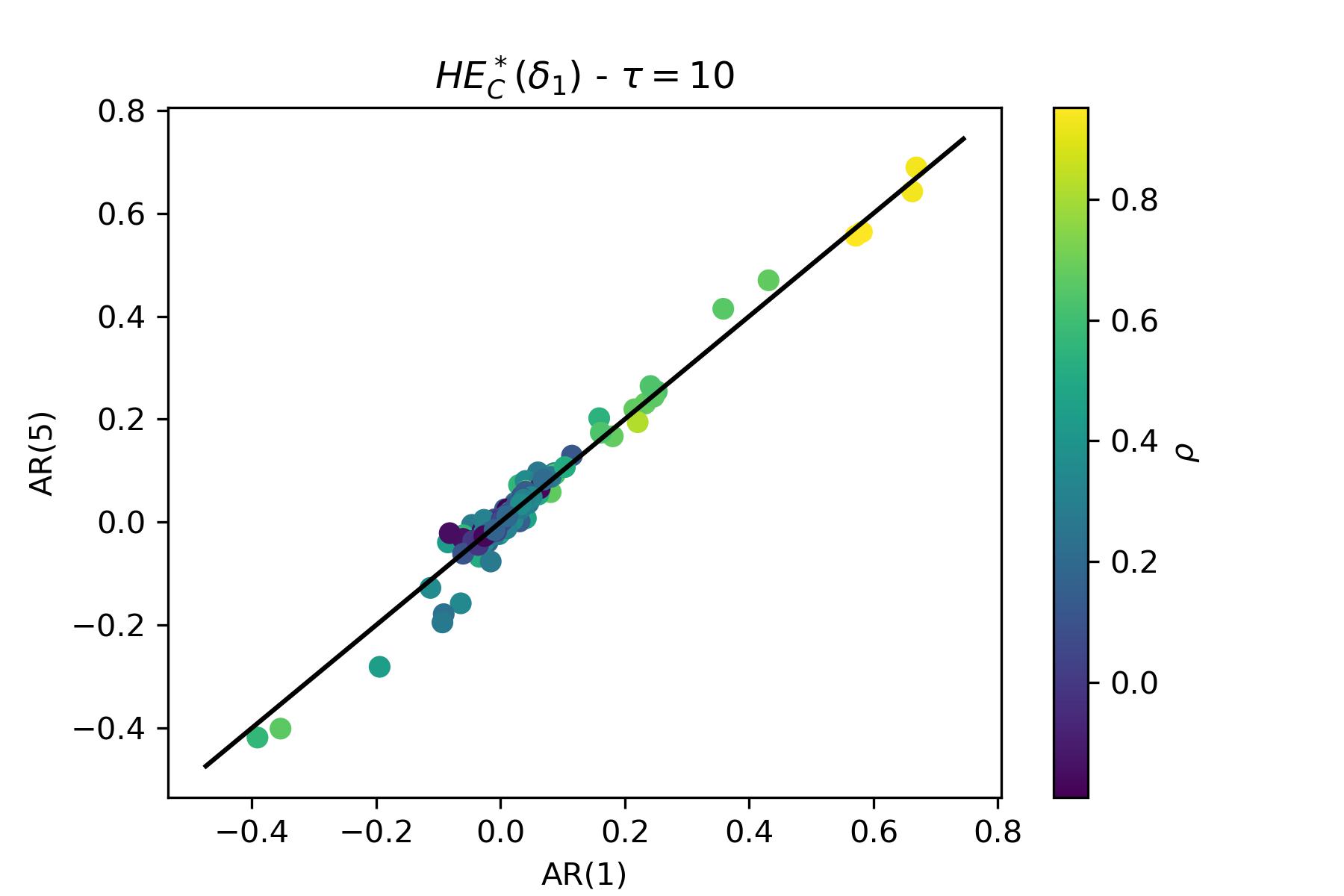}
\includegraphics[width=0.5\linewidth]{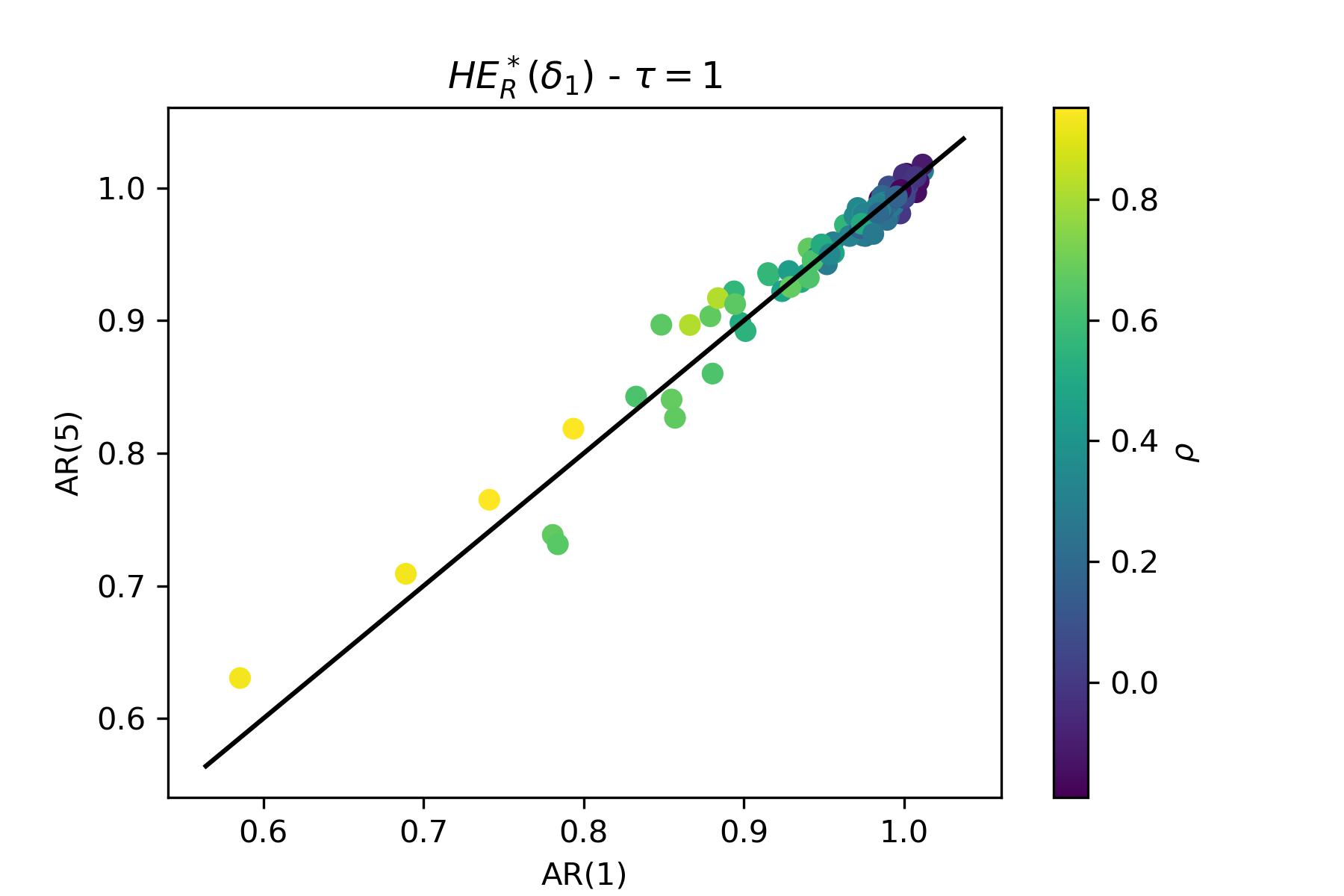}
\includegraphics[width=0.5\linewidth]{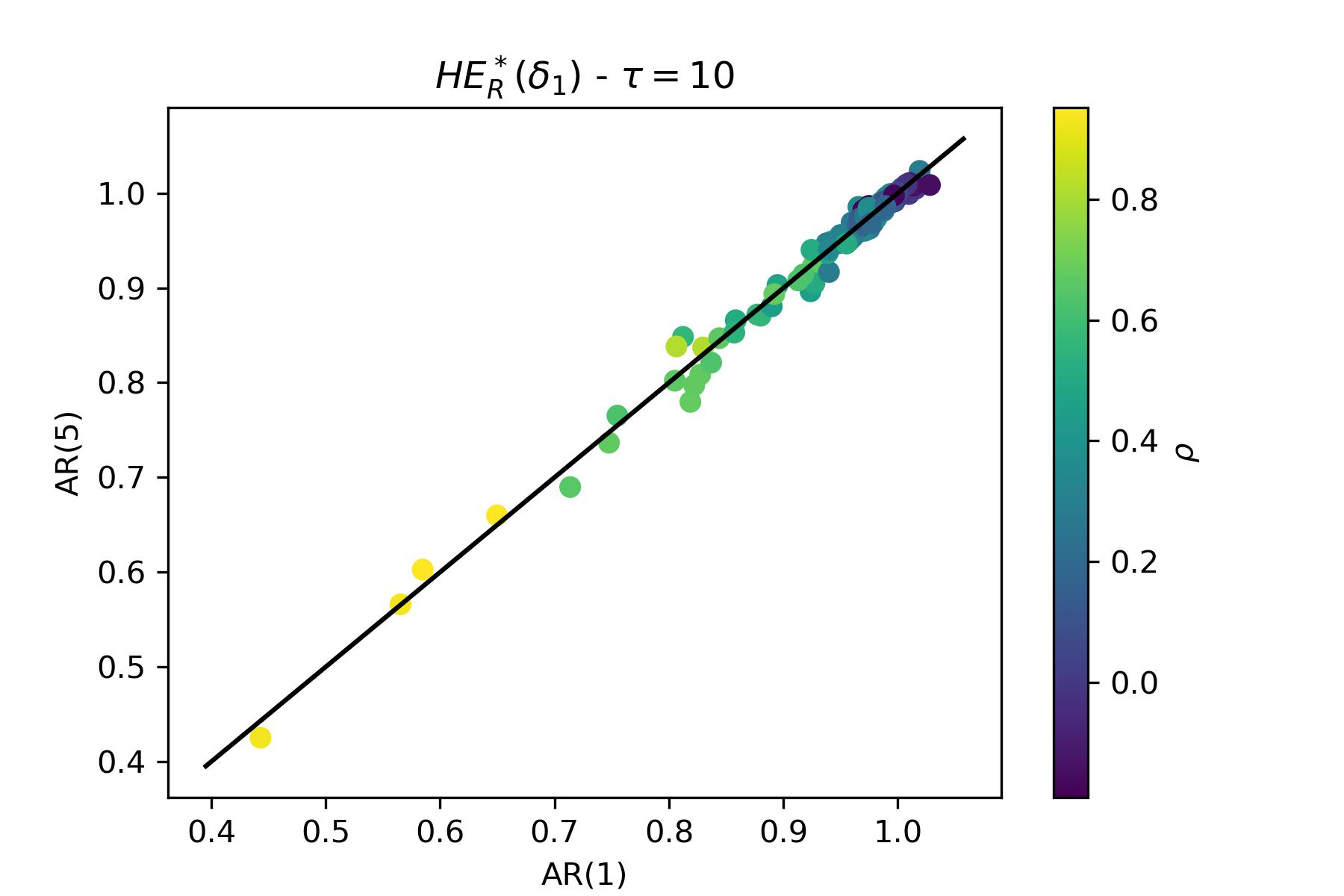}
    \caption{Comparison between the hedge effectiveness metrics of the robust methodology when realized variances and covariances are fitted by AR(1) and AR(5) processes, and two different prediction horizons are considered, i.e., $\tau = 1$ and $\tau = 10$. Each point is associated with a pair of instruments and its color is the return correlation $\rho$ of the pair. The threshold $\delta_1$ is the first quartile of the asset returns. The dark line corresponds to the bisector.}
    \label{fig:HEr_AR1_vs_AR5}
\end{figure}

\begin{figure}
\includegraphics[width=0.5\linewidth]{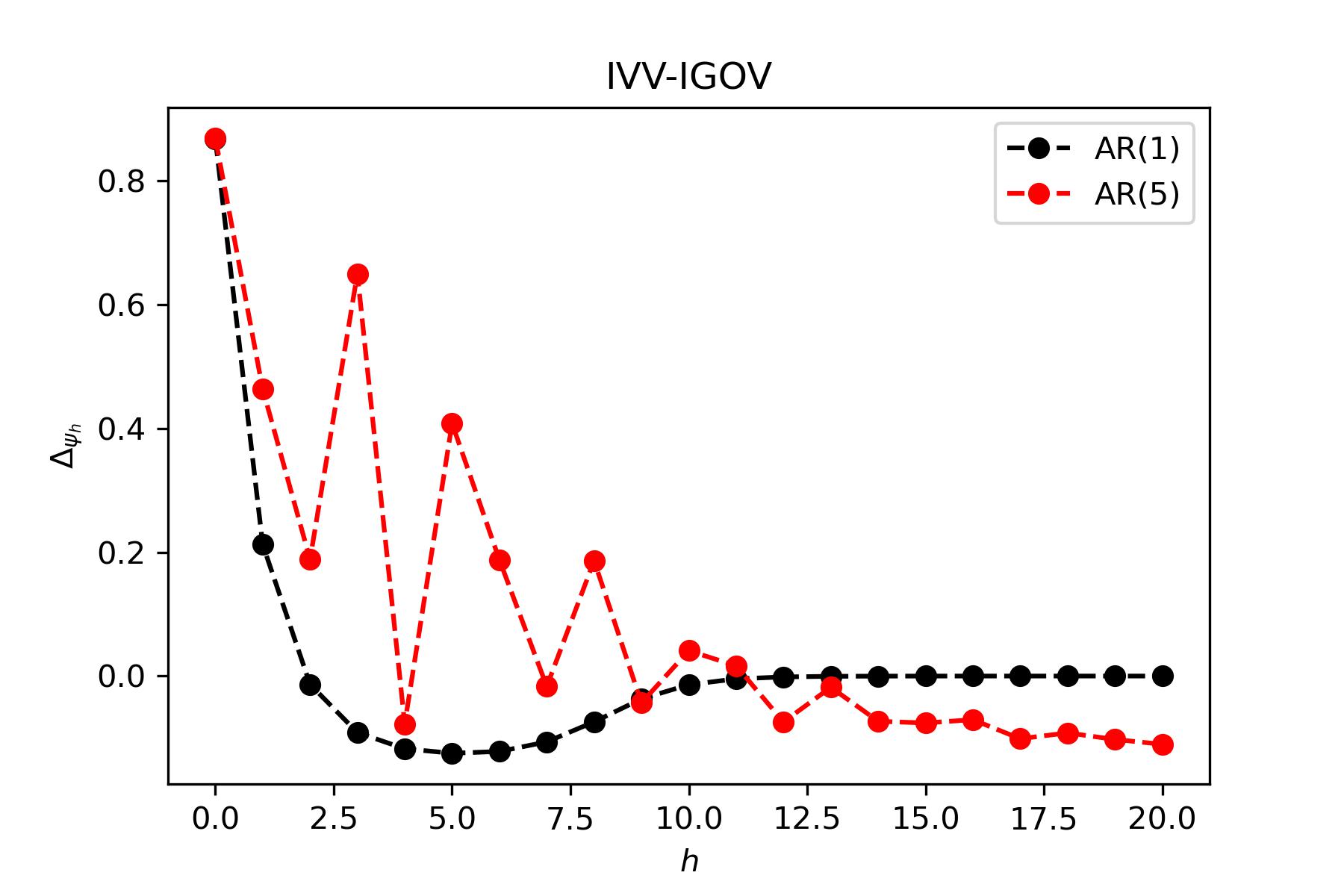}
\includegraphics[width=0.5\linewidth]{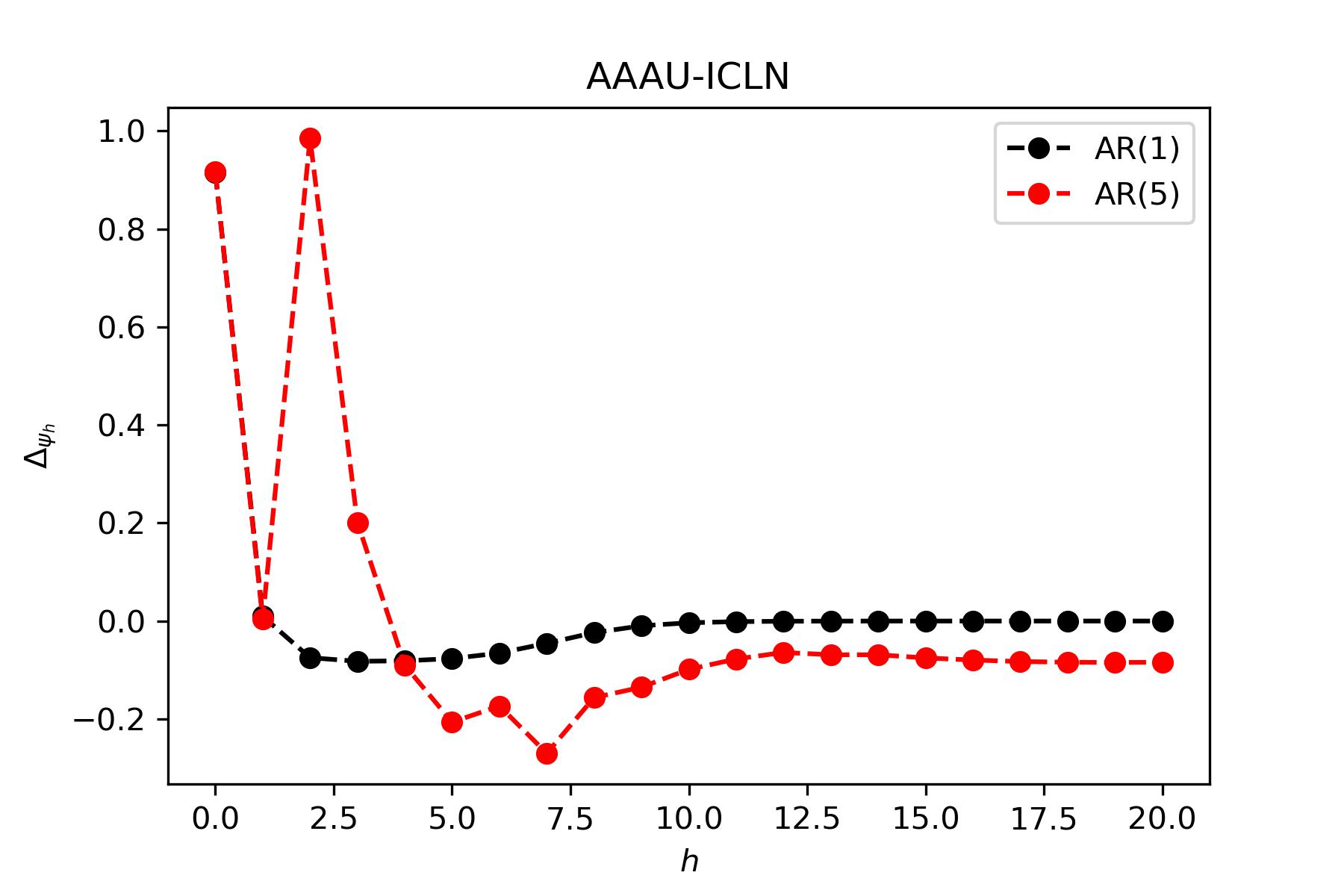}
    \caption{Quantity defined in Eq. \eqref{eq_def_DeltaIRF} as a function of the forecast horizon. Two different instrument pairs are considered: in the title the first label refers to the asset and the second to the hedging instrument.}
    \label{fig:DeltaIRF_AR1_vs_AR5}
\end{figure}

\subsection{Performance and risk-adjusted metrics}
\label{performance}

We study here the role of robustness in terms of standard financial performance and risk-adjusted metrics, with a focus on transaction costs. 
Transaction costs are defined as

\begin{equation}\label{eq_transaction_costs}
TC_t = \lvert \Delta h_t \rvert \times \text{bp},
\end{equation}

where $\Delta h_t = h_t - h_{t-1}$ represents the daily variation of either the robust or the standard hedge ratio, and \textit{bp} denotes the transaction cost per unit of hedge ratio turnover in basis points. Following previous literature (see, e.g., \citealt{avellaneda2010statistical,fischer2018deep,flori2021revealing}), we consider a unit transaction cost of 5 or 10 basis points.

Specifically, we take into account the following portfolio metrics. The profit and loss (P\&L) measures the total return or loss generated by the strategy; the Sharpe ratio (SR) is a measure of the risk-adjusted return of an investment, defined as the excess return per unit of risk \citep{sharpe1998sharpe}; the Omega ratio ($\Omega$) compares probability-weighted gains and losses relative to a given benchmark, offering an integrated perspective on the overall portfolio performance \citep{keating2002universal}; the maximum drawdown (DD) captures the largest observed decline from a portfolio’s peak to its subsequent trough; the Value at Risk (VaR) represents the loss threshold not expected to be exceeded with 95\% confidence over the specified horizon \citep{linsmeier2000value}; the Expected Shortfall (ES) quantifies the average loss conditional on exceeding this threshold, thus providing a more comprehensive assessment of tail risk \citep{acerbi2002coherence}.

We summarize these performances over the full sample period in Fig.~\ref{fig: transaction_costs}-\ref{fig: transaction_costs2}. The scatter plots compare the outcomes obtained using standard versus robust hedge ratios, indicating also the asset class of each hedged asset. In addition, Tables \ref{tab: transaction_costs_aaau_no_return}–\ref{tab:transaction_costs_ung} in \ref{app: performance} provide a detailed breakdown of the scatter-plot results by reporting, for each hedged asset, the difference between the metrics based on robust and standard hedge ratios.

\begin{figure}[ht!]
\centering
\raggedright \hspace{0.05cm} {\scriptsize A}  \hspace{3.75cm} {\scriptsize B} \hspace{5.75cm} {\scriptsize C}\\
{\includegraphics[width=.3\textwidth,height=3.5cm]{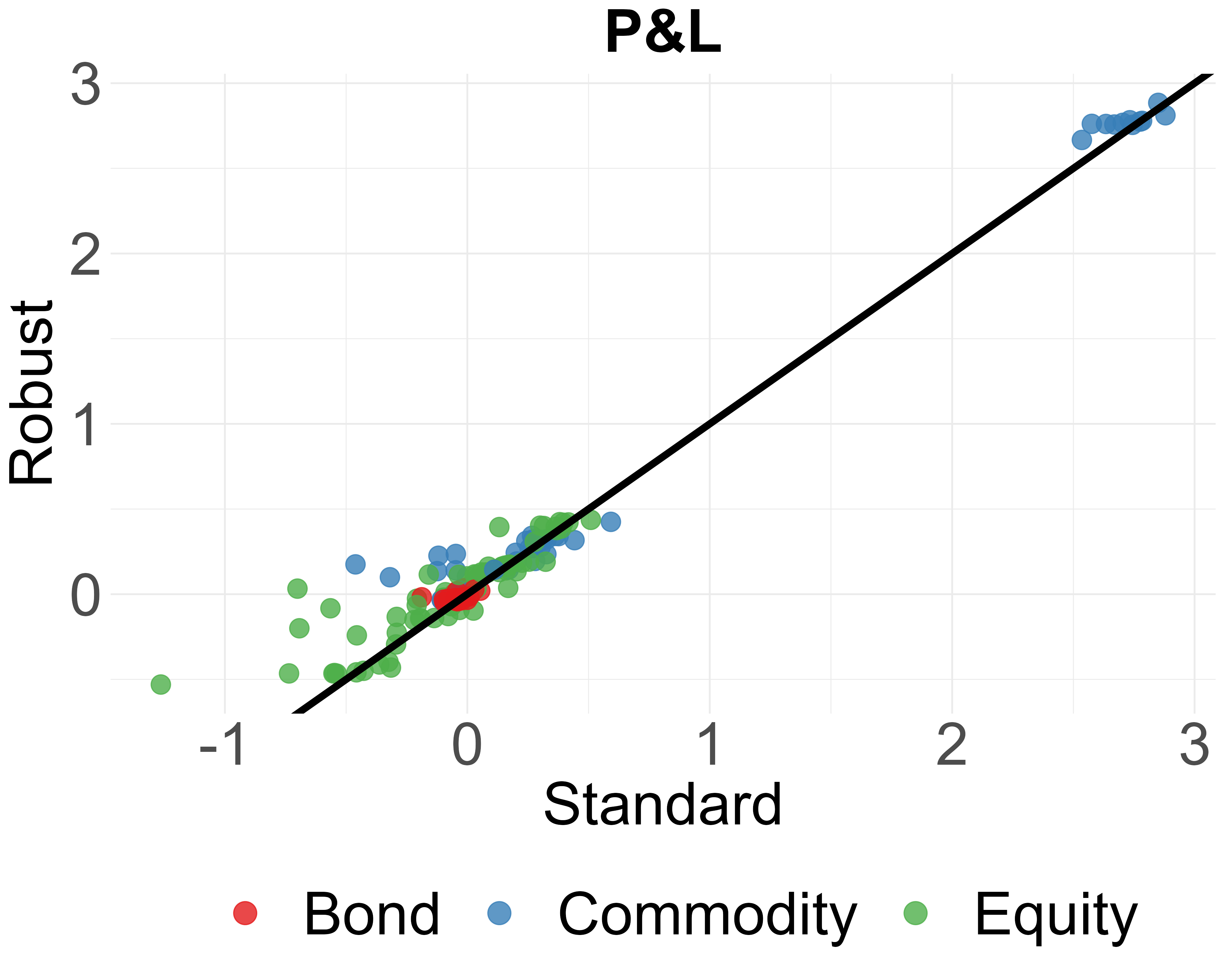}} \quad
{\includegraphics[width=.3\textwidth,height=3.5cm]{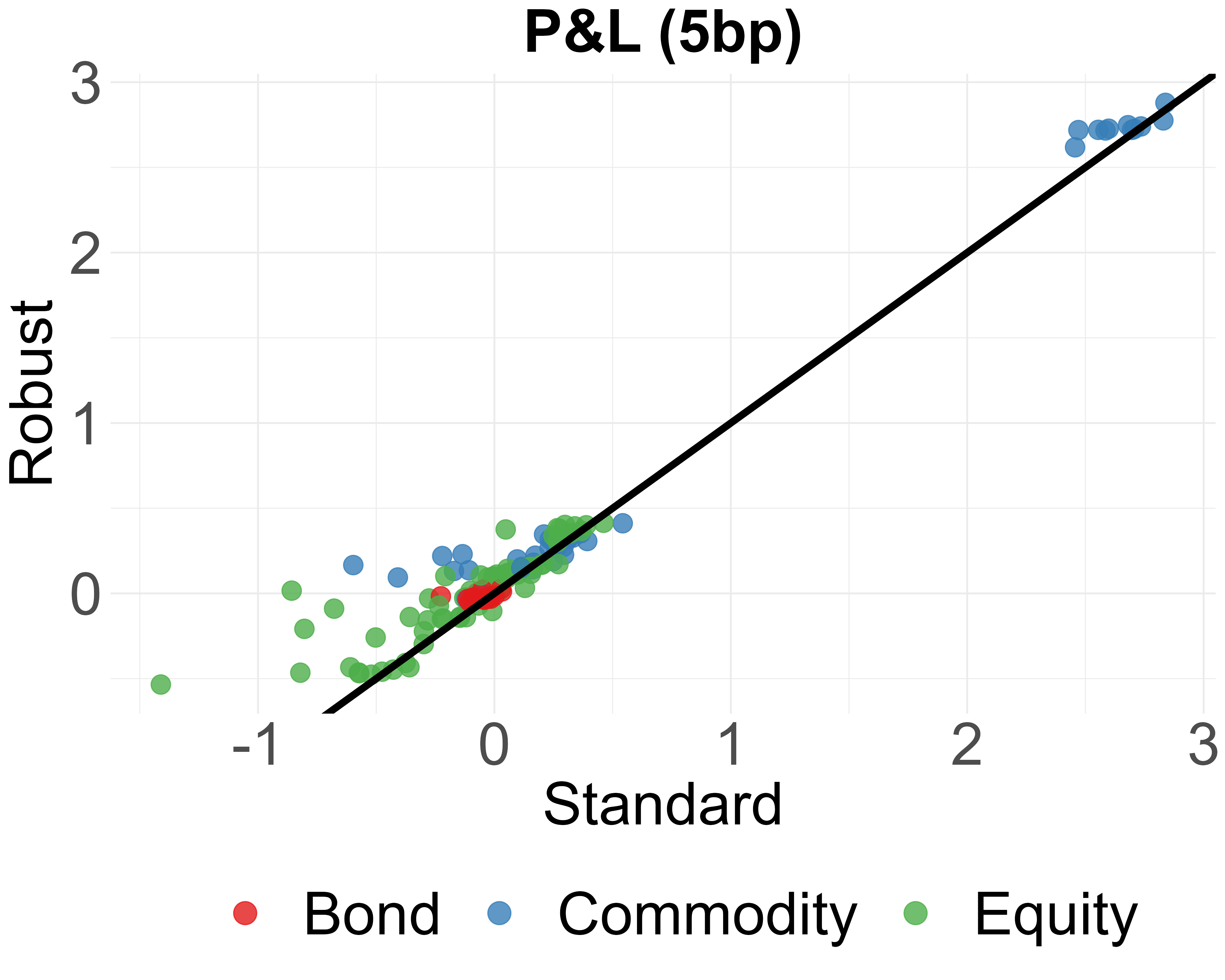}} \quad
{\includegraphics[width=.3\textwidth,height=3.5cm]{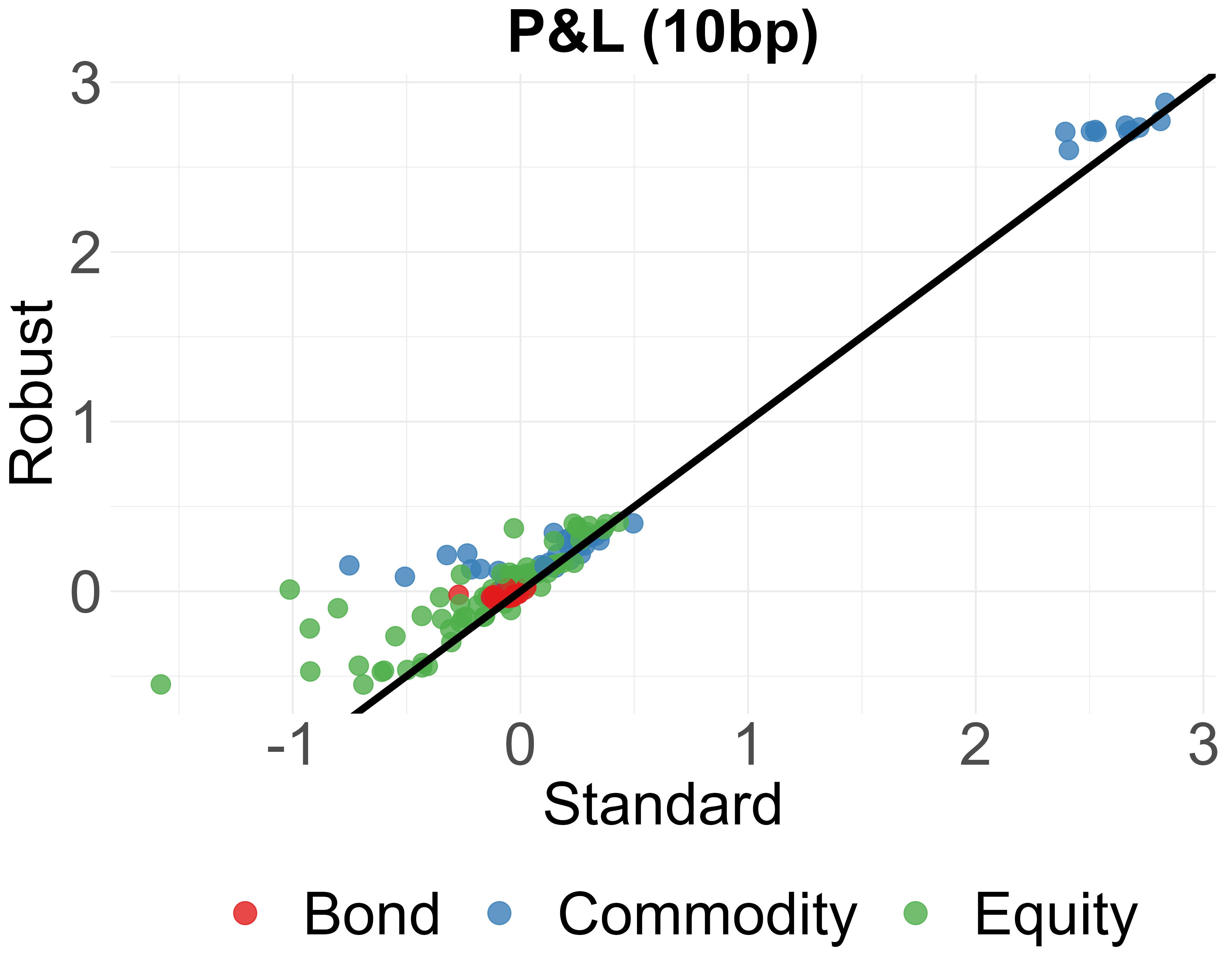}} \quad
\raggedright \hspace{0.05cm} {\scriptsize D}  \hspace{3.75cm} {\scriptsize E} \hspace{5.75cm} {\scriptsize F}\\
{\includegraphics[width=.3\textwidth,height=3.5cm]{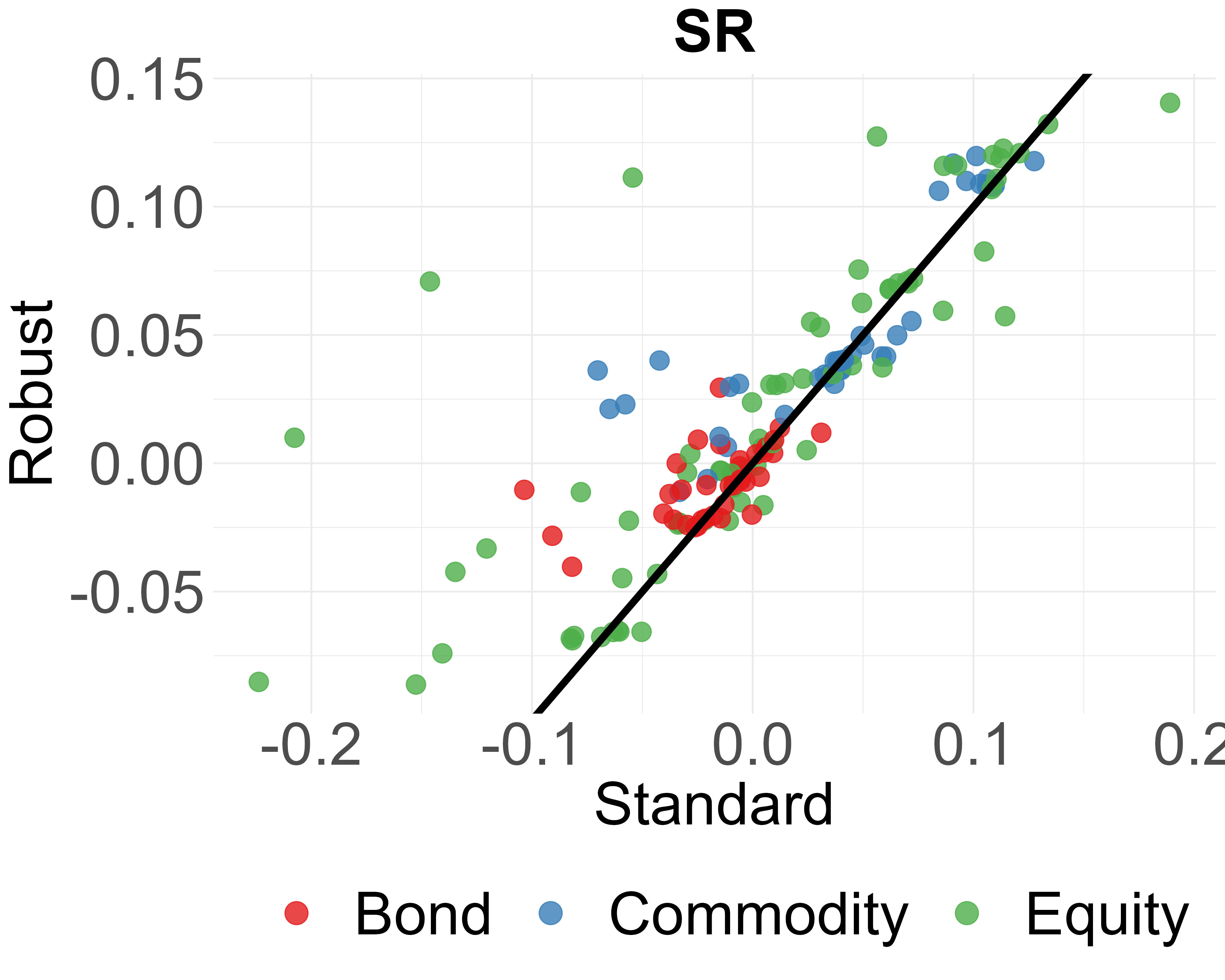}} \quad
{\includegraphics[width=.3\textwidth,height=3.5cm]{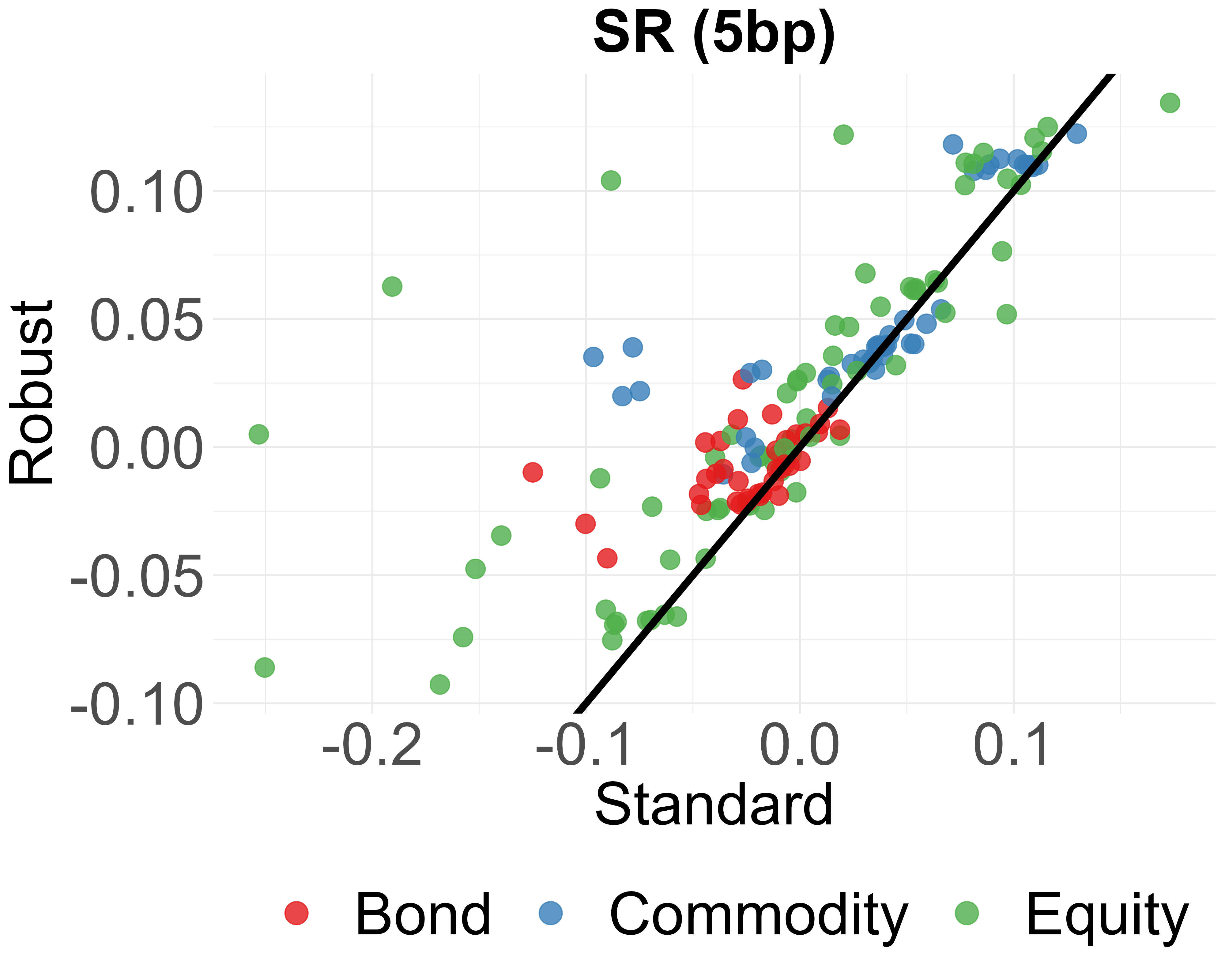}} \quad
{\includegraphics[width=.3\textwidth,height=3.5cm]{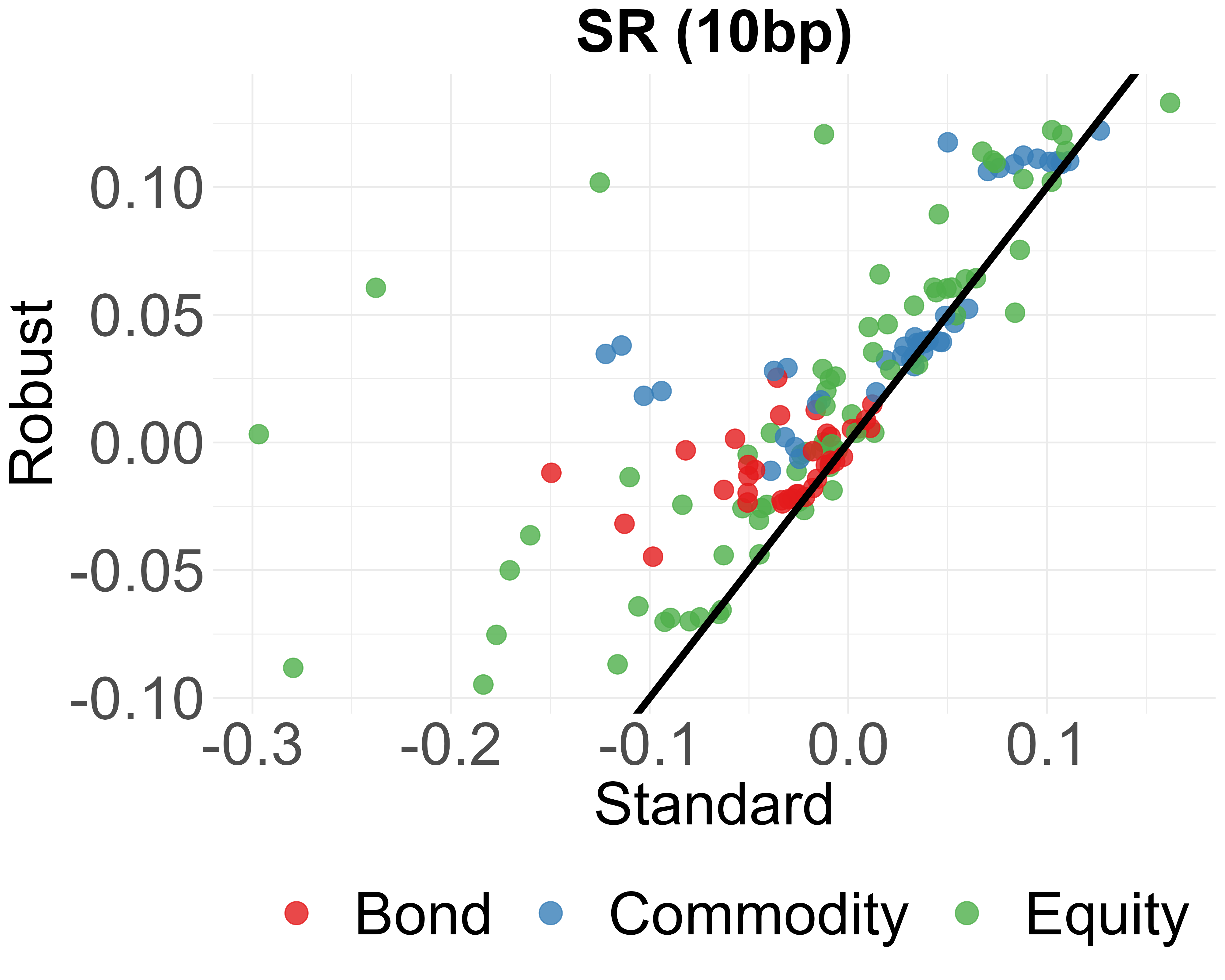}} \quad
\raggedright \hspace{0.05cm} {\scriptsize G}  \hspace{3.75cm} {\scriptsize H} \hspace{5.75cm} {\scriptsize I}\\
{\includegraphics[width=.3\textwidth,height=3.5cm]{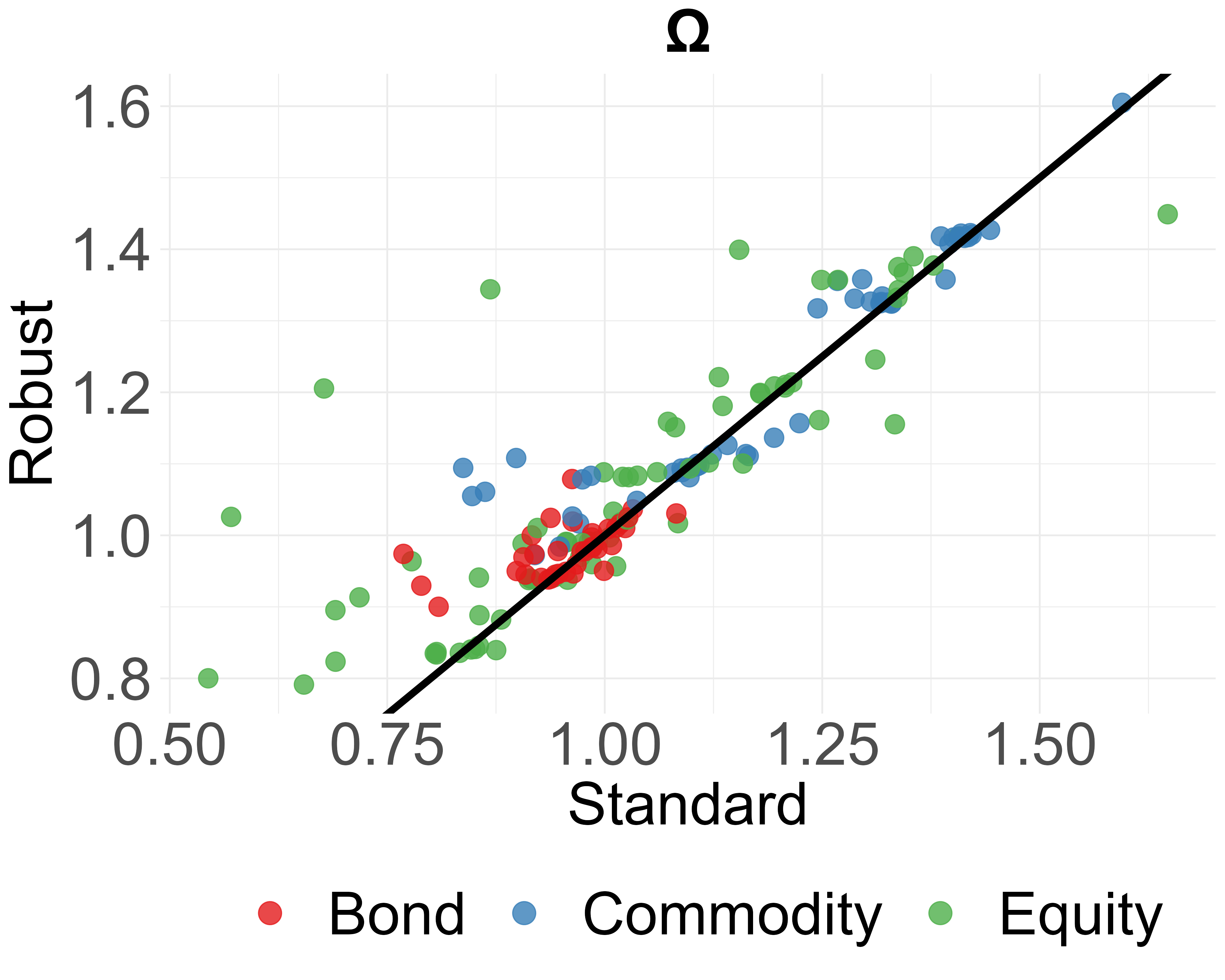}} \quad
{\includegraphics[width=.3\textwidth,height=3.5cm]{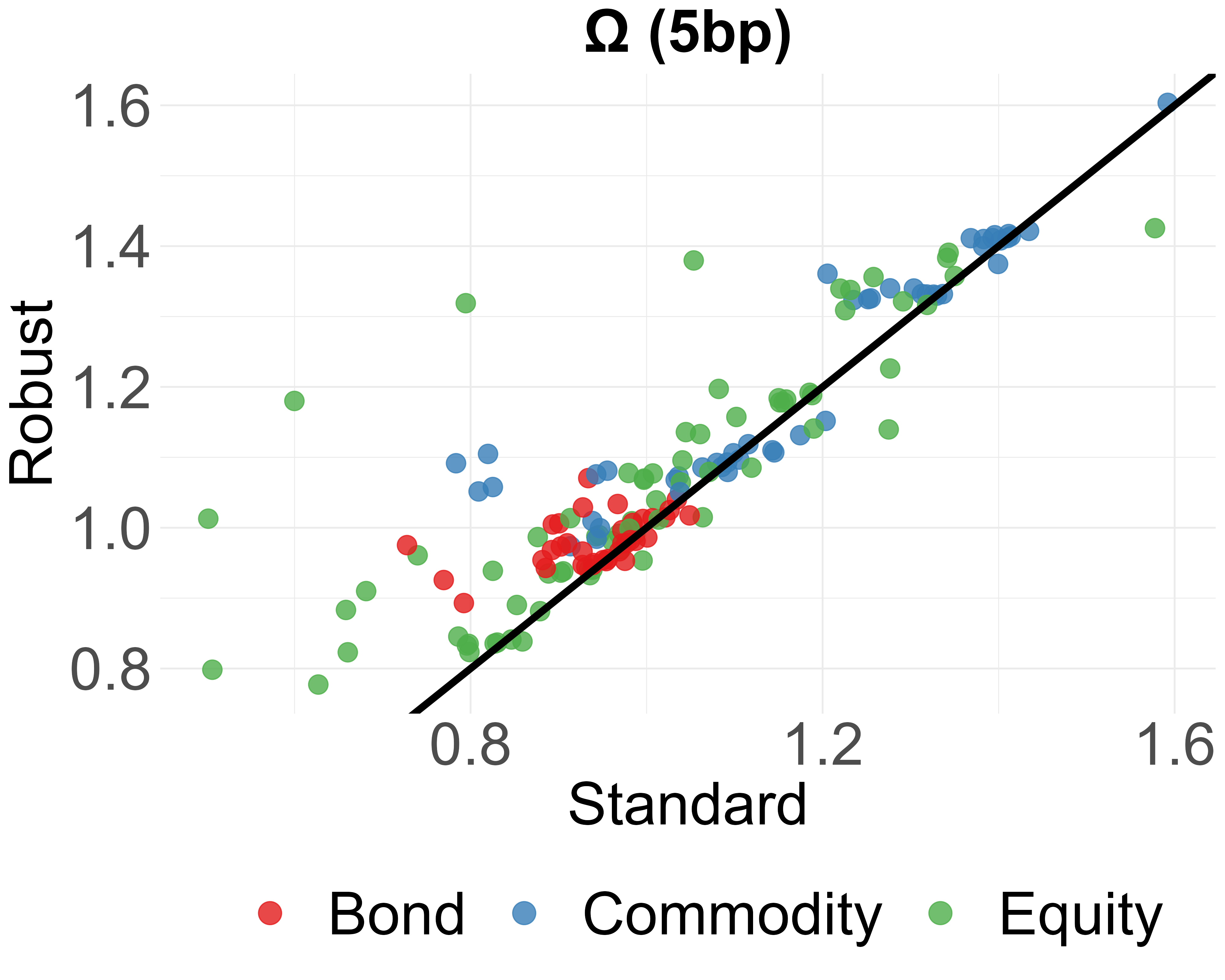}} \quad
{\includegraphics[width=.3\textwidth,height=3.5cm]{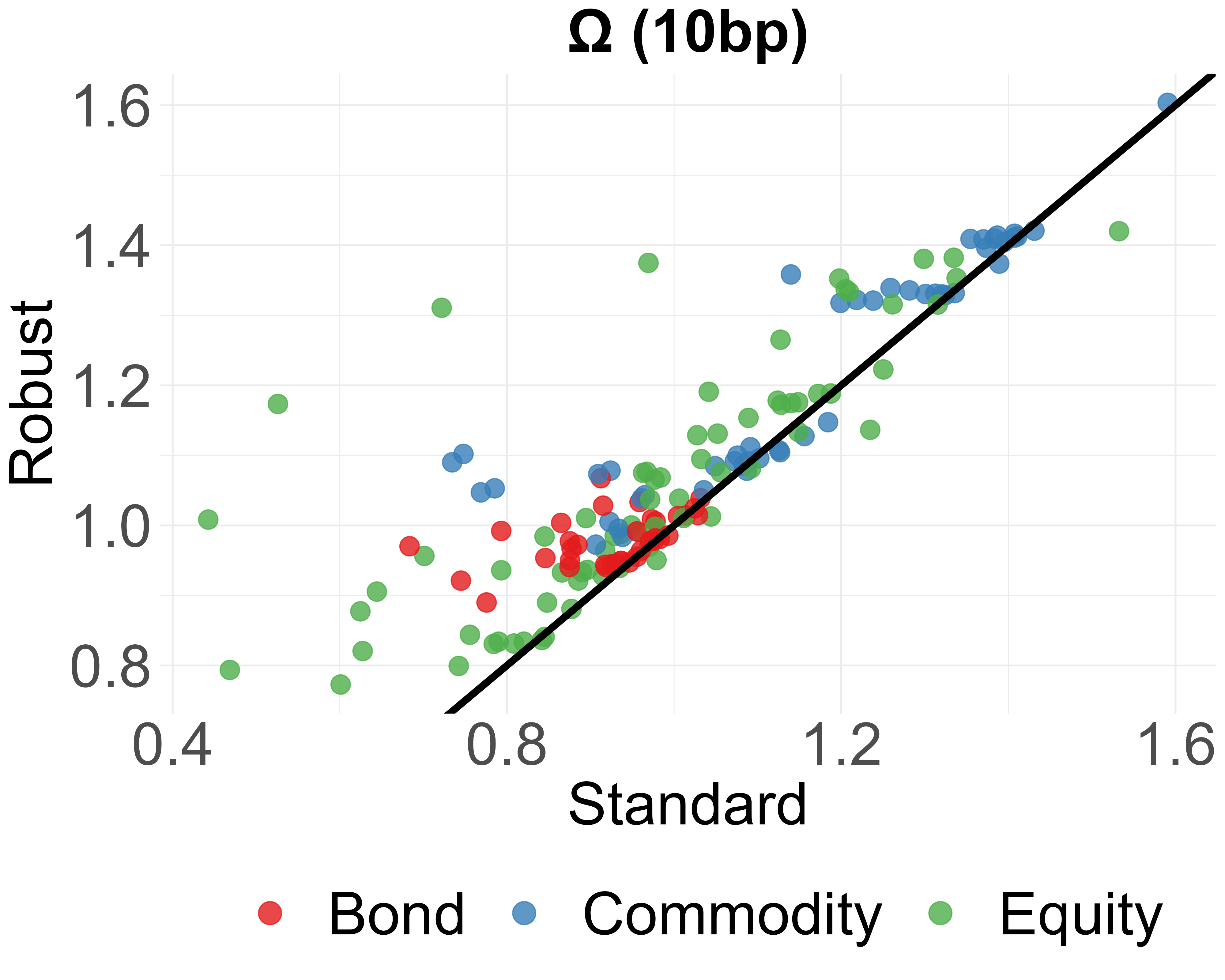}} \quad
\quad
\caption{Scatter plot of various performance metrics obtained using the standard Hedge Ratios and the robust Hedge Ratios. P\&L, SR, and $\Omega$ correspond to profit and loss, Sharpe ratio, and Omega ratio, respectively. Transaction costs are indicated in curly brackets. The color represents the asset class of the hedged asset.}
    \label{fig: transaction_costs}
\end{figure}

\begin{figure}[ht!]
\centering
\raggedright \hspace{0.05cm} {\scriptsize A}  \hspace{3.75cm} {\scriptsize B} \hspace{5.75cm} {\scriptsize C}\\
{\includegraphics[width=.3\textwidth,height=3.5cm]{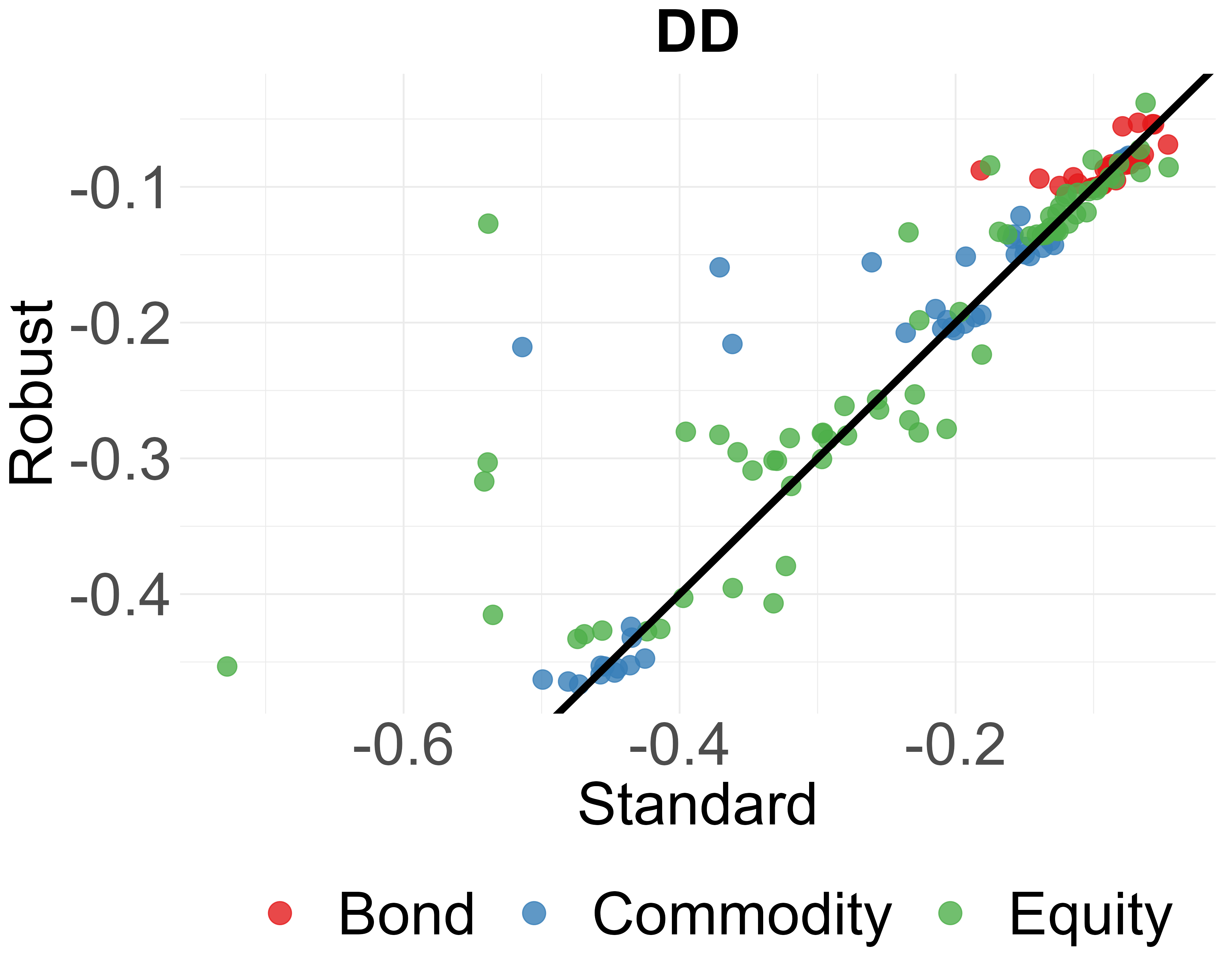}} \quad
{\includegraphics[width=.3\textwidth,height=3.5cm]{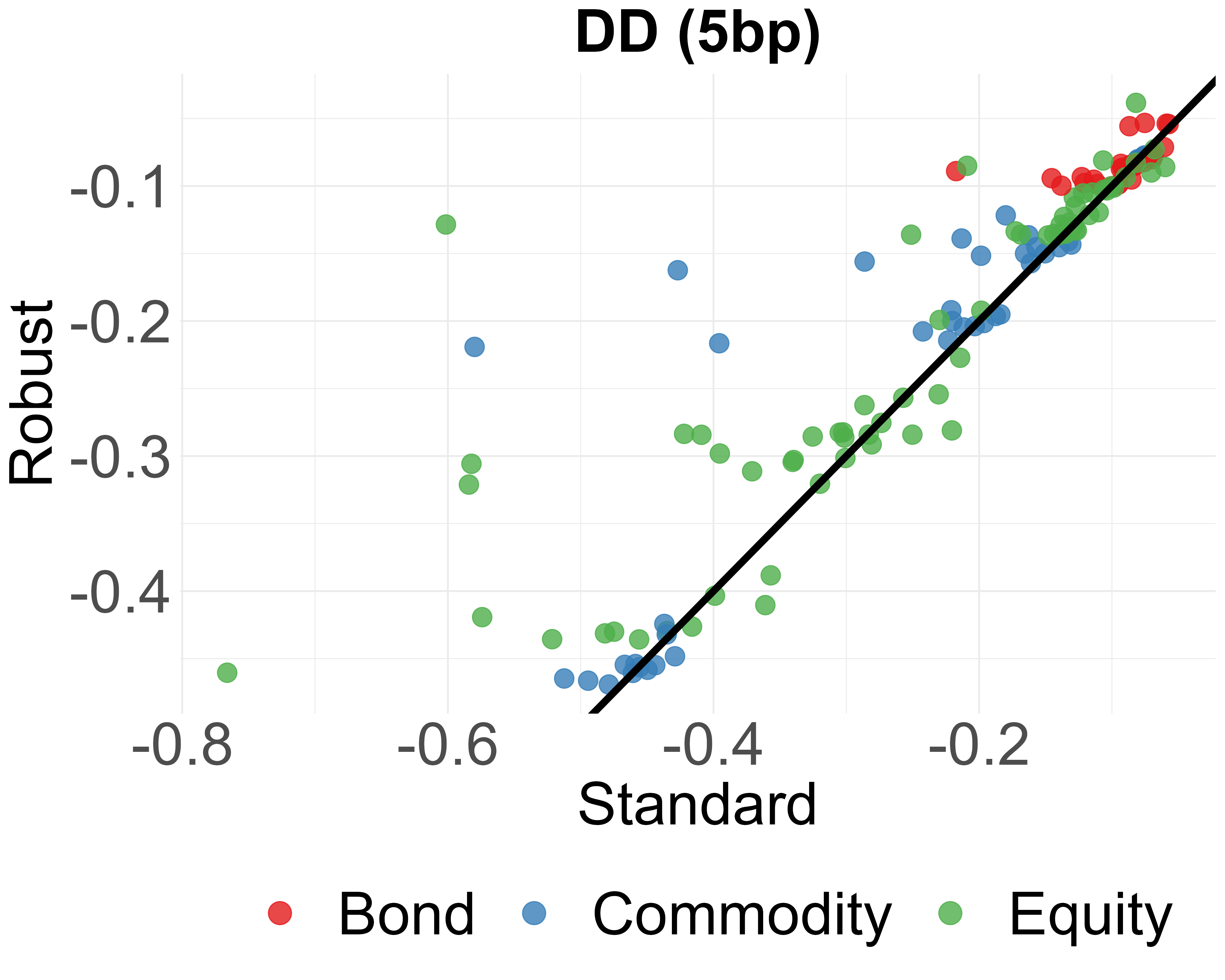}} \quad
{\includegraphics[width=.3\textwidth,height=3.5cm]{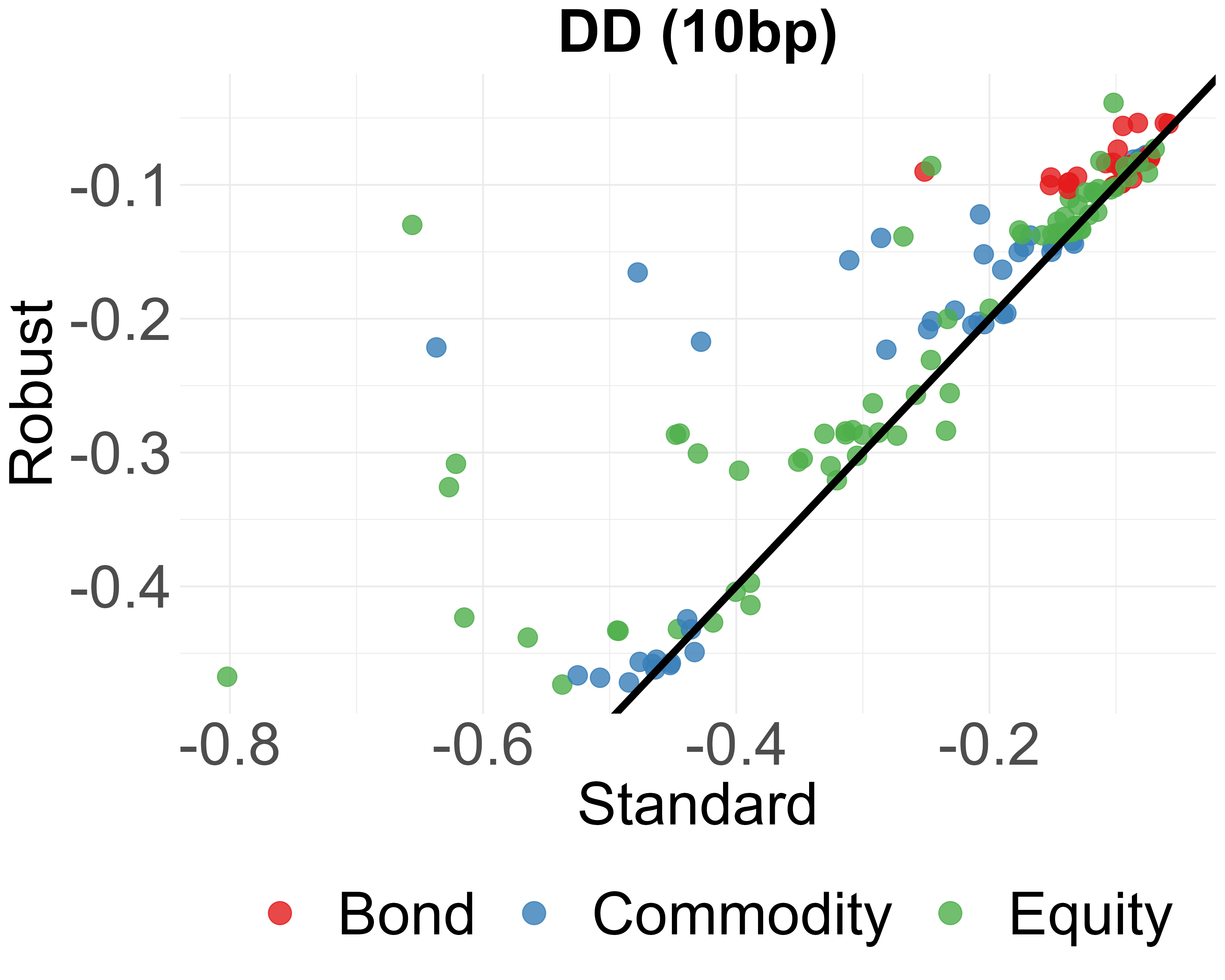}} \quad
\raggedright \hspace{0.05cm} {\scriptsize D}  \hspace{3.75cm} {\scriptsize E} \hspace{5.75cm} {\scriptsize F}\\
{\includegraphics[width=.3\textwidth,height=3.5cm]{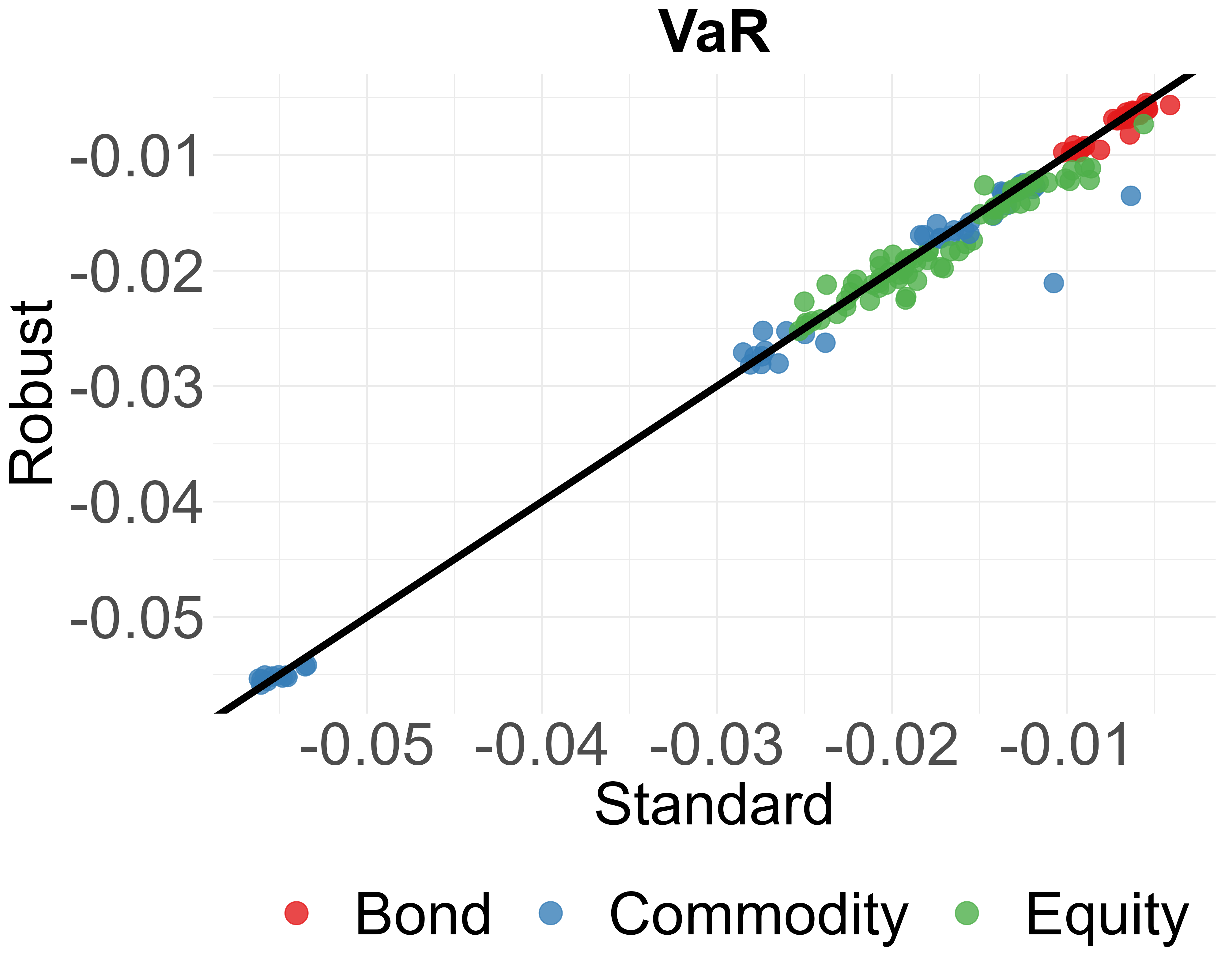}} \quad
{\includegraphics[width=.3\textwidth,height=3.5cm]{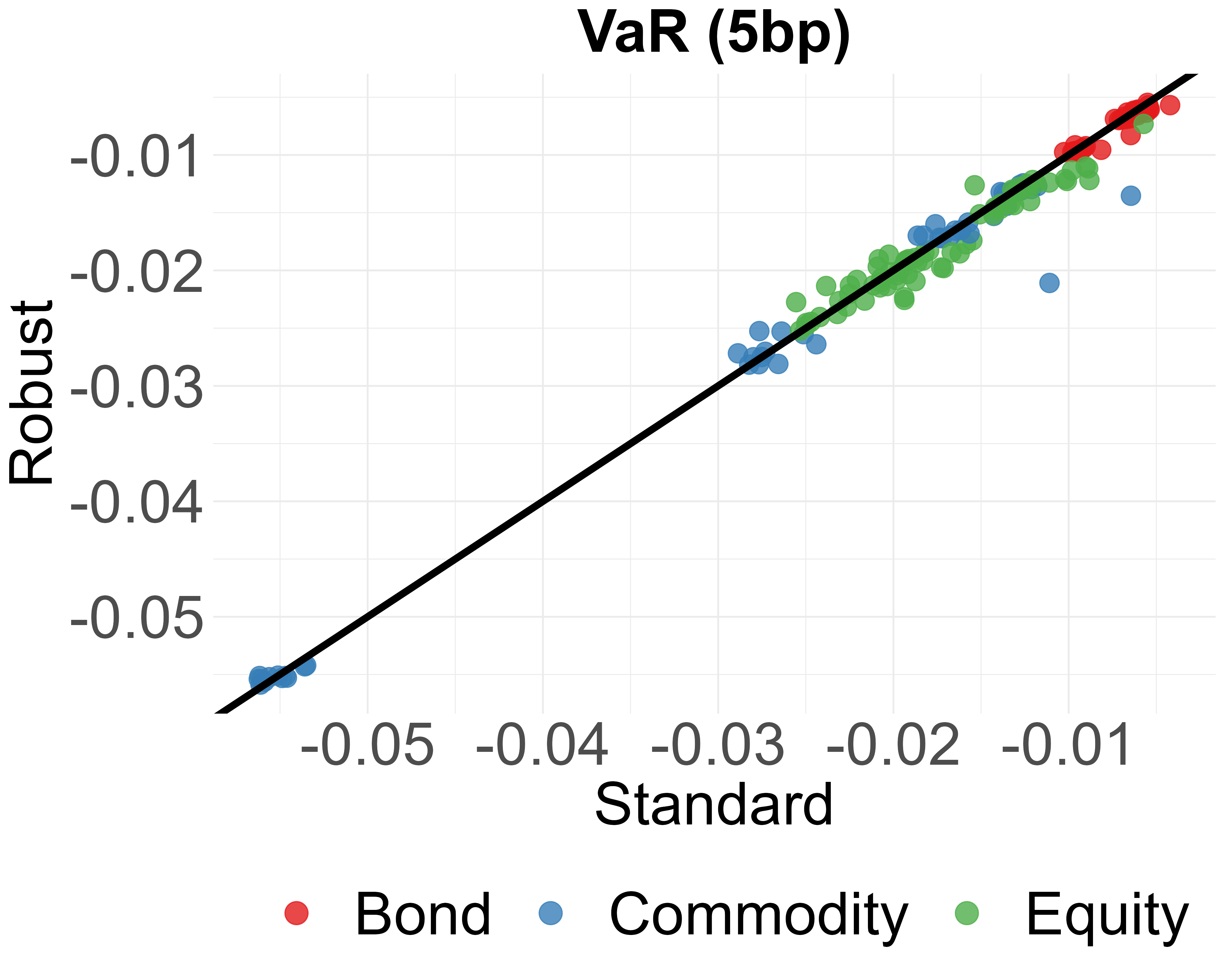}} \quad
{\includegraphics[width=.3\textwidth,height=3.5cm]{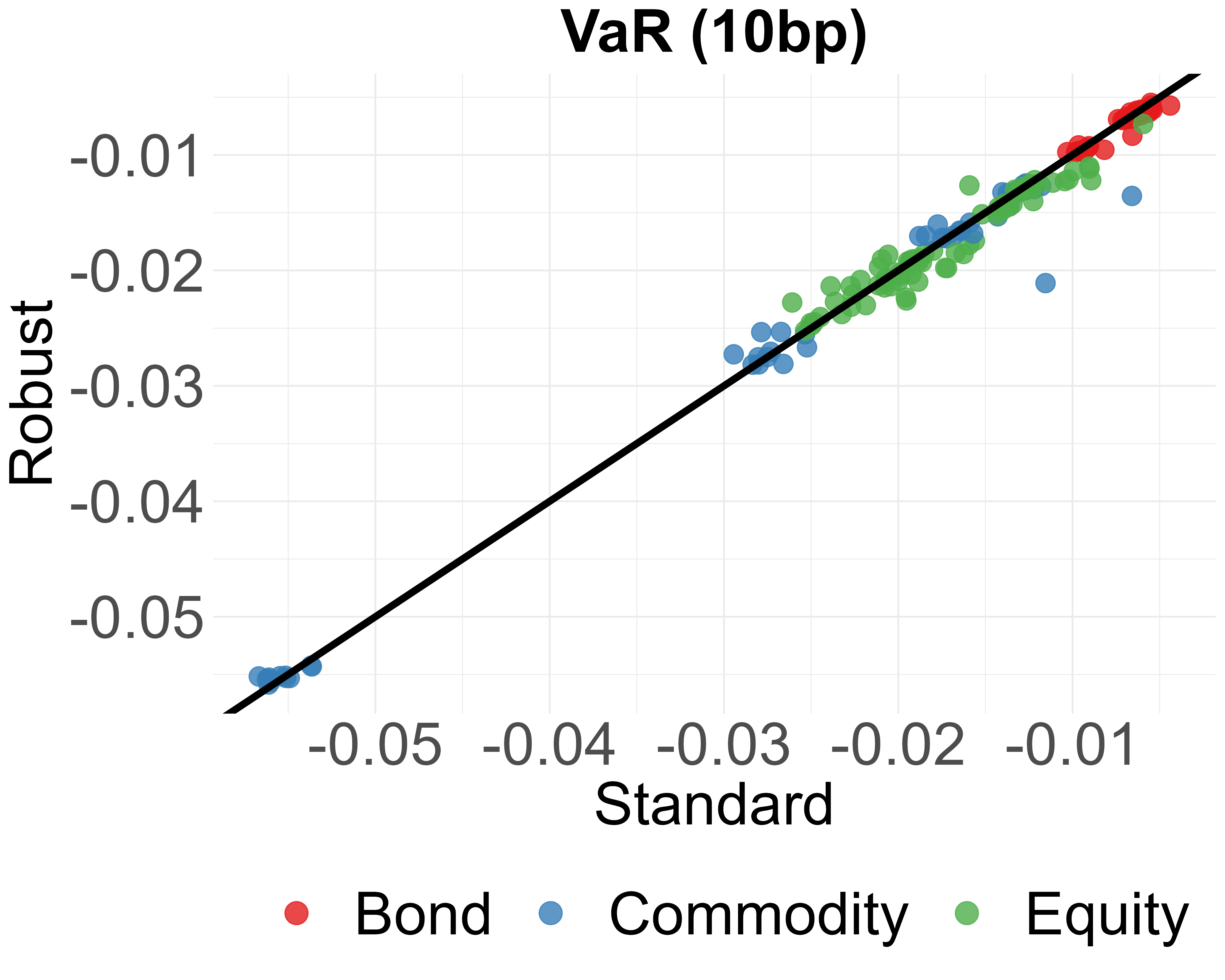}} \quad
\raggedright \hspace{0.05cm} {\scriptsize G}  \hspace{3.75cm} {\scriptsize H} \hspace{5.75cm} {\scriptsize I}\\
{\includegraphics[width=.3\textwidth,height=3.5cm]{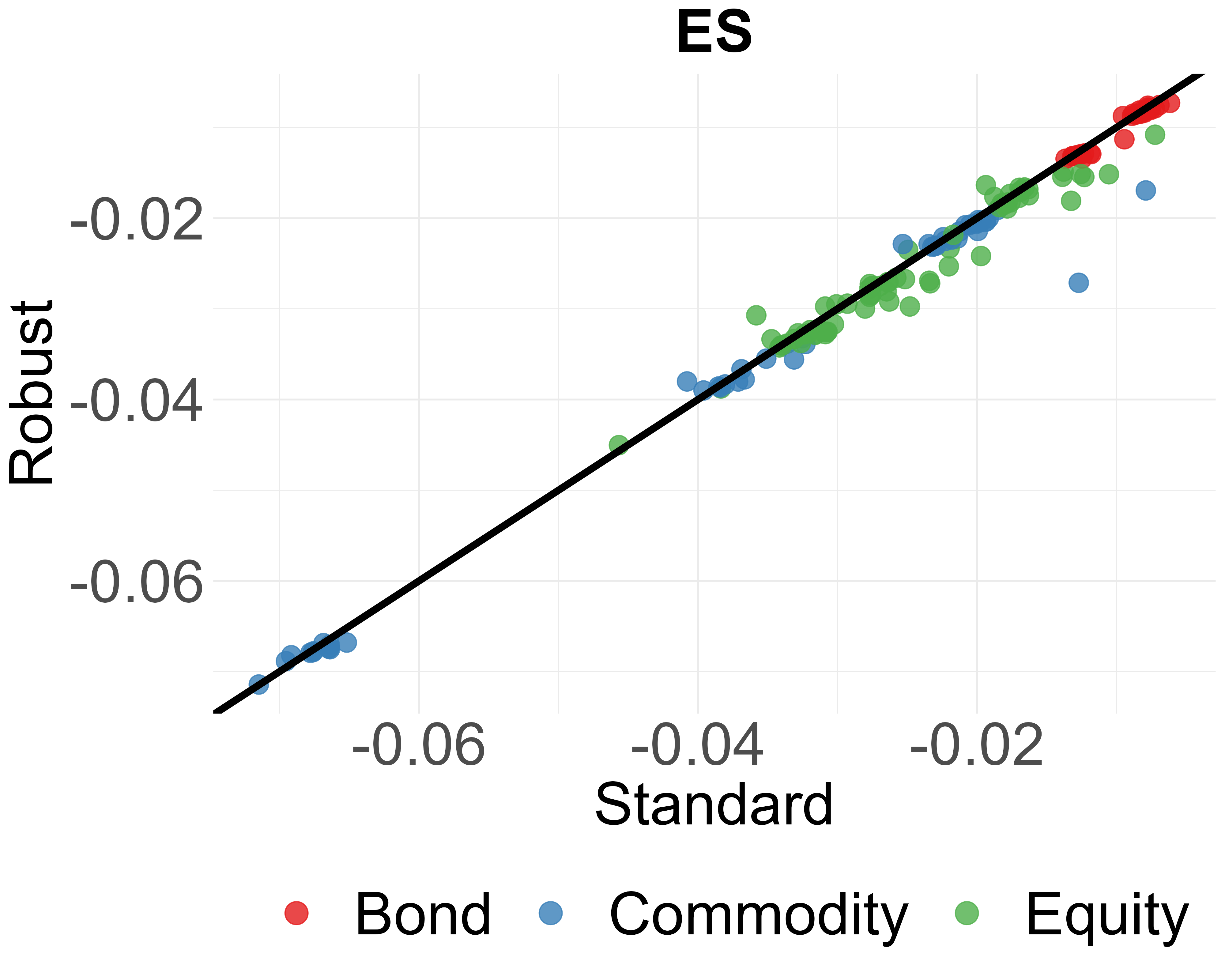}} \quad
{\includegraphics[width=.3\textwidth,height=3.5cm]{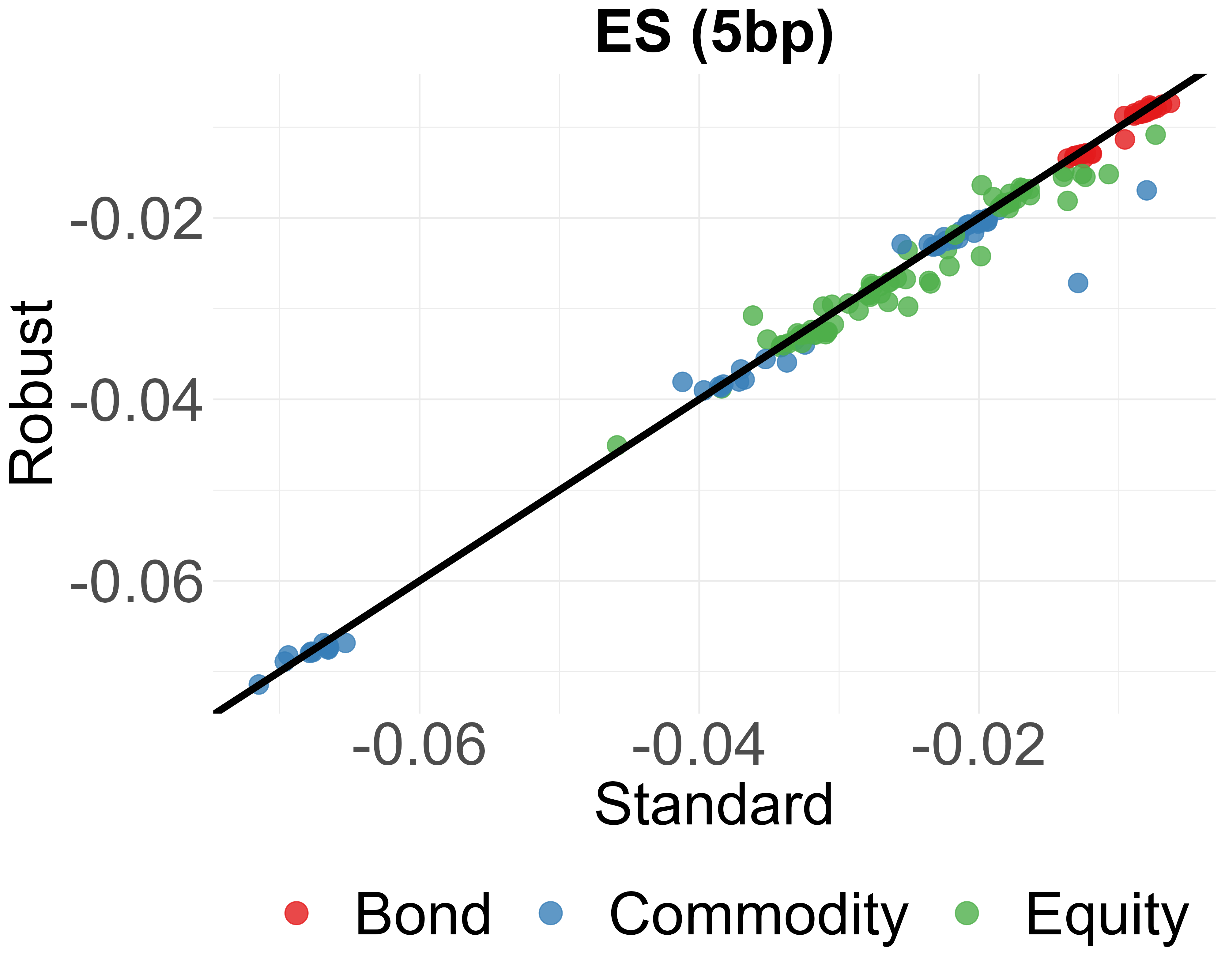}} \quad
{\includegraphics[width=.3\textwidth,height=3.5cm]{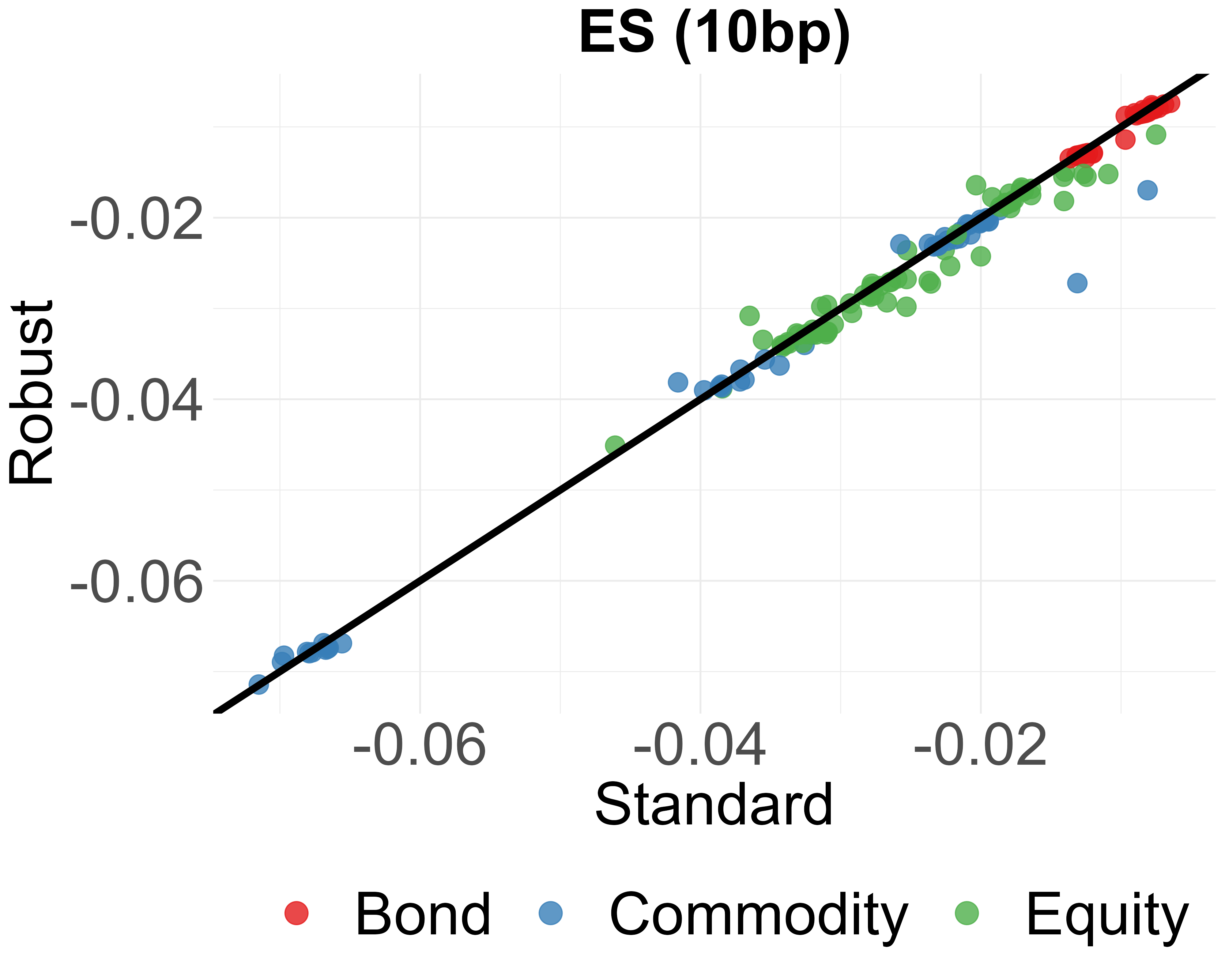}} \quad
\quad
\caption{Scatter plot of various performance metrics obtained using the standard Hedge Ratios and the robust Hedge Ratios. DD, VaR, and ES correspond to maximum drawdown, 95\% Value at Risk, 95\% Expected Shortfall, respectively. Transaction costs are indicated in curly brackets. The color represents the asset class of the hedged asset.}
    \label{fig: transaction_costs2}
\end{figure}

Overall, Fig. \ref{fig: transaction_costs} shows that robust hedge ratios tend to enhance portfolio performance in terms of P\&L, SR, and $\Omega$. The robust method compared to the standard method tends to perform better in terms of P\&L (plots A-B-C), SR (plots D-E-F), and $\Omega$ (plots G-H-I), indicating improved risk-adjusted performances. Similarly, Fig. \ref{fig: transaction_costs2} displays better portfolio performance for the robust method in terms DD, VaR, and ES. The DD (plots A-B-C) tends to be lower in the robust case, indicating larger declines for the standard method. VaR (plots D-E-F) and ES (plots G-H-I) tend to be slightly more favorable in the robust case, suggesting marginally better protection of the robust methodology against extreme losses. Importantly, these improvements generally increase when transaction costs are introduced (as also detailed in Tables \ref{tab: transaction_costs_aaau_no_return}–\ref{tab:transaction_costs_ung} in \ref{app: performance}).

When the analysis is restricted to the asset class of the hedged instrument, several patterns emerge. For precious metals, such as gold (see Table \ref{tab: transaction_costs_aaau_no_return}), robust hedging consistently enhances both P\&L and SR, signalling improved stability in risk compensation. For example, the P\&L (SR) differences between robust and standard methods for the pairs AAAU/CORP and AAAU/IVV amount to 7.589\% (2.600\%) and 6.954\% (2.181\%), respectively, when transaction costs are excluded. Developed-market equity indices also exhibit improvements. The differences in P\&L and $\Omega$ remain positive in most cases when the hedged asset is either the S\&P 500 (Table \ref{tab:transaction_costs_ivv}) or the NASDAQ 100 (Table \ref{tab:transaction_costs_qqqm}). The evidence is more mixed for government bonds and for energy-related equity indices. For Treasury bond indices (Tables \ref{tab:transaction_costs_govt} and \ref{tab:transaction_costs_igov}), differences in most measures, such as VaR and ES, tend to be negative. Similarly, the performance of couples including the clean energy index (Table \ref{tab:transaction_costs_icln}) as the hedged asset generally deteriorates in the robust setting: for instance, the difference in DD for ICLN/ASHR and ICLN/EWH equals $-$7.485\% and $-$5.635\%, respectively, while the differences in VaR (ES) for ICLN/IEV and ICLN/EWH are equal to $-$0.492\% ($-$0.309\%) and $-$0.328\% ($-$0.290\%).

These results improve when considering transaction costs. For instance, when considering a unit transaction cost of 5bps, the P\&L (SR) differences for the pairs AAAU/CORP and AAAU/IVV rise to 13.601\% (4.658\%) and 8.305\% (2.653\%), respectively. Analogously, with a unit transaction cost of 10bps, more couples (eleven out of twelve) compared to the gross-case present positive differences in P\&L and SR when the hedged asset is either the S\&P 500 or the NASDAQ 100.

To assess the statistical significance of the differences between the robust‐hedged and standard‐hedged strategies, two bootstrap procedures are employed. Both approaches rely on 10,000 replications and evaluate the metrics over one-year intervals (i.e., 250 trading days). The first procedure samples blocks randomly, whereas the second preserves the temporal dependence structure of the data.

In the first approach, 250-day blocks are sampled with replacement from the original series, and all performance and risk metrics are recomputed for both methods within each sampled block. Each replication therefore returns a paired observation consisting of the robust and the corresponding standard metric. Columns 2 and 3 of Table \ref{tab:transaction_costs_diff} report the resulting distributional differences. P\&L, SR, and $\Omega$ display statistically significant positive differences at the 1\% level, indicating higher profitability, improved risk-adjusted returns, and a more favourable gain-to-loss profile under the robust strategy. DD is also significantly more favorable in the robust case. Differences in VaR and ES, although positive, are not statistically significant, suggesting that tail-risk improvements are not consistently detectable under this sampling scheme.

To confirm the robustness of these findings when temporal dependence is preserved, the analysis is repeated using the Maximum Entropy Bootstrap (MEB) of \citet{vinod2009maximum}. The MEB generates pseudo-time series that retain key dependence features, such as autocorrelation, by resampling values while preserving their rank ordering. In our implementation, the MEB is applied separately to the robust and standard cases; each replication yields a pseudo-series of the same length of 250 days as in the previous procedure and each observation consists of paired block-level metrics, ensuring direct comparability. Paired differences are computed across the 10,000 replications, and inference is based on the empirical distribution of these differences. P-values are computed as the proportion of bootstrap replications yielding a difference with opposite sign relative to the sample estimate. Columns 4 and 5 of Table \ref{tab:transaction_costs_diff} show that the MEB results confirm those obtained under random block sampling.

\begin{table}[!htbp]
\scalebox{0.82}{
\centering
\setlength{\tabcolsep}{8pt}
\renewcommand{\arraystretch}{1.2}
\begin{tabular}{lcccc}
\hline
\hline
\textbf{Measure} 
& \textbf{Mean Difference} 
& \textbf{p-value} 
& \textbf{Mean Difference} 
& \textbf{p-value} \\
\textbf{}
& \textbf{} 
& \textbf{} 
& \textbf{(temporal order)} 
& \textbf{(temporal order)} \\

\hline
P\&L      & $0.091$ & $0.002$ & $0.088$ & $0.006$ \\
SR        & $0.095$  & $0.001$ & $0.089$  & $0.004$ \\
$\Omega$  & $0.560$  & $0.000$ & $1.020$  & $0.000$ \\
DD        & $0.722$ & $0.007$ & $0.628$ & $0.005$ \\ 
VaR       & $0.039$  & $0.180$ & $0.041$  & $0.223$ \\
ES        & $0.021$  & $0.147$ & $0.019$  & $0.137$ \\
\hline
\hline
\end{tabular}}
\caption{Columns \textit{Mean Difference} and \textit{p-value} report the differences between metrics of portfolios, considering a unit transaction cost of 5 bps, obtained using robust hedge ratios and those using standard hedge ratios, and their p-values, respectively. Columns \textit{Mean Difference (temporal order)} and \textit{p-value (temporal order)} report the corresponding values obtained by preserving the temporal order of the sampled observations.}
\label{tab:transaction_costs_diff}
\end{table}

\section{Conclusions}\label{sec:concl}

This paper develops a robust framework for dynamic minimum-variance hedging that explicitly incorporates volatility and covariance forecast uncertainty through a tractable optimization approach. By combining realized high-frequency risk measures, autoregressive volatility modeling, and box-uncertainty sets, the proposed methodology delivers hedge ratios that are more stable and less sensitive to estimation errors. Empirical evidence across a diversified set of ETFs shows that the robust approach achieves comparable variance reduction while improving downside protection, reducing turnover, and enhancing risk-adjusted performance, particularly when transaction costs are taken into account. Bootstrap analysis further supports the statistical relevance of these improvements, highlighting the practical value of robustness in real-world hedging applications.

Future research could extend this framework by exploring alternative uncertainty sets or Bayesian formulations. Additional work may also investigate multi-asset portfolio settings, liquidity considerations, and the interaction between robust hedging and market microstructure effects, thereby further enhancing the applicability of robust risk management techniques to higher frequencies.

\section*{Acknowledgments}

\noindent AR and PM acknowledge the financial support by the Italian Ministry of University and Research under the PRIN 2022 PNRR project {\it Climate Change, Uncertainty and Financial Risk: Robust
Approaches based on Time-Varying Parameters} (grant n. P20229CJRS, CUP B53D23026390001).

\noindent MP acknowledges the financial support by the European Union - Next Generation EU (PRIN research project 2022FPLY97), as well as by the 'AIDMIX' and 'RATIONALISTS' research projects of the University of Perugia, funded under the FRB-2021 and FRB-2022 programs.

\clearpage
\newpage
\bibliographystyle{elsarticle-harv}
\bibliography{biblio_2207}

\clearpage
\newpage
\appendix

\section{Proof of Proposition \ref{prop:1}}\label{sec:app1}
\begin{proof}
For $h\in\mathbb{R}$ and $(\tilde\sigma_S^2,\tilde\sigma_F^2, \tilde\sigma_{SF})\in\mathcal U$, define
\[
\phi(h,\tilde\sigma_S^2,\tilde\sigma_F^2, \tilde\sigma_{SF})
:=
\tilde\sigma_S^2
+ h^2 \tilde\sigma_F^2
- 2 h \tilde\sigma_{SF}.
\]

We first solve the inner maximization problem. Fix $h\in\mathbb{R}$. The function $\phi(h,\cdot,\cdot, \cdot)$ is affine and separable in the variables $\tilde\sigma_S^2$, $\tilde\sigma_F^2$, and $\tilde\sigma_{SF}$ therefore we can rewrite the problem as
\begin{equation}\label{eq:opt_split}
\begin{split}
\max_{(\tilde\sigma_S^2,\tilde\sigma_F^2, \tilde\sigma_{SF})\in\mathcal U}\phi(h,\tilde\sigma_S^2,\tilde\sigma_F^2, \tilde\sigma_{SF})
&=\max_{\tilde{\sigma}_S^2\in[\sigma^2_{S}-\Theta_{S},\sigma^2_{S}+\Theta_{S}]}
\tilde\sigma_S^2\ +  \\
&+ \max_{\tilde{\sigma}_F^2\in[\sigma^2_{F}-\Theta_{F},\sigma^2_{F}+\Theta_{F}]} h^2 \tilde\sigma_F^2 \  + \\
&+ \max_{\tilde{\sigma}_{SF}\in[\sigma_{SF}-\Theta_{SF},\sigma_{SF}+\Theta_{SF}]} (-2 h \tilde\sigma_{SF}).
\end{split}
\end{equation}
We observe that 
\begin{equation*}
\begin{split}
&\max_{\tilde{\sigma}_S^2\in[\sigma^2_{S}-\Theta_{S},\sigma^2_{S}+\Theta_{S}]}
\tilde\sigma_S^2 = \sigma_S^2 + \Theta_S \\
&  \max_{\tilde{\sigma}_F^2\in[\sigma^2_{F}-\Theta_{F},\sigma^2_{F}+\Theta_{F}]} h^2 \tilde\sigma_F^2 = h^2(\sigma_F^2 + \Theta_F) .
\end{split}
\end{equation*}
On the other hand, for the third term in Eq. \eqref{eq:opt_split} we need to distinguish three cases depending on the value of $h$, i.e., 
\begin{equation*}
 \max_{\tilde{\sigma}_{SF}\in[\sigma_{SF}-\Theta_{SF},\sigma_{SF}+\Theta_{SF}]} (-2 h \tilde\sigma_{SF}) = 
\begin{cases}
-2 h (\sigma_{SF} - \Theta_{SF}) &\text{ if } h > 0 \\
0 &\text{ if } h = 0 \\
-2 h (\sigma_{SF} + \Theta_{SF}) &\text{ if } h < 0. \\
\end{cases}
\end{equation*}
We can define the worst-case objective function as
\begin{equation*}
\psi(h):= (\sigma_S^2 + \Theta_S) + h^2(\sigma_F^2 + \Theta_F) -2h\sigma_{SF} + 2|h|\Theta_{SF}
\end{equation*}
and the robust problem is equivalent to
\[
h^* = \arg\min_{h\in\mathbb{R}} \psi(h). 
\]

We observe that the function $\psi$ is not differentiable in $h=0$ if $\Theta_{SF} > 0$ and that it is strictly convex on $\mathbb{R}$ since it is the sum of the strictly convex function $h^2(\sigma_F^2 + \Theta_F)$ where $\sigma_F^2+\Theta_F>0$ by assumption, the affine function $-2h\sigma_{SF}$, and the convex function $ 2|h|\Theta_{SF}$. Therefore, $h^*$ is the unique global minimizer of the robust problem if and only if it satisfies the first-order optimality conditions:
\begin{equation*}
\begin{cases}
2h^*(\sigma_F^2 + \Theta_F) - 2\sigma_{SF} + 2\Theta_{SF} = 0 &\text{ if } h^* > 0 \\ 
2h^*(\sigma_F^2 + \Theta_F) - 2\sigma_{SF} - 2\Theta_{SF} = 0 &\text{ if } h^* < 0 \\
0 \in \partial \psi(0) &\text{ if } h^* = 0.
\end{cases}
\end{equation*} 
This implies that if $\sigma_{SF} > \Theta_{SF}$, the optimal hedge ratio is
\begin{equation*}
h^* = \frac{\sigma_{SF} - \Theta_{SF}}{\sigma^2_F + \Theta_F},
\end{equation*}
if $\sigma_{SF} < -\Theta_{SF}$, it is 
\begin{equation*}
h^* = \frac{\sigma_{SF} + \Theta_{SF}}{\sigma^2_F + \Theta_F},
\end{equation*}
and it is null if $(- 2\sigma_{SF} + 2\Theta_{SF}) \geq 0 \land (-2\sigma_{SF} - 2\Theta_{SF}) \leq 0$, i.e., $\sigma_{SF} \in [-\Theta_{SF}, \Theta_{SF}]$. 

We can conclude that 
\begin{equation*}
 h^*= 
\begin{cases}
\frac{\sigma_{SF} - \Theta_{SF}}{\sigma^2_F + \Theta_F} &\text{ if } \sigma_{SF} > \Theta_{SF} \\
0 &\text{ if } -\Theta_{SF} \leq \sigma_{SF} \leq \Theta_{SF}\\
\frac{\sigma_{SF} + \Theta_{SF}}{\sigma^2_F + \Theta_F} &\text{ if } \sigma_{SF} < -\Theta_{SF} \\
\end{cases}
\end{equation*}
and this can be rewritten as 
\begin{equation*}
h^*=\frac{sgn(\sigma_{SF}) \Big(|\sigma_{SF}| - \Theta_{SF}\Big)^+}{\sigma_F^2+\Theta_F},
\end{equation*}
which is the unique optimal solution of the robust problem.

\end{proof}

\section{Uncertainty box for AR(p) models}\label{sec:app2}
In Section \ref{sec:method}, we illustrate the robust hedging approach that we propose. As it is shown by Eq. \eqref{eq:h}, the hedge ratio is expressed in terms of the volatility of the asset and the hedging instrument, the covariance between them, and their uncertainty boxes which represent the main novelty of our method. In this appendix we show that when volatility is modeled as an AR(p) process, closed-form expressions for these uncertainty intervals can be derived. The procedure to obtain the uncertainty intervals related to covariance is analogous.

In the setting of our approach, the estimated volatility over a horizon $\tau$ is the square root of the integrated variance over $\tau$ steps, namely $\hat{\sigma}_{t+\tau}^2 = \sum_{j=1}^\tau \hat{y}_{t+j}$ where $\hat{y}_{t+j}$ is the $j$-th step-ahead forecast of the AR(p) process. In particular, given the realized variance as an estimate of the squared volatility, we consider $y_{t+1}=\phi_0+\sum_{j=0}^{p-1}\phi_iy_{t-j}+\eta_{t+1}$ with the assumption of zero mean and finite second-order moment $\sigma_\eta^2$ of the error term. Then, the uncertainty interval $\Theta_{\tau}$ associated with $\hat{\sigma}_{t+\tau}^2$ is described by the standard deviation of the forecast error $e_\tau\equiv\left(\sum_{j=1}^{\tau}y_{t+j}\right)-\left(\sum_{j=1}^\tau\hat{y}_{t+j}\right)$. 

In the following, we aim to obtain $\Theta_{\tau}$ analytically. As preliminary steps, we derive closed-form expressions for $Var(y_{t+j}- \hat{y}_{t+j})$ when both AR(1) and generic AR(p) models are considered. Then, using these results, we compute $Var(e_\tau)$ and compare the output resulting from the closed-form expression with its empirical counterpart.


\subsection{$Var(y_{t+j}- \hat{y}_{t+j})$ for AR(1)}
Given $\mathcal{F}_t$ the information set up to time $t$, the $j$-step-ahead forecast is defined as $\hat{y}_{t+j}=\mathbb{E}(y_{t+j}\vert \mathcal{F}_t)$ and, as such, if $j > 1$, it can be computed recursively starting from the one-step-ahead forecast. For the AR(1) model, under the assumptions outlined in the introduction of this appendix, we obtain 
$$\hat{y}_{t+ j} = \mathbb{E}[y_{t+j}|\mathcal{F}_t] = \phi_0 + \phi_0\phi_1 + \phi_0\phi_1^2 + \ldots + \phi_0\phi_1^{j-1} + \phi_1^{j}y_{t}.$$
Similarly, $y_{t+j}$ can be rewritten as 
$$y_{t+j} = \hat{y}_{t+ j} + \phi_1^{j-1}\eta_{t+1} + \phi_1^{j-2}\eta_{t+2} + \ldots + \eta_{t+j}$$ 
and the variance of the forecast error can be obtained:
\begin{equation}\label{eq_uncertainty_ar1}
    Var(y_{t+j}- \hat{y}_{t+j}) = Var\Big(\sum_{i=1}^{j}(\phi_1^{j-i})^2\eta_{t+i}\Big) = \Big(\sum_{i=0}^{j - 1}\phi_1^{2i}\Big)\sigma_\eta^2.
\end{equation}


\subsection{$Var(y_{t+j}- \hat{y}_{t+j})$ for AR(p)}

Eq. \eqref{eq_uncertainty_ar1} holds for linear autoregressive processes of order 1. In order to generalize this formula for a generic AR(p) model, it is useful to find its Moving Average (MA) representation. Using the lag operator, the original AR(p) model for the time series after mean removal $\{\tilde{y}_{t}\}_t$ can be stated as
\begin{equation}\label{eq_AR_representation}
    \Phi(L)\tilde{y}_{t+j} = \eta_{t+j}
\end{equation}
where $\Phi(L) = 1 - \phi_1 L - \phi_2 L^2 + \ldots - \phi_p L^p$ is a lag polynomial with $\sum_{i= 1}^p\phi_i^2 < \infty$\footnote{Let $X_t$ be a stochastic process. Then, the lag operator $L$ is defined as $L X_t \equiv X_{t-1} $. It is linear and admits power exponent i.e., $L^kX_t = X_{t-k}$. Expressions like $\theta_0 + \theta_1 L + \theta_2 L^2 + \ldots + \theta_n L^n $ are called \textit{lag polynomials} and are denoted as $\theta(L)$.}. An AR process is always invertible, i.e., it can always be written in its MA($\infty$) representation as 
$$\tilde{y}_{t+j} = \Psi(L)\eta_{t+j} = \sum_{i=0}^\infty\psi_i\eta_{t+j-i} $$
where $\Psi(L) = \psi_0 + \psi_1 L + \psi_2 L^2 + \ldots $ is a lag polynomial such that $\sum_{i= 0}^\infty\psi_i^2 < \infty$ and $\Phi(L)\Psi(L)=1$. By exploiting this MA representation, the conditional expected value of $\tilde{y}_{t+j}$ given the information set $\mathcal{F}_t$ is
$$\mathbb{E}[\tilde{y}_{t+j}|\mathcal{F}_t] = \sum_{i=j}^\infty\psi_i\eta_{t+j-i}$$
and the $j$-step-ahead forecast error
\begin{equation}\label{eq_tau-step-ahead_forecast_error}
    \tilde{y}_{t+j} - \mathbb{E}[\tilde{y}_{t+j}|\mathcal{F}_t] = \sum_{i=0}^{j-1}\psi_i\eta_{t+j-i}.
\end{equation}
Consequently, we obtain its variance as
\begin{equation}\label{eq_variance_tau-step-ahead_forecast_error}
    Var\Big(\tilde{y}_{t+j} - \mathbb{E}[\tilde{y}_{t+j}|\mathcal{F}_t]\Big) = \Bigg(\sum_{i=0}^{j - 1}\psi_i^2\Bigg)\sigma_\eta^2
\end{equation}
and the coefficients $\psi_i$ can be determined in terms of $\phi_i$ by solving the condition 
\begin{equation}\label{eq_invertibility}
\Phi(L)\Psi(L) = 1 \iff (1 - \phi_1L - \phi_2 L^2 - \ldots - \phi_p L^p)(\psi_0 + \psi_1 L + \psi_2 L^2 + \ldots ) = 1.
\end{equation}
This implies 
\begin{equation*}
    \begin{split}
        &\psi_0 = 1 \\
        &\psi_i = \sum_{k=1}^{\min(i, p)} \phi_k\psi_{i-k}, \ i \geq 1. 
    \end{split}
\end{equation*}
Therefore, the coefficients $\psi_i$ depend on the order of the autoregressive process. 

To illustrate this, consider the following example: we aim to compute the uncertainty interval for $j=4$, given an autoregressive process of order $p=3$. As a first step, by exploiting the AR representation in Eq. \eqref{eq_AR_representation}, we obtain the $j$-step-ahead forecast error:
\begin{equation*}
    \begin{split}
        \tilde{y}_{t+4} - \mathbb{E}[\tilde{y}_{t+4}|\mathcal{F}_t] =& \Big(\phi_0 + \phi_1\tilde{y}_{t+3} + \phi_2\tilde{y}_{t+2} + \phi_3\tilde{y}_{t+1} + \eta_{t+4}\Big) + \\ 
        &-\Big(\phi_0 + \mathbb{E}[\phi_1\tilde{y}_{t+3} + \phi_2\tilde{y}_{t+2} + \phi_3\tilde{y}_{t+1}|\mathcal{F}_t]\Big) = \\
        =& \Big(\phi_0 + \phi_0\phi_1 + (\phi_1^2 + \phi_2)\tilde{y}_{t+2} + (\phi_1\phi_2 + \phi_3)\tilde{y}_{t+1} + \phi_1\phi_3\tilde{y}_{t} + \phi_1\eta_{t+3} + \eta_{t+4}\Big) + \\ 
        &- \Big( \phi_0 + \phi_0\phi_1 + 
 \mathbb{E}[(\phi_1^2 + \phi_2)\tilde{y}_{t+2} + (\phi_1\phi_2 + \phi_3)\tilde{y}_{t+1}|\mathcal{F}_t]+ \phi_1\phi_3\tilde{y}_{t}\Big)=\\
        =&\ldots = \\
         =& [(\phi_1^2 + \phi_2)\phi_1 + \phi_1\phi_2 + \phi_3]\eta_{t+1} + (\phi_1^2 + \phi_2)\eta_{t+2}+\phi_1^2\eta_{t+3} + \eta_{t+4}.
    \end{split}
\end{equation*}
We observe that the constant terms, as well as those involving $\tilde{y}_{t}$ cancel out, since they appear in both $\tilde{y}_{t+4}$ and $\mathbb{E}[\tilde{y}_{t+4}|\mathcal{F}_t]$. This implies that the variance of the $4$-step-ahead forecast error is 
\begin{equation}\label{eq_var_p3_tau4}
\begin{split}
    Var\Big(\tilde{y}_{t+4} - \mathbb{E}[\tilde{y}_{t+4}|\mathcal{F}_t]\Big)= \{ 1 + \phi_1^2 + (\phi_1^2 + \phi_2)^2 + [(\phi_1^2 + \phi_2)\phi_1 + \phi_1\phi_2 + \phi_3]^2 \}\sigma_{\eta}^2.
    \end{split}
\end{equation}
By comparing Eq. \eqref{eq_var_p3_tau4} with Eq. \eqref{eq_variance_tau-step-ahead_forecast_error}, we observe that if $p=3$ and $j=4$, the coefficients $\{\psi_i, \ i = 1, \ldots, j - 1\}$ are:
\begin{equation*}
    \begin{split}
        \psi_0 &= 1 \\
        \psi_1 &= \phi_1 \\
        \psi_2 &= \phi_1^2 + \phi_2 \\
        \psi_3 &= (\phi_1^2 + \phi_2)\phi_1 + \phi_1\phi_2 + \phi_3,
    \end{split}
\end{equation*}
i.e., they are the coefficients which solve the condition  
$$(1 - \phi_1L - \phi_2 L^2 - \phi_3 L^3)(\psi_0 + \psi_1 L + \psi_2 L^2 + \ldots ) = 1.$$

Finally, we note that if $p=1$, the condition in Eq. \eqref{eq_invertibility} implies that $\psi_0 = 1$ and $\psi_i = \phi_1^{i}$, $\forall i \geq 1$. Therefore, the $j$-step-ahead forecast error is 
$$
Var\Big(\tilde{y}_{t+j} - \mathbb{E}[\tilde{y}_{t+j}|\mathcal{F}_t]\Big) = \left(\sum_{i=0}^{j-1}\phi_1^{2i}\right)\sigma_\eta^2$$
as shown in Eq. \eqref{eq_uncertainty_ar1}.

\subsection{Uncertainty interval for integrated variance}
Now, we can focus on the integrated variance over $\tau$ steps and derive its uncertainty interval. Referring to Eq. \eqref{eq_tau-step-ahead_forecast_error} and \eqref{eq_variance_tau-step-ahead_forecast_error}, we obtain:
\begin{equation}\label{eq_Theta_th}
\begin{split}
    Var(e_\tau) &= \sum_{j=1}^{\tau}Var(\tilde{y}_{t+j} - \mathbb{E}[\tilde{y}_{t+j}|\mathcal{F}_t])+ \\ &+  \sum_{j, k = 1; j \neq k}^{\tau} Cov\Big((\tilde{y}_{t+j} - \mathbb{E}[\tilde{y}_{t+j}|\mathcal{F}_t]), (\tilde{y}_{t+k} - \mathbb{E}[\tilde{y}_{t+k}|\mathcal{F}_t])\Big) = \\
    &= \sum_{j=1}^{\tau} \Bigg[\sum_{i=0}^{j - 1}\psi_i^2\Bigg]\sigma_\eta^2 + \sum_{j, k = 1; j \neq k}^{\tau} Cov\Big( \sum_{l=0}^{j-1}\psi_l\eta_{t+j-l}, \sum_{m=0}^{k-1}\psi_m\eta_{t+k-m}\Big) = \\
    &= \sum_{j=1}^{\tau} \Bigg[\sum_{i=0}^{j - 1}\psi_i^2\Bigg]\sigma_\eta^2 + \sum_{j, k = 1; j \neq k}^{\tau} \sum_{l=0}^{j-1} \sum_{m=0}^{k-1}\psi_l\psi_m\mathbb{E}[\eta_{t+j-l}\eta_{t+k-m}] = \\
    &= \sum_{j=1}^{\tau} \Bigg[\sum_{i=0}^{j - 1}\psi_i^2\Bigg]\sigma_\eta^2 + \Bigg(\sum_{j, k = 1; j \neq k}^{\tau} \sum_{l=0}^{j-1} \sum_{m=0}^{k-1}\psi_l\psi_m \mathbb{I}[m = k - j + l]\Bigg)\sigma_\eta^2 = \\
    &= \sum_{j=1}^{\tau} \Bigg[\sum_{i=0}^{j - 1}\psi_i^2\Bigg]\sigma_\eta^2 + \Bigg(\sum_{j, k = 1; j \neq k}^{\tau} \sum_{l=\max(0, j- k)}^{j-1} \psi_l\psi_{k - j + l}\Bigg)\sigma_\eta^2
\end{split} 
\end{equation}
where the coefficients $\psi_i$ depend on the order $p$ of the autoregressive process and solve the condition in Eq. \eqref{eq_invertibility}. Then, the uncertainty interval $\Theta_{\tau}$ associated with $\hat{\sigma}_{t+\tau}^2$ is the square root of $Var(e_\tau)$.

If $p=1$, Eq. \eqref{eq_Theta_th} simplifies to:
\begin{equation*}
\begin{split}
    Var(e_\tau) 
    &= \sum_{j=1}^{\tau} \Bigg[\sum_{i=0}^{j - 1}\phi_1^{2i}\Bigg]\sigma_\eta^2 + \Bigg(\sum_{j, k = 1; j \neq k}^{\tau} \phi_1^{k - j}\sum_{l=\max(0, j- k)}^{j-1} \phi_1^{2l}\Bigg)\sigma_\eta^2,
\end{split} 
\end{equation*}
and it can be easily shown that it is equivalent to:
\begin{equation*}
\begin{split}
    Var(e_\tau) 
    &= \sum_{j=1}^{\tau} \Bigg[\sum_{i=0}^{j - 1}\phi_1^{2i}\Bigg]\sigma_\eta^2 + \Bigg(\sum_{j, k = 1; j \neq k}^{\tau} \sum_{l=1}^{j} \sum_{m=1}^{k} \phi_1^{j-l}\phi_1^{k - m} \mathbb{E}[\eta_{t+l}\eta_{t+m}]\Bigg) = \\
    &= \sum_{j=1}^{\tau} \Bigg[\sum_{i=0}^{j - 1}\phi_1^{2i}\Bigg]\sigma_\eta^2 + \Bigg(\sum_{j, k = 1; j \neq k}^{\tau} \phi_1^{j+k} \sum_{l=1}^{\min(j,k)} \phi_1^{-2l} \Bigg) \sigma_\eta^2.
\end{split} 
\end{equation*}

\subsection{Empirical vs. theoretical uncertainty intervals for integrated variance}
As shown in previous subsections, a closed-form of uncertainty intervals related to the realized variance can be obtained for linear autoregressive processes. In Fig. \ref{fig:Theta_tau_emp_vs_formula} we compare the uncertainty intervals that we obtain by estimating them directly from the predictions with the corresponding ones computed using the formula of Eq. \eqref{eq_Theta_th}\footnote{Only the pairs of instruments such that fitting the AR processes directly on the realized variances, and not their logarithm, returns positive coefficients, are considered.}, for AR(1) and AR(5), $\tau=1, \ldots, 10$. We observe that overall there is a good agreement however, when AR(1) is employed to fit the processes, the intervals computed directly from the predictions overestimate the outputs of the theoretical formula and the effect is stronger for higher prediction horizons $\tau$. This could be due to the fact that in the derivation of Eq. \eqref{eq_Theta_th}, we assume that error terms are orthogonal and homoskedastic, i.e., $\mathbb{E}[\eta_{j}\eta_k] = \sigma_\eta^2\mathbb{I}[j = k]$. In practice, this assumption may not hold and additional terms would contribute to the covariance of forecast errors in Eq. \eqref{eq_Theta_th}. This effect is more pronounced for higher prediction horizons $\tau$ since this amounts to aggregate more contributions in the computation of the $\tau$-step ahead forecast.


\begin{figure}
    \includegraphics[width=0.5\linewidth]{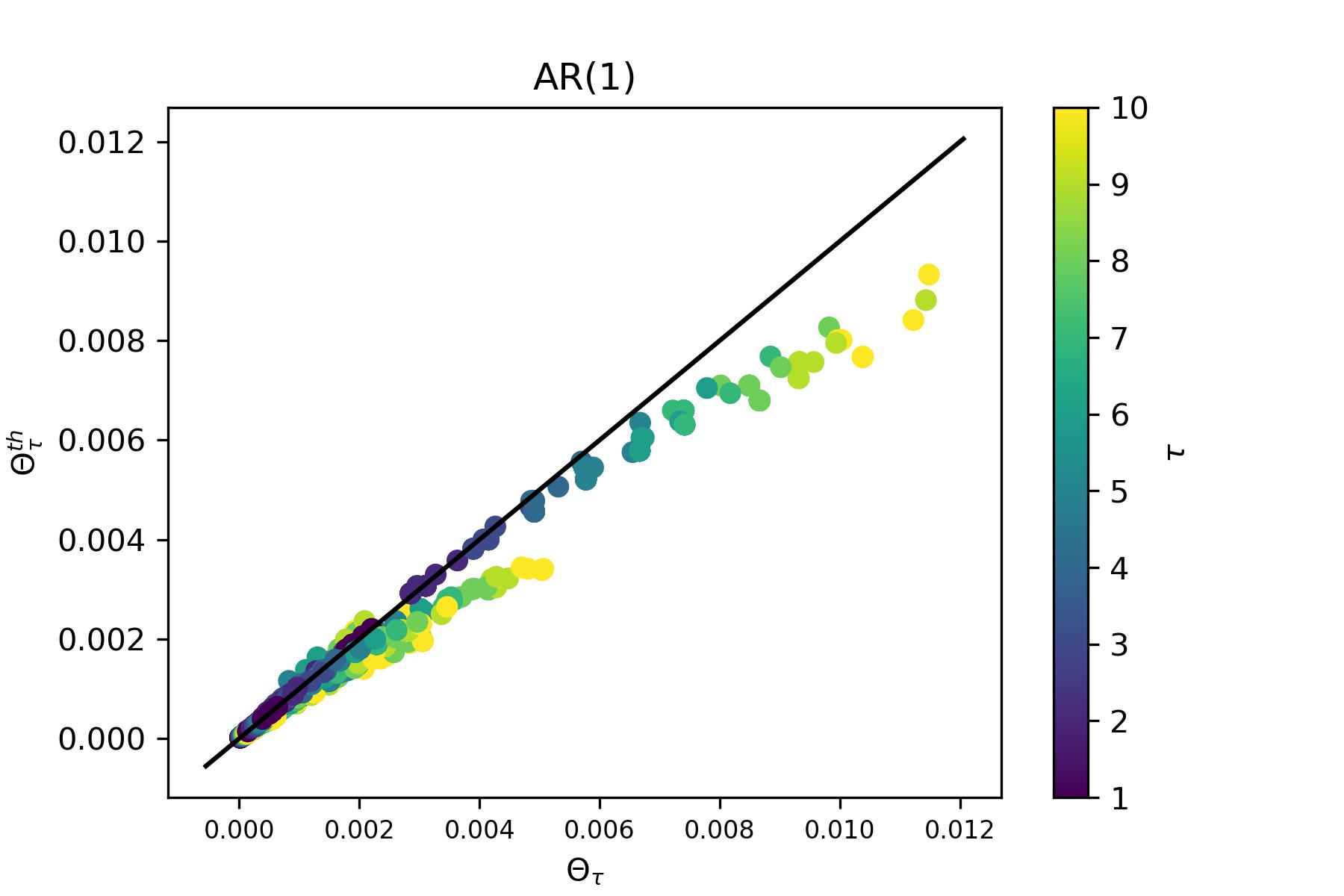}
    \includegraphics[width=0.5\linewidth]{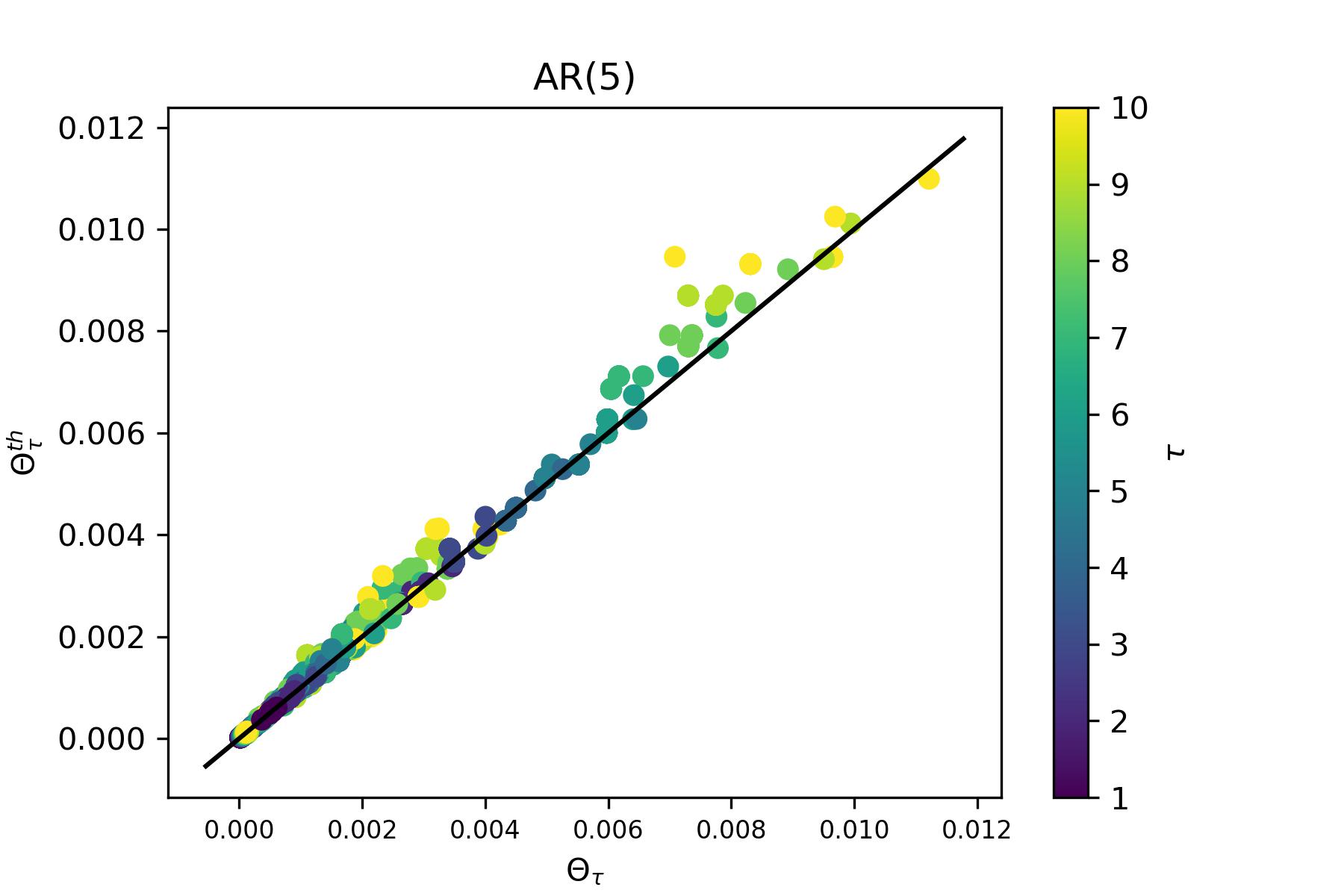}
    \caption{Comparison between the uncertainty intervals that we obtain by estimating them directly from the predictions with the corresponding ones computed using the formula of Eq. \eqref{eq_Theta_th}. Each point is associated with the hedging instrument of a pair and its color is the prediction horizon $\tau$. Realized variances and covariances are fitted by AR(1) (left) and AR(5) (right) processes. The dark line corresponds to the bisector.}
    \label{fig:Theta_tau_emp_vs_formula}
\end{figure}

\section{Additional results}\label{app:plots}

In this appendix, additional plots that are cited in the main text are displayed. They are Fig. \ref{fig:correlations_hist}-\ref{fig:pairtype}-\ref{fig:meanHR_compare}-\ref{fig:AR5}. 

In Fig. \ref{fig:pairtype}, the variable \textit{Pair type} takes the value 1 when both instruments in the pair are equity ETFs, 2 when they are bond ETFs, and 3 when they are commodity ETFs. Values from 4 to 9 correspond to mixed-type co-occurrences.

\begin{figure}
    \centering
    \includegraphics[width=0.5\linewidth]{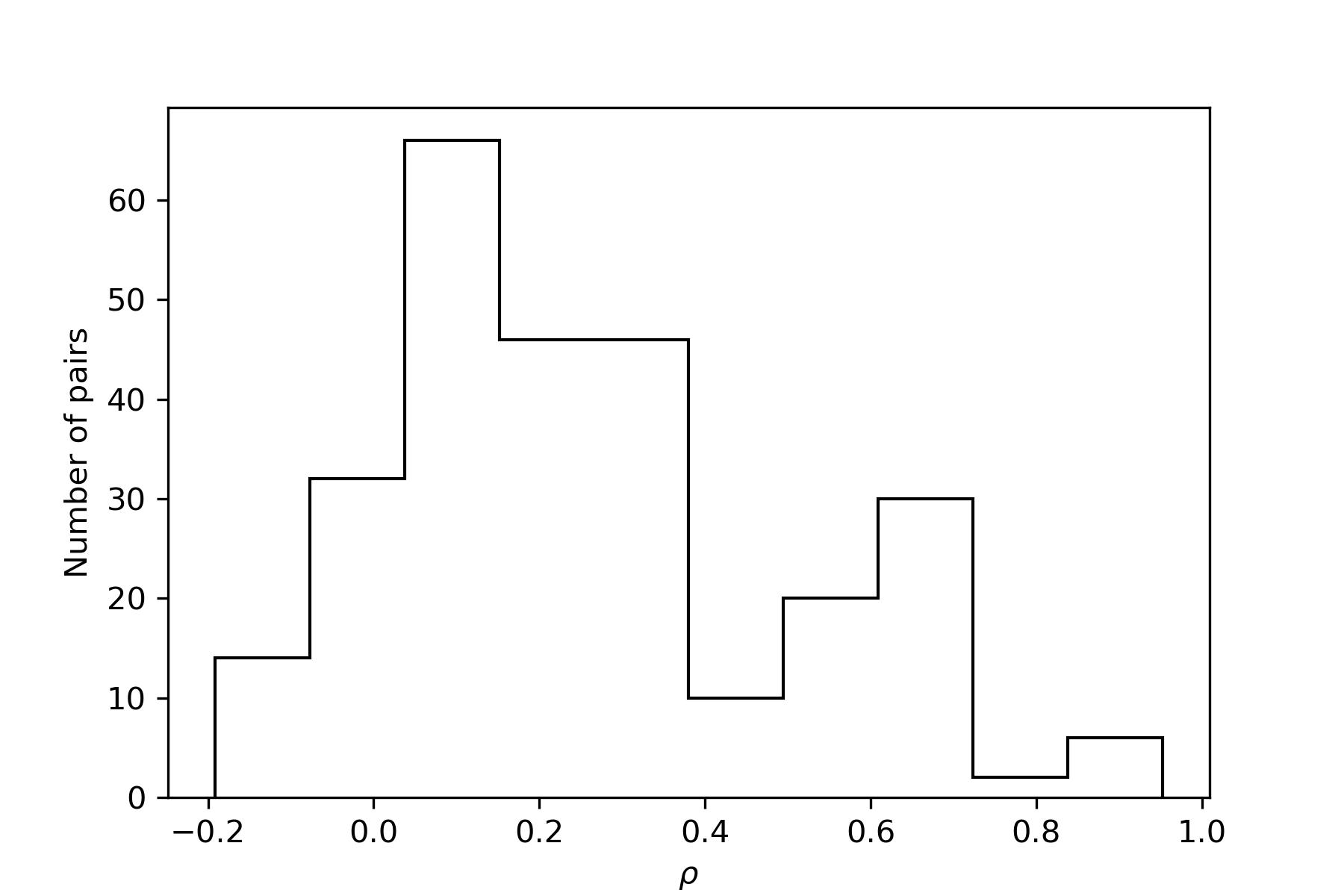}
    \caption{Histogram of return correlations between instrument pairs in the data set.}
    \label{fig:correlations_hist}
\end{figure}

\begin{figure}
\includegraphics[width=0.5\linewidth]{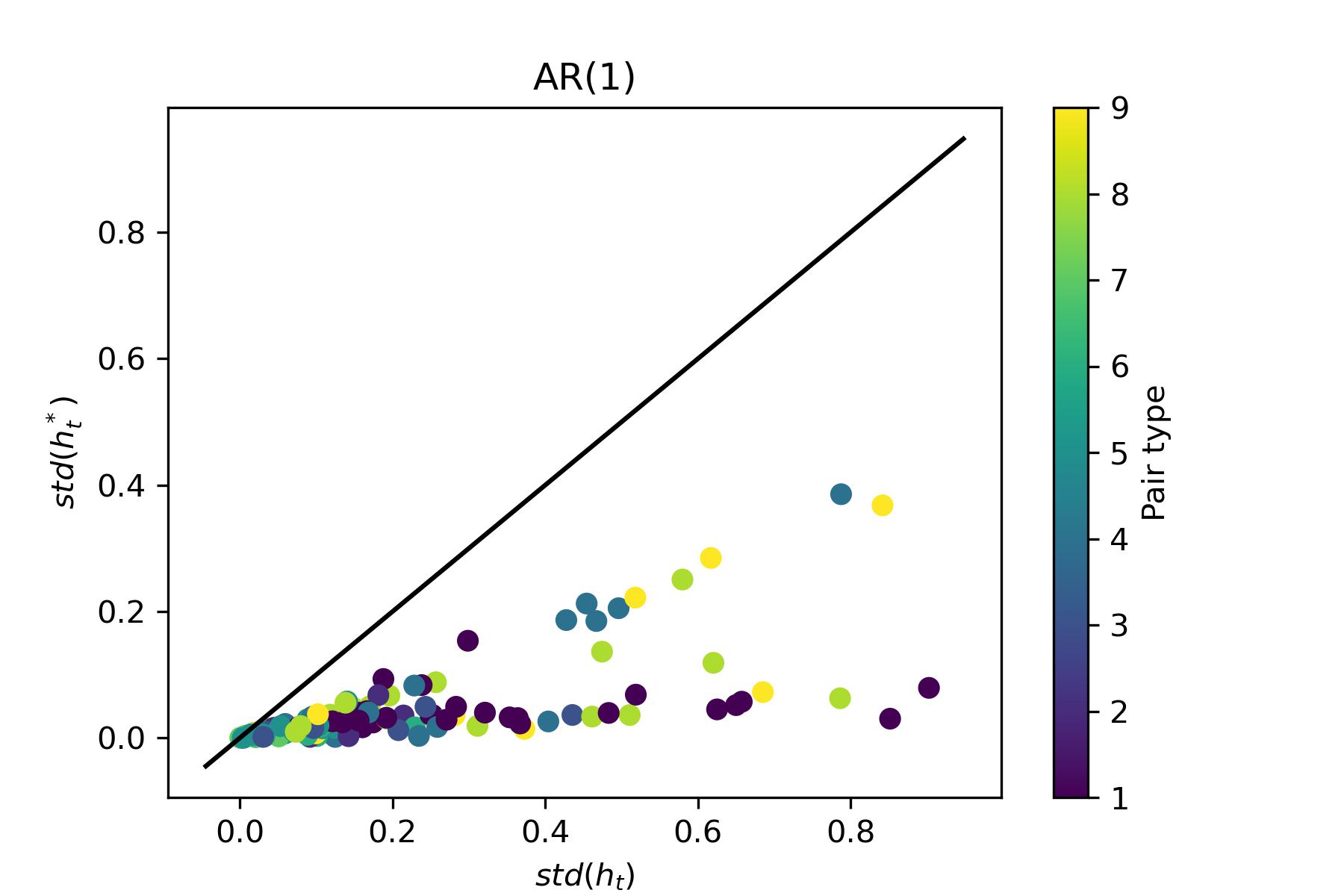} \includegraphics[width=0.5\linewidth]{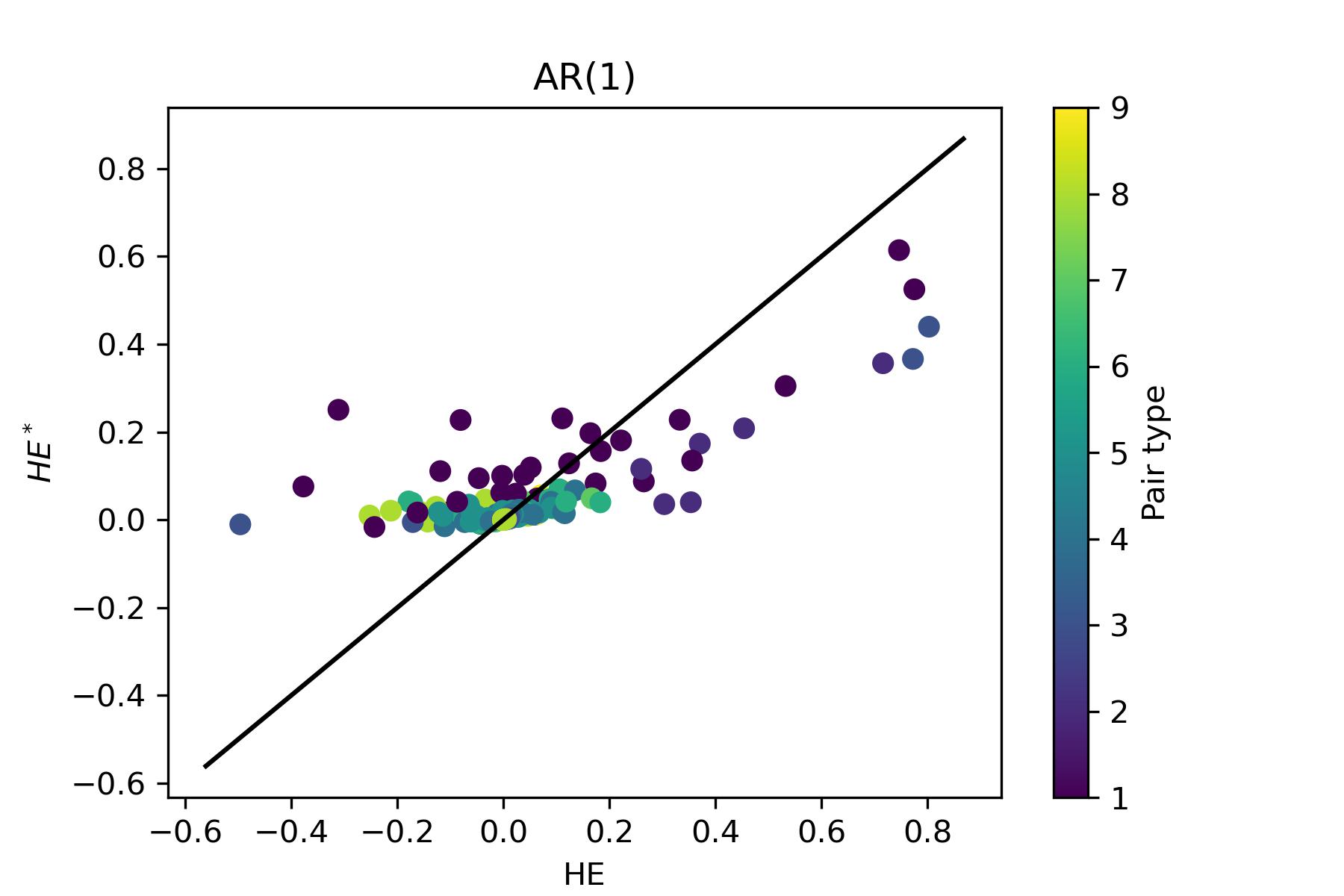}
\includegraphics[width=0.5\linewidth]{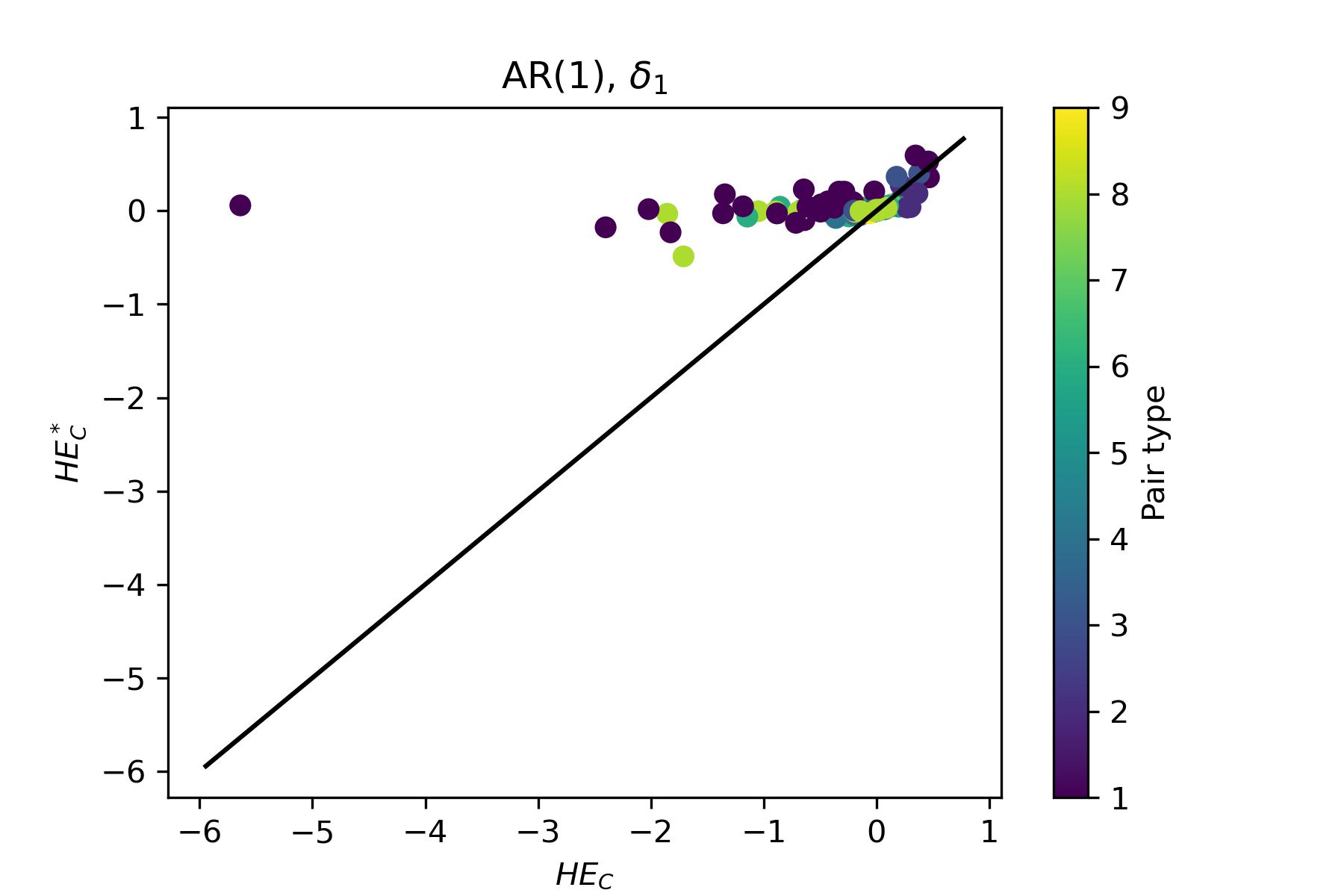}
\includegraphics[width=0.5\linewidth]{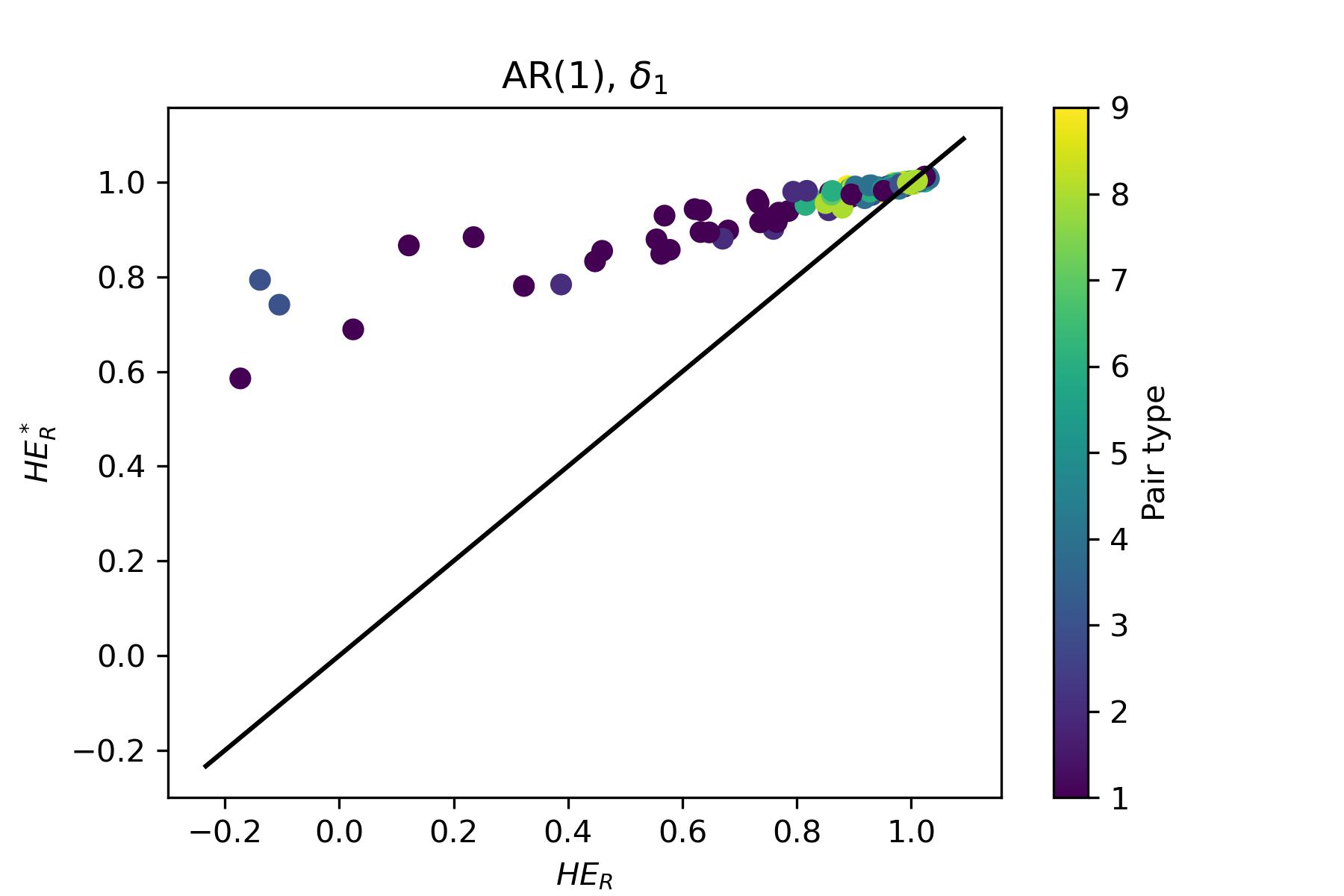}
    \caption{Comparison between the outputs of the standard and robust methodologies when realized variances and covariances are fitted by AR(1) processes. Each point is associated with a pair of instruments and its color is the \textit{Pair type}, i.e., a variable that represents the type of instruments. The threshold $\delta_1$ is the first quartile of the asset returns. The dark line corresponds to the bisector.}
    \label{fig:pairtype}
\end{figure}

\begin{figure}
    \centering
    \includegraphics[width=0.5\linewidth]{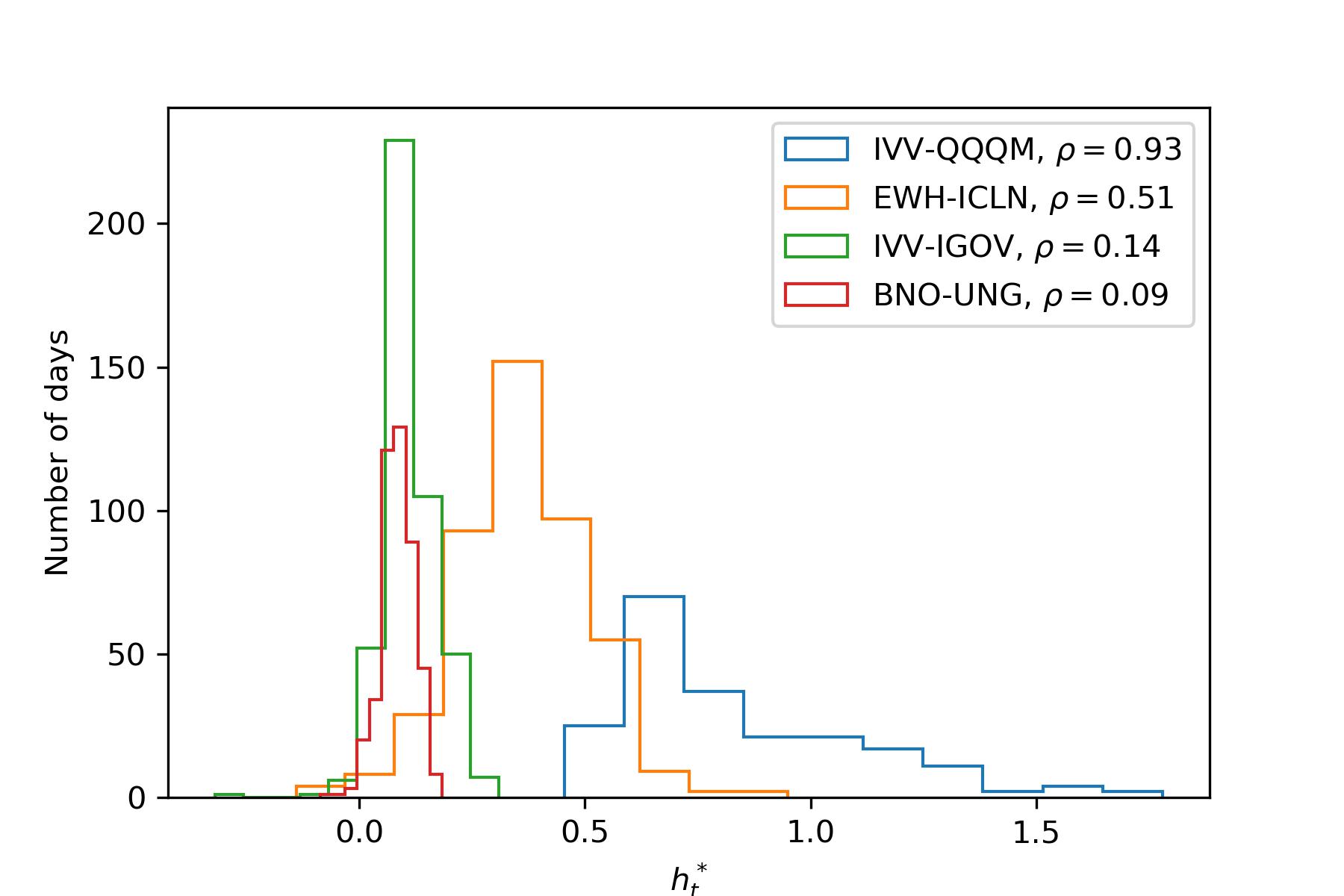}
    \caption{Comparison between the hedge ratios for pairs of instruments with different return correlations $\rho$. The first label refers to the asset and the second to the hedging instrument.}
    \label{fig:meanHR_compare}
\end{figure}

\begin{figure}
\includegraphics[width=0.5\linewidth]{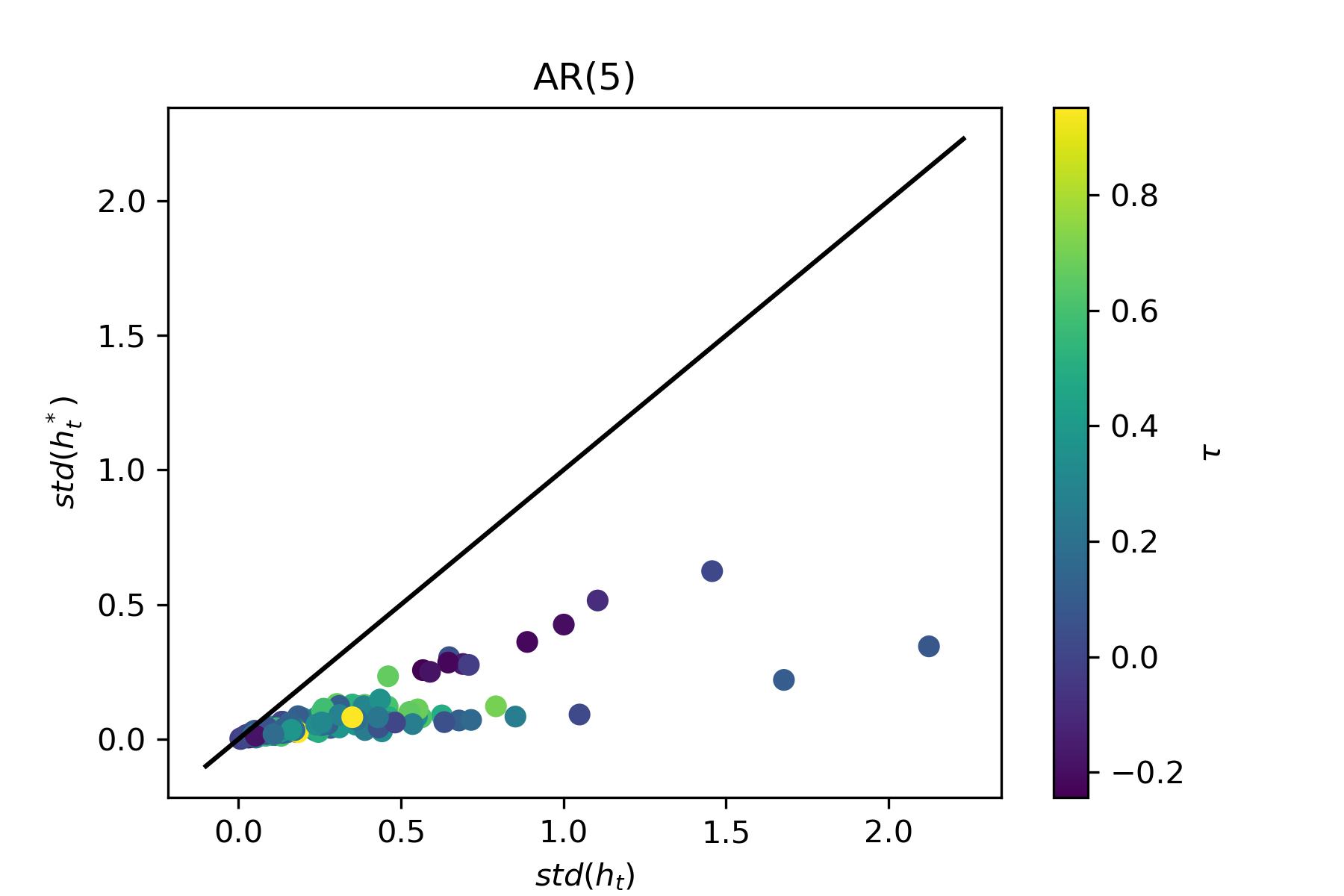} \includegraphics[width=0.5\linewidth]{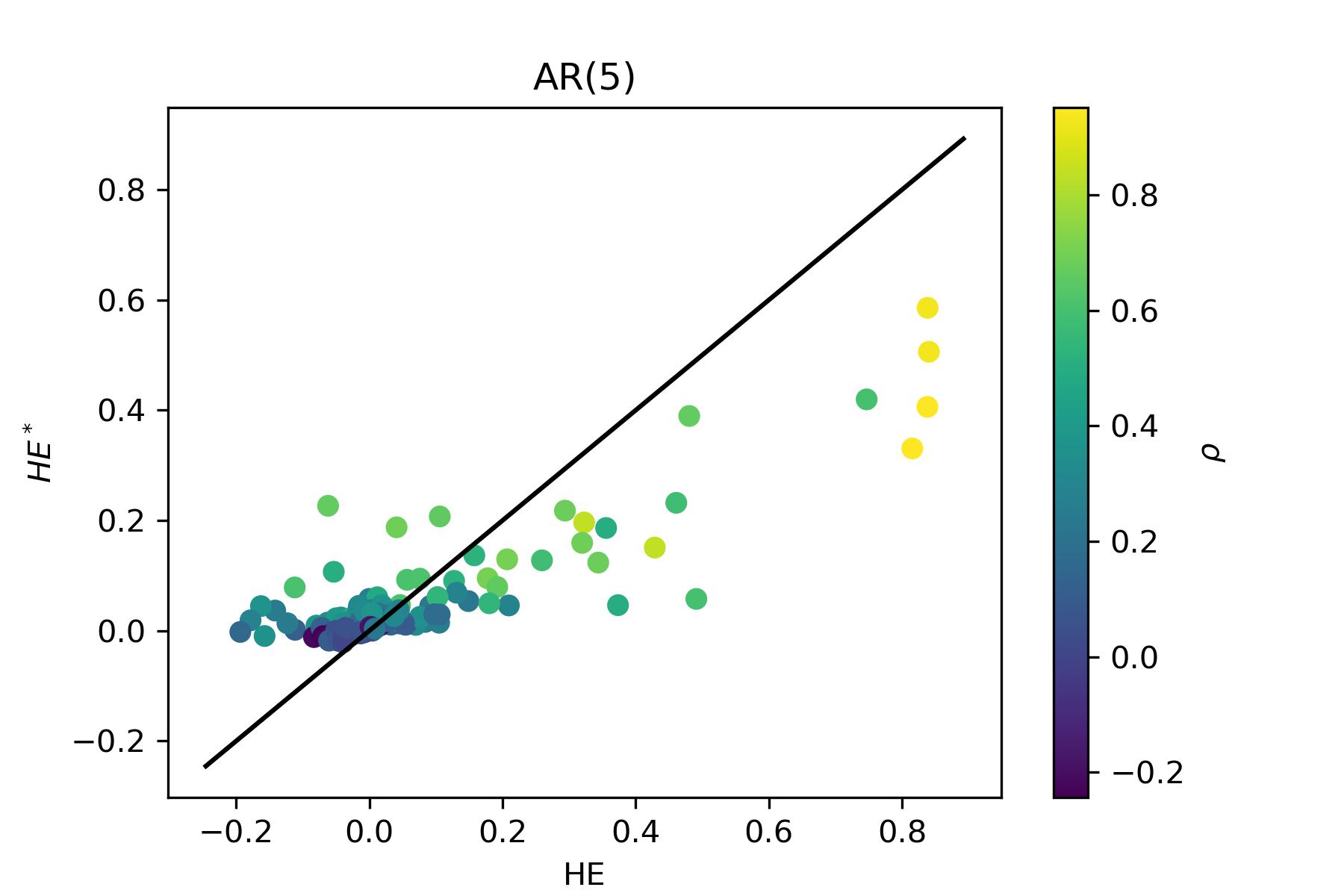}
\includegraphics[width=0.5\linewidth]{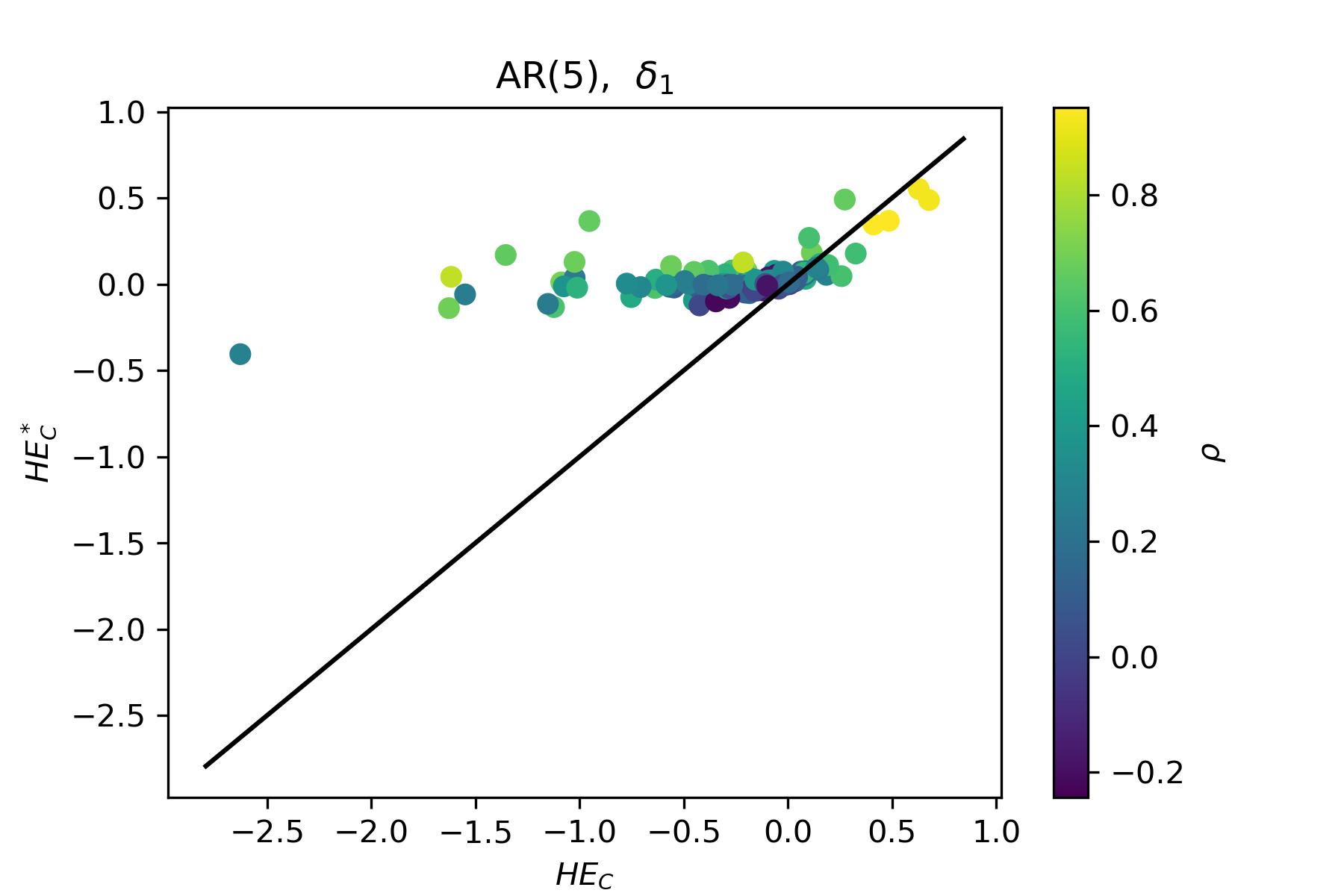}
\includegraphics[width=0.5\linewidth]{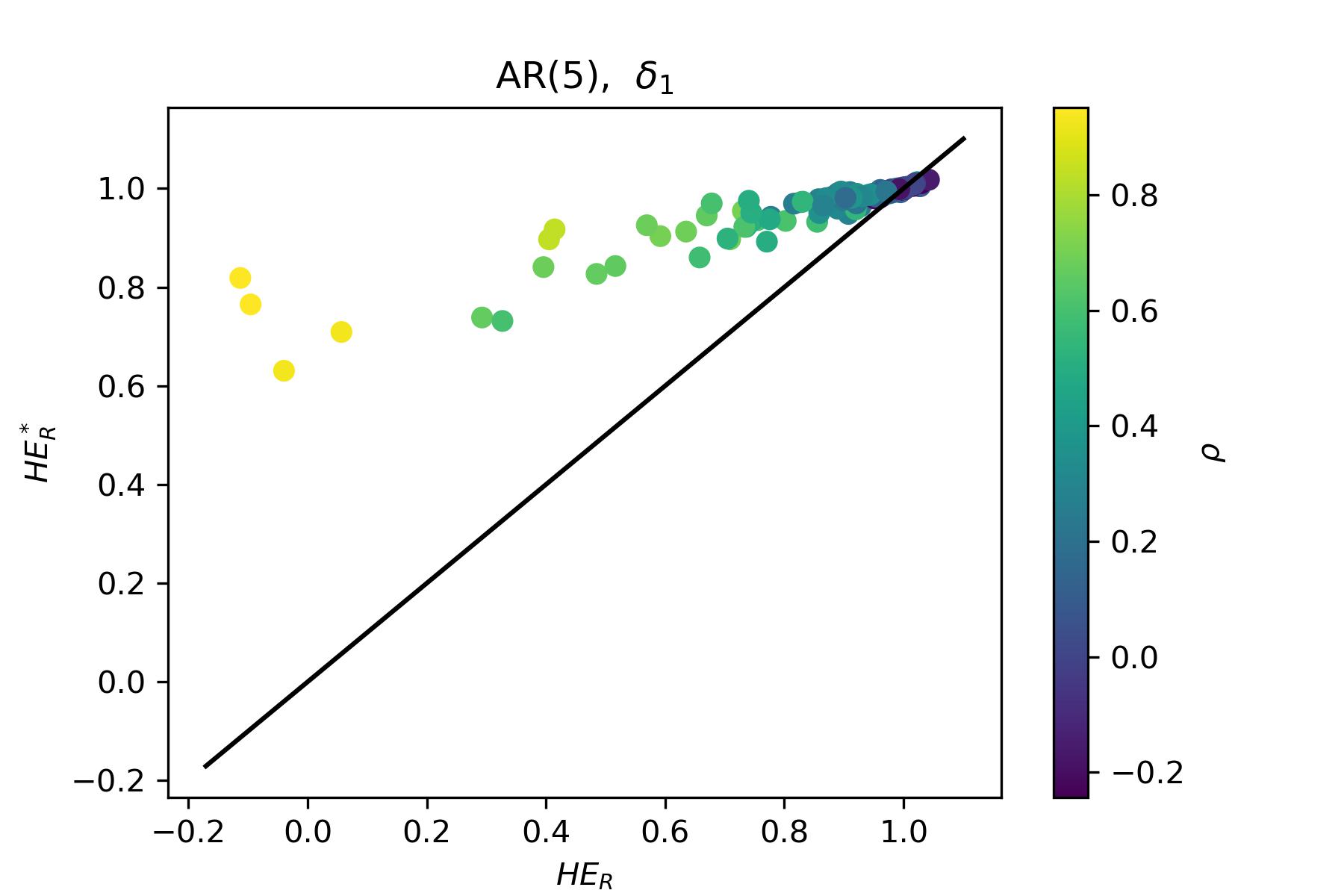}
    \caption{Comparison between the outputs of the standard and robust methodologies when realized variances and covariances are fitted by AR(5) processes. Each point is associated with a pair of instruments and its color is the return correlation $\rho$ of the pair. The threshold $\delta_1$ is the first quartile of the asset returns. The dark line corresponds to the bisector.}
    \label{fig:AR5}
\end{figure}

\section{HAR-type model for variance and covariance}\label{app:HAR}
In this appendix, we perform a robustness check by modeling variances and covariances with an HAR-type model \citep{corsi}. In particular, we consider
\begin{equation*}
    y_{t+1} = \phi_0 + \phi_1 y_t + \frac{\phi_2}{4}\sum_{i=1}^4 y_{t-i},
\end{equation*}
and the coefficients $\phi_j$ are estimated by means of an OLS regression. We notice that this HAR-type specification is analogous to the AR(5) model with parameters $\phi_0, \phi_1, \frac{\phi_2}{4}, \frac{\phi_2}{4}, \frac{\phi_2}{4}, ,\frac{\phi_2}{4}$. In Fig. \ref{fig:RMSE_HAR_AR} we show the Root Mean Squared Error (RMSE) that we obtain out-of-sample when realized variances are predicted with the HAR-type, the AR(5), and the AR(1) models. As expected, the latter is outperformed by the two other models. On the other hand, the HAR-type model performs slightly better than the AR(5).

Given the volatility predictions, we run our robust methodology for hedging and in Fig. \ref{fig:HAR_1}-\ref{fig:HAR_2}, a comparison between the results that we obtain by relying on the HAR-type and AR(5) models are provided. We observe that when $\tau =1$, the HAR-type specification is associated to a lower standard deviation of the hedge ratio than the AR(5). Similarly to what happens when we compare AR(1) and AR(5) in Subsection \ref{subsec:AR1_vs_AR5}, the HAR-type model adapts less effectively to shocks than the AR(5) model. Indeed, the latter weights differently observations at lower lags while the former considers their average. However, the effectiveness metrics are approximately equivalent for the two models.

\begin{figure}
    \centering
    \includegraphics[width=0.5\linewidth]{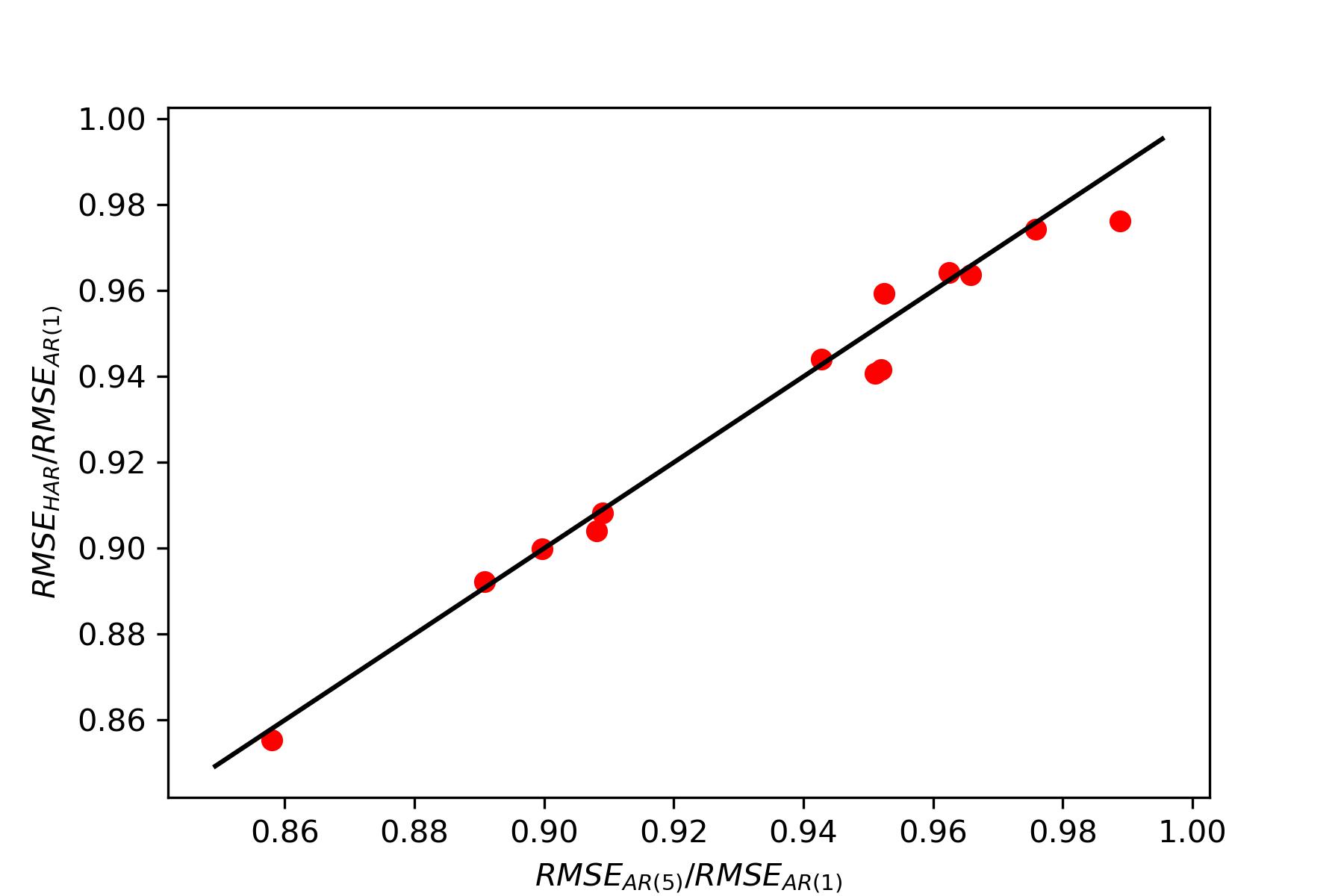}
    \caption{Comparison between the Root Mean Squared Error (RMSE) that we obtain out-of-sample when realized variances are predicted with the HAR-type and the AR(5) models. Values are normalized with the corresponding RMSEs obtained with the AR(1) model. The dark line corresponds to the bisector.}
    \label{fig:RMSE_HAR_AR}
\end{figure}

\begin{figure}
\includegraphics[width=0.5\linewidth]{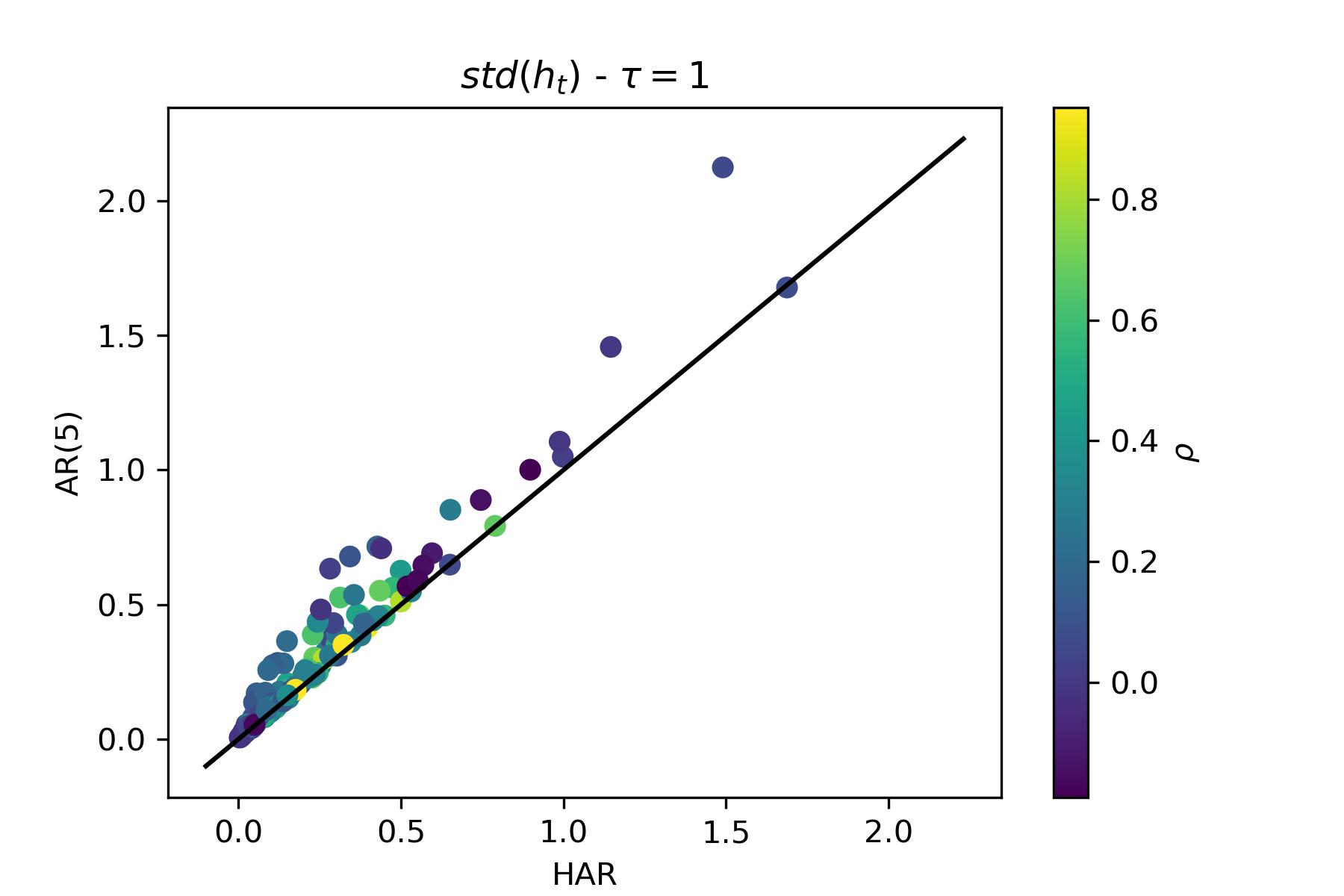}
\includegraphics[width=0.5\linewidth]{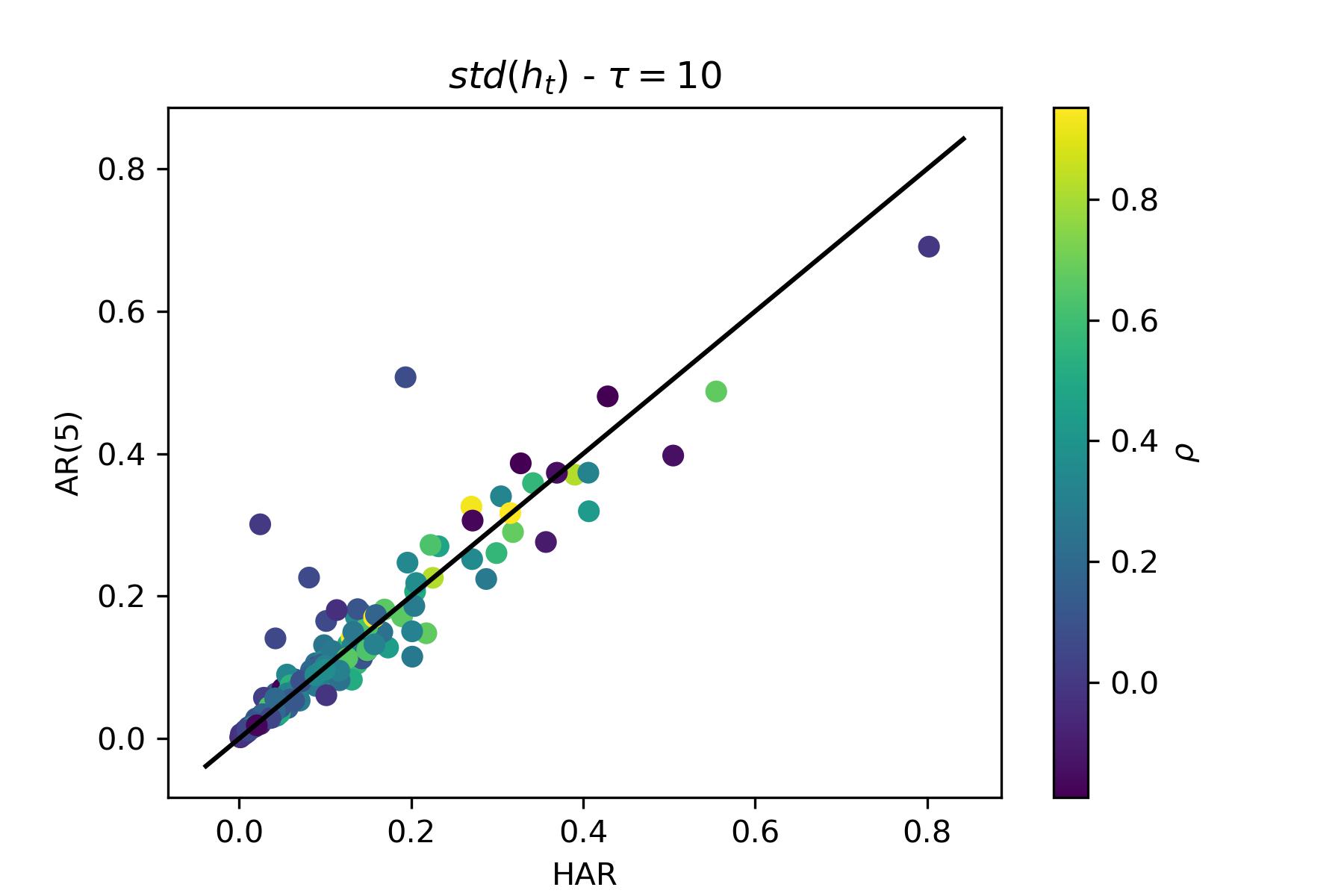}
\includegraphics[width=0.5\linewidth]{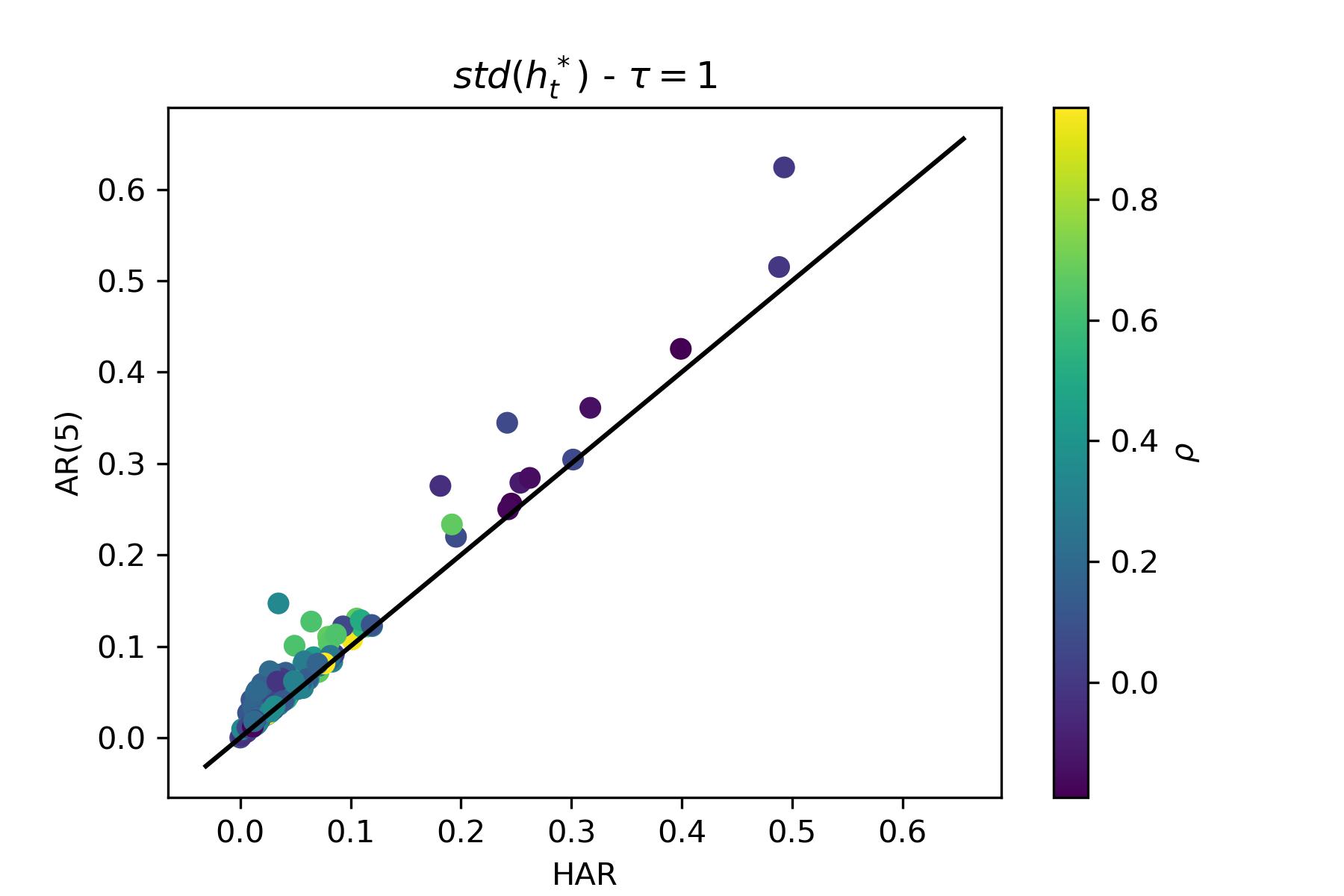}
\includegraphics[width=0.5\linewidth]{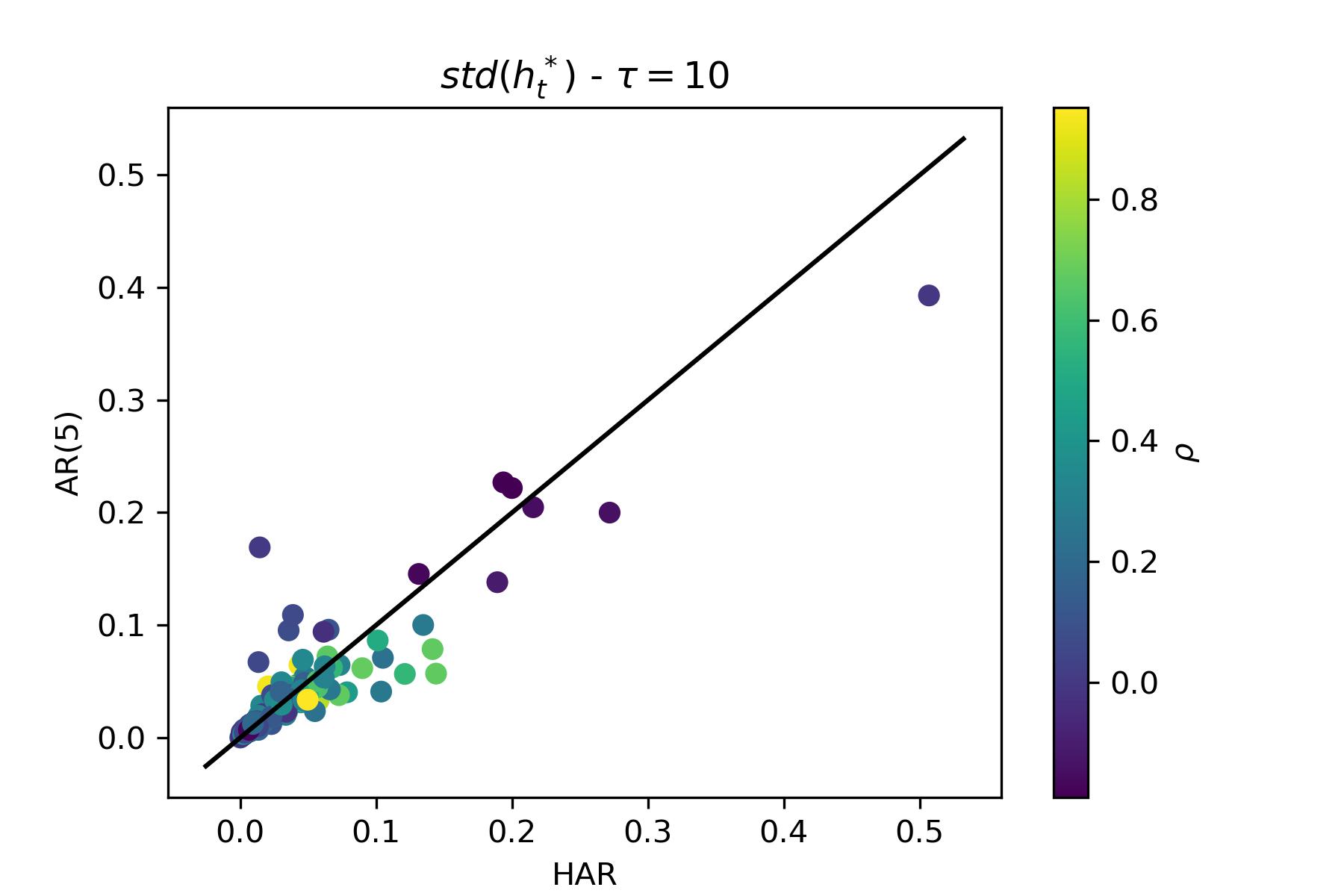}
\includegraphics[width=0.5\linewidth]{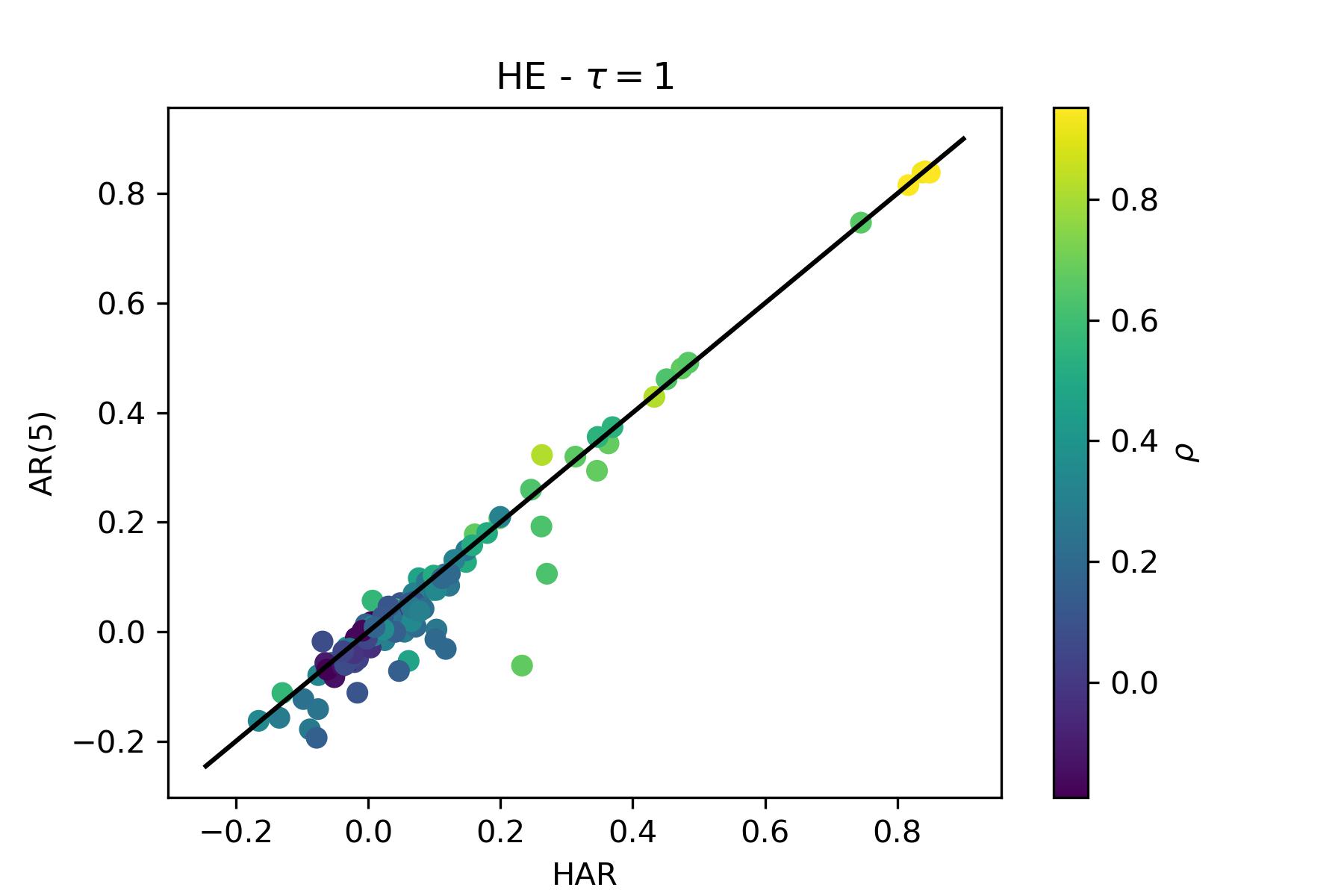}
\includegraphics[width=0.5\linewidth]{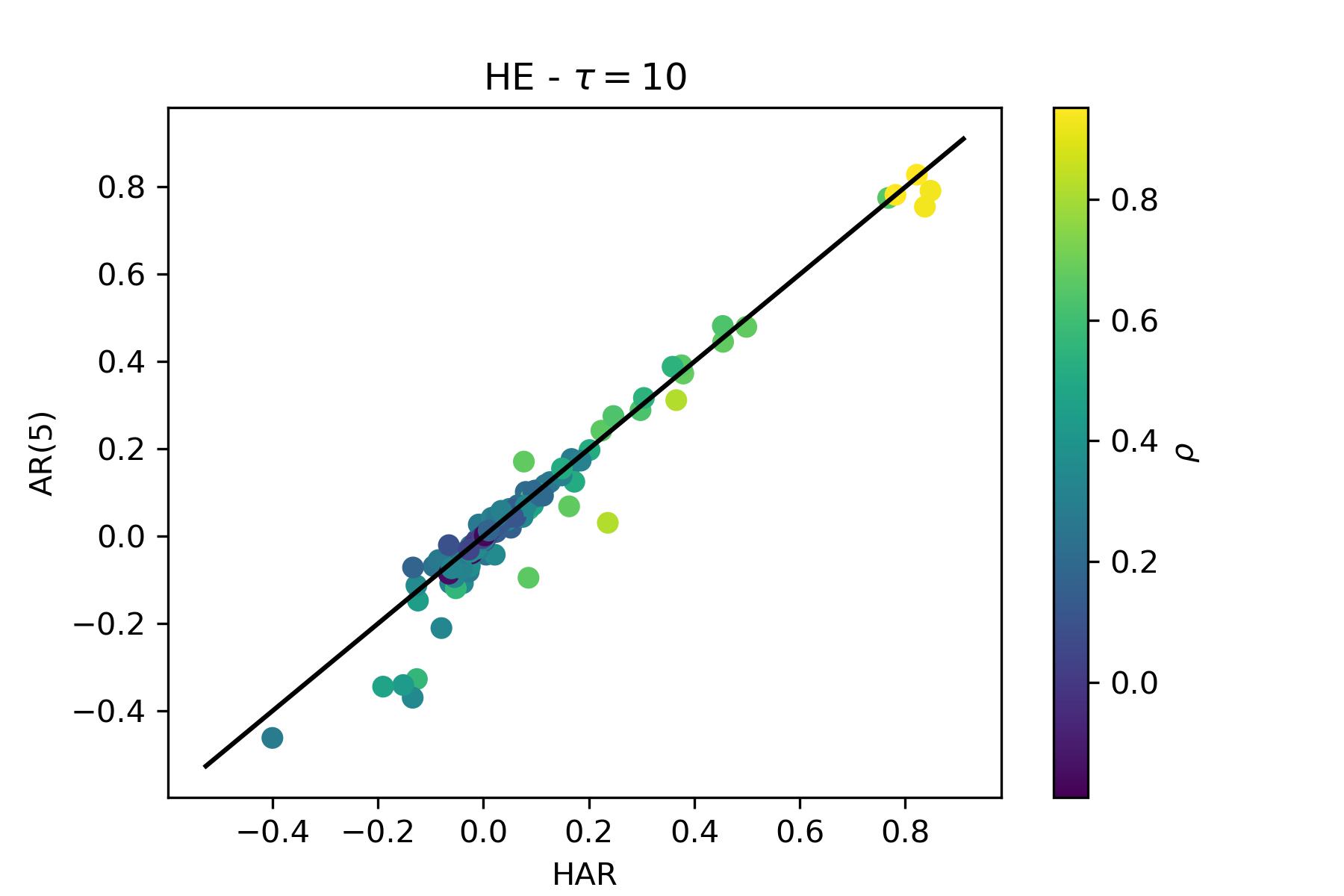}
\includegraphics[width=0.5\linewidth]{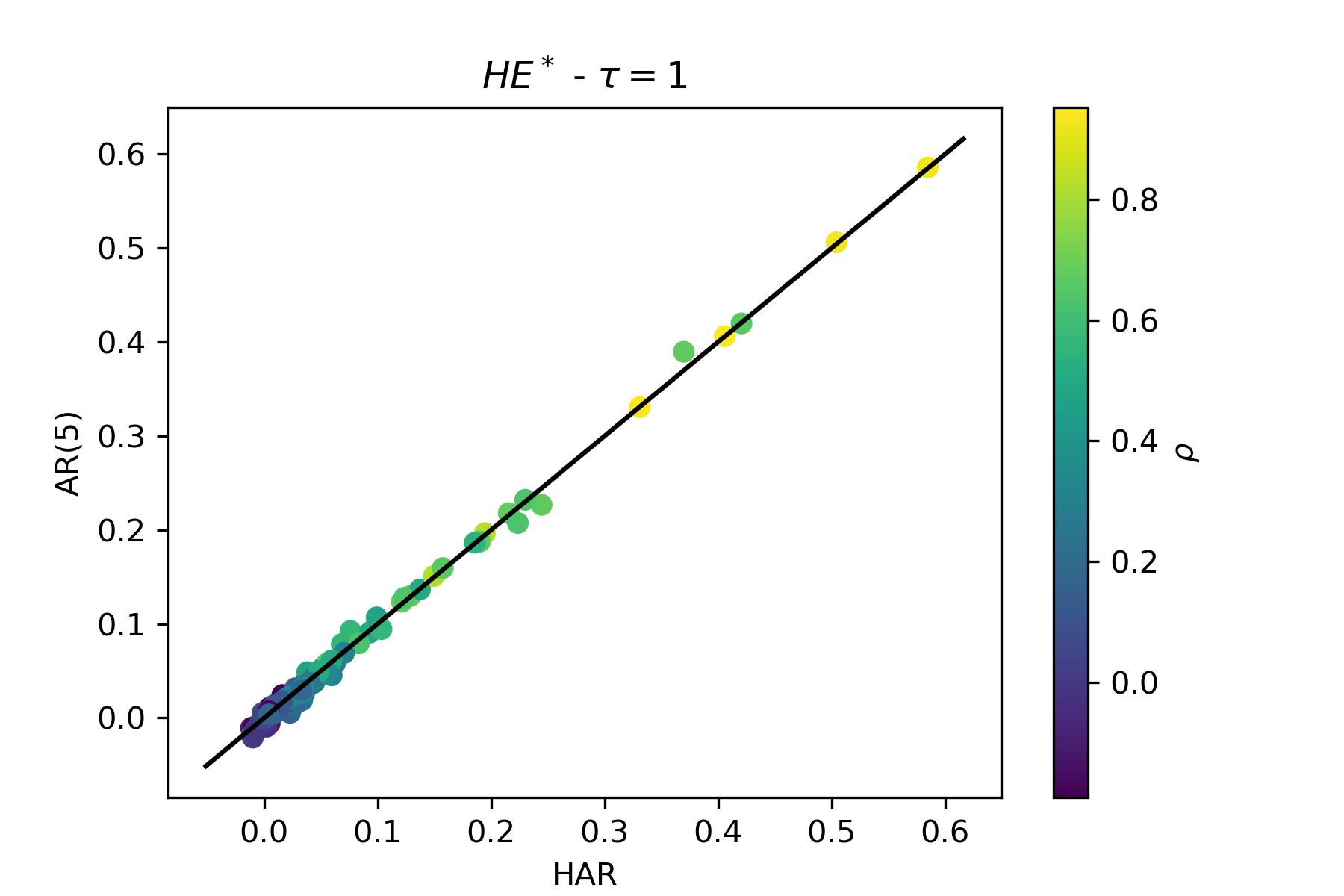}
\includegraphics[width=0.5\linewidth]{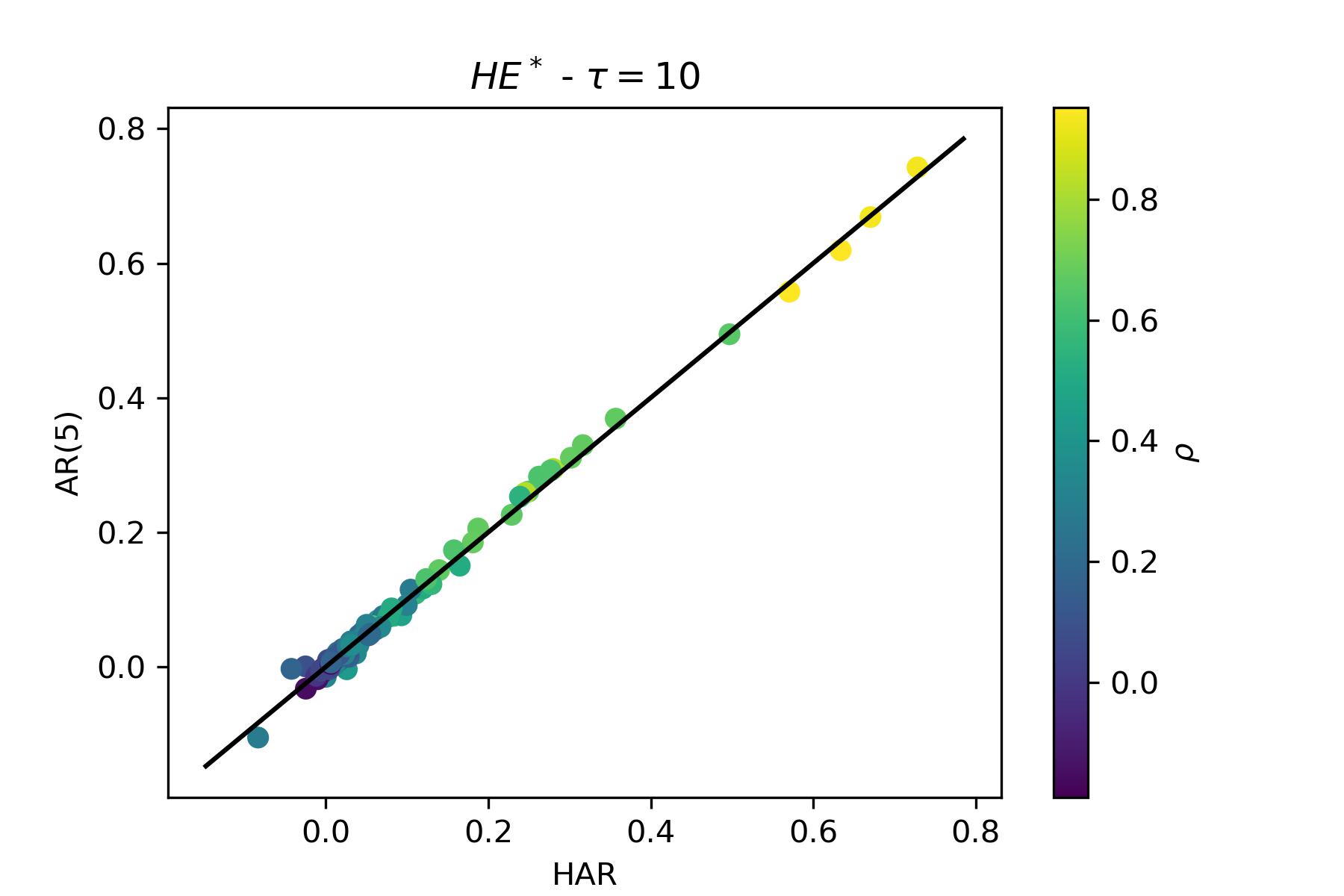}
    \caption{Comparison between the results that are obtained when realized variances and covariances are fitted by an HAR-type model and an AR(5) process. Two different prediction horizons are considered, i.e., $\tau = 1$ and $\tau = 10$. Each point is associated with a pair of instruments and its color is the return correlation $\rho$ of the pair. The dark line corresponds to the bisector.}
    \label{fig:HAR_1}
\end{figure}

\begin{figure}
\includegraphics[width=0.5\linewidth]{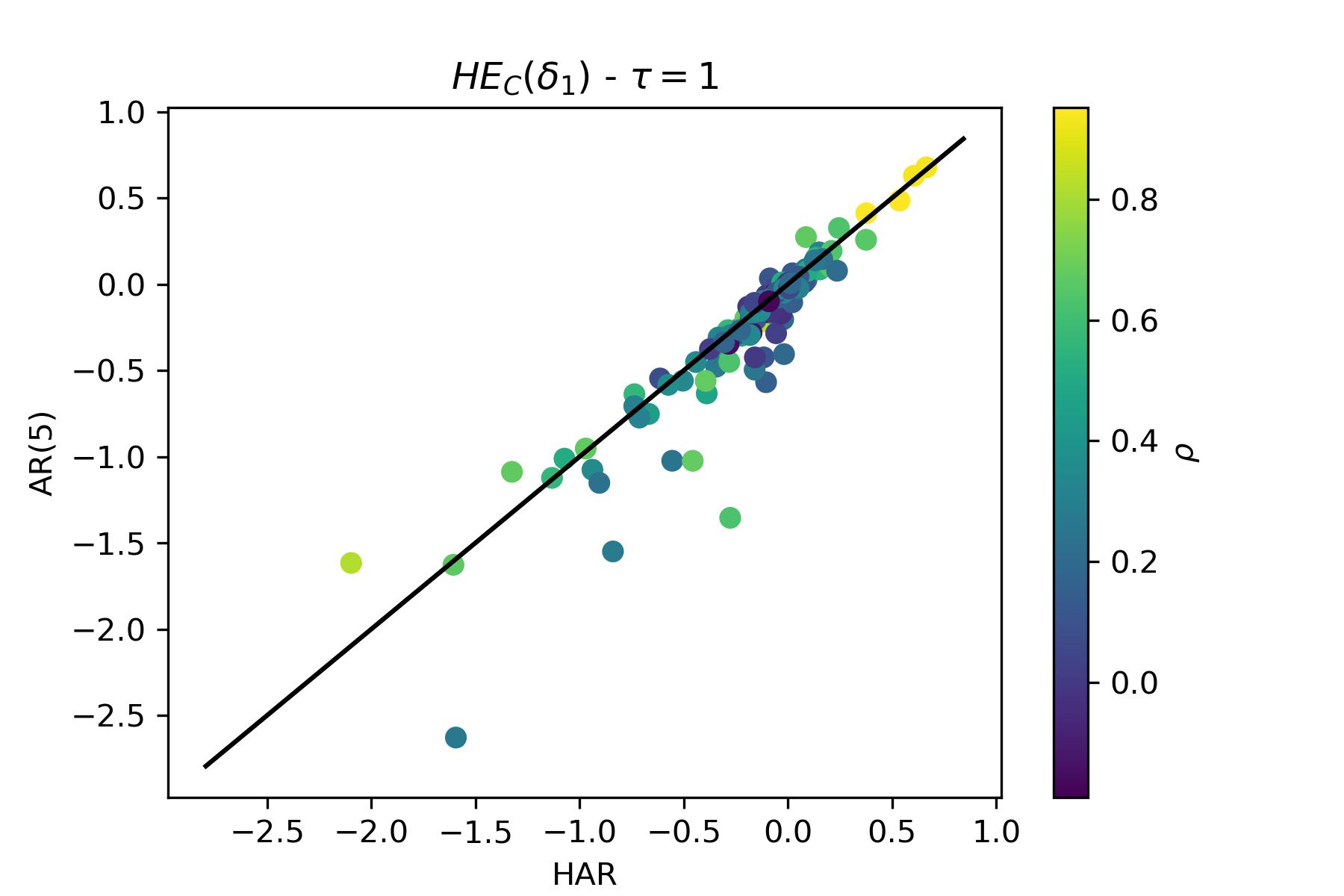}
\includegraphics[width=0.5\linewidth]{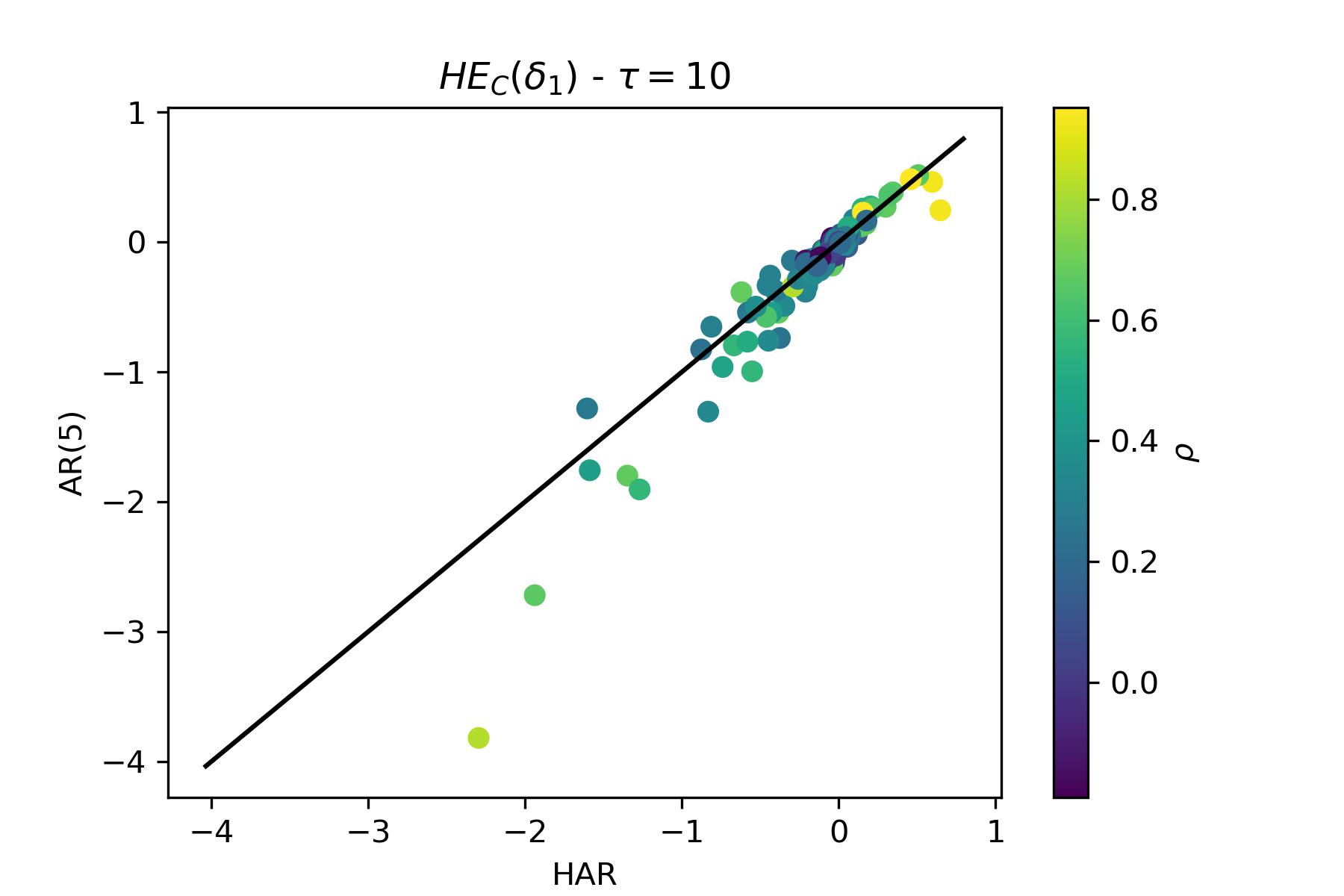}
\includegraphics[width=0.5\linewidth]{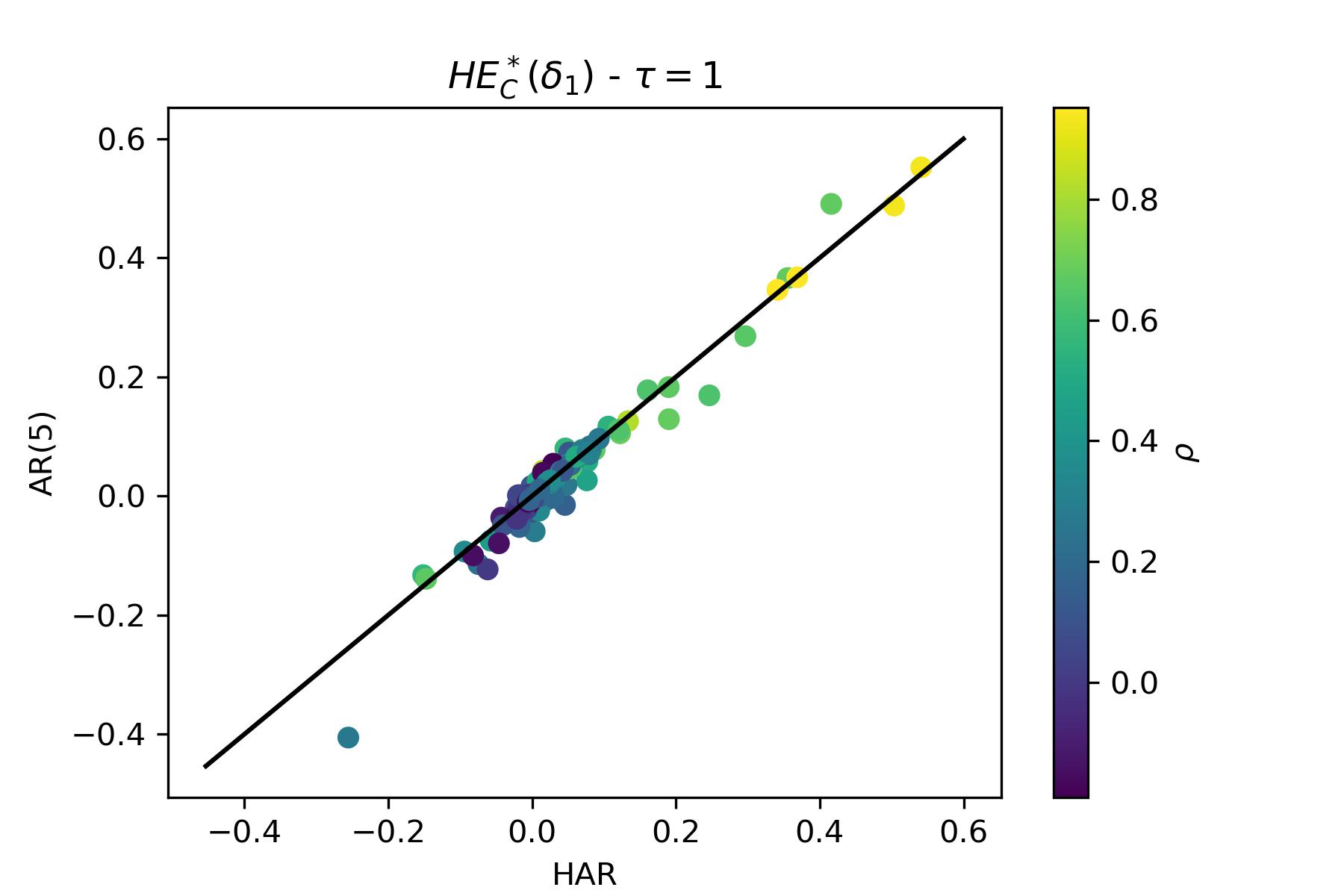}
\includegraphics[width=0.5\linewidth]{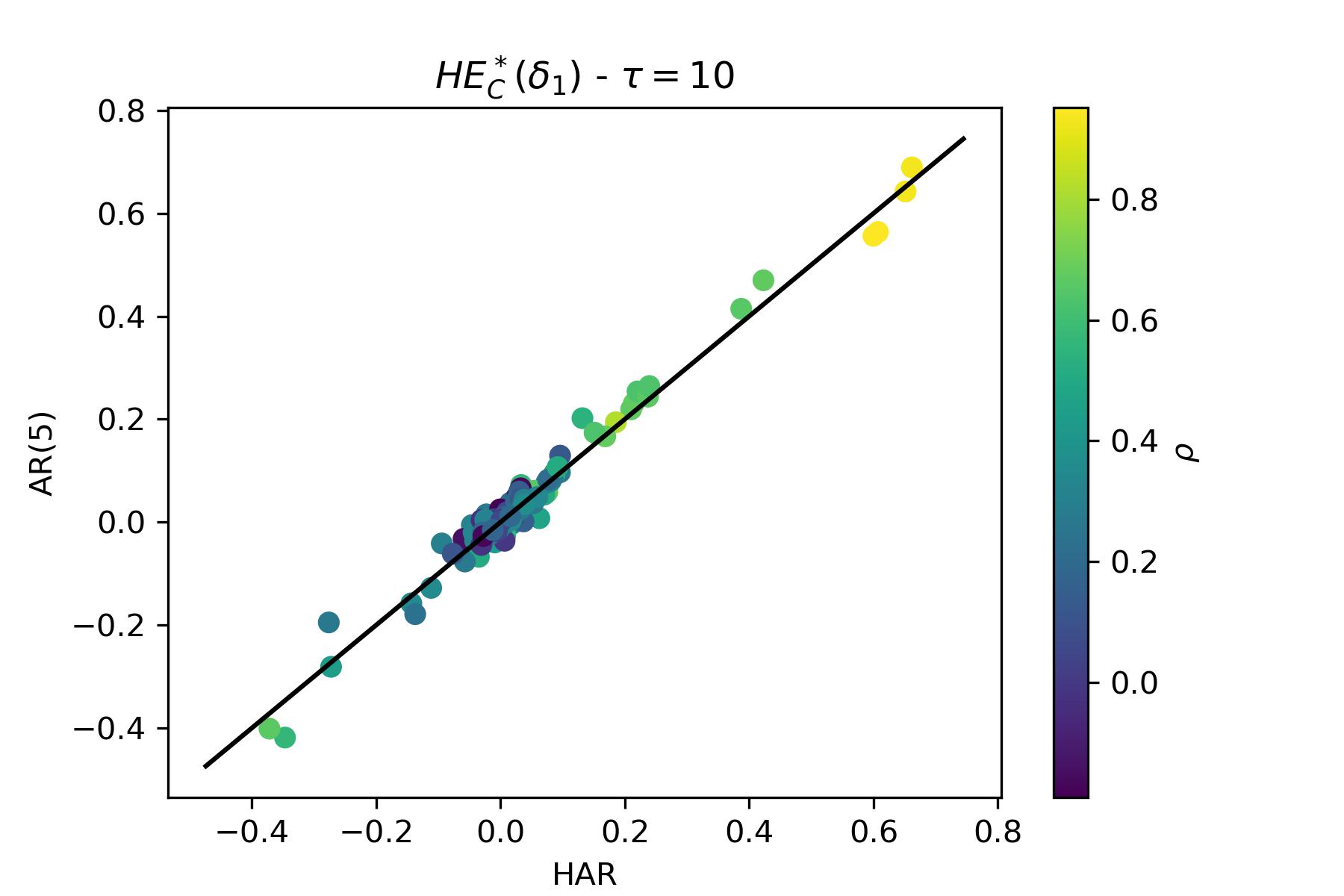}
\includegraphics[width=0.5\linewidth]{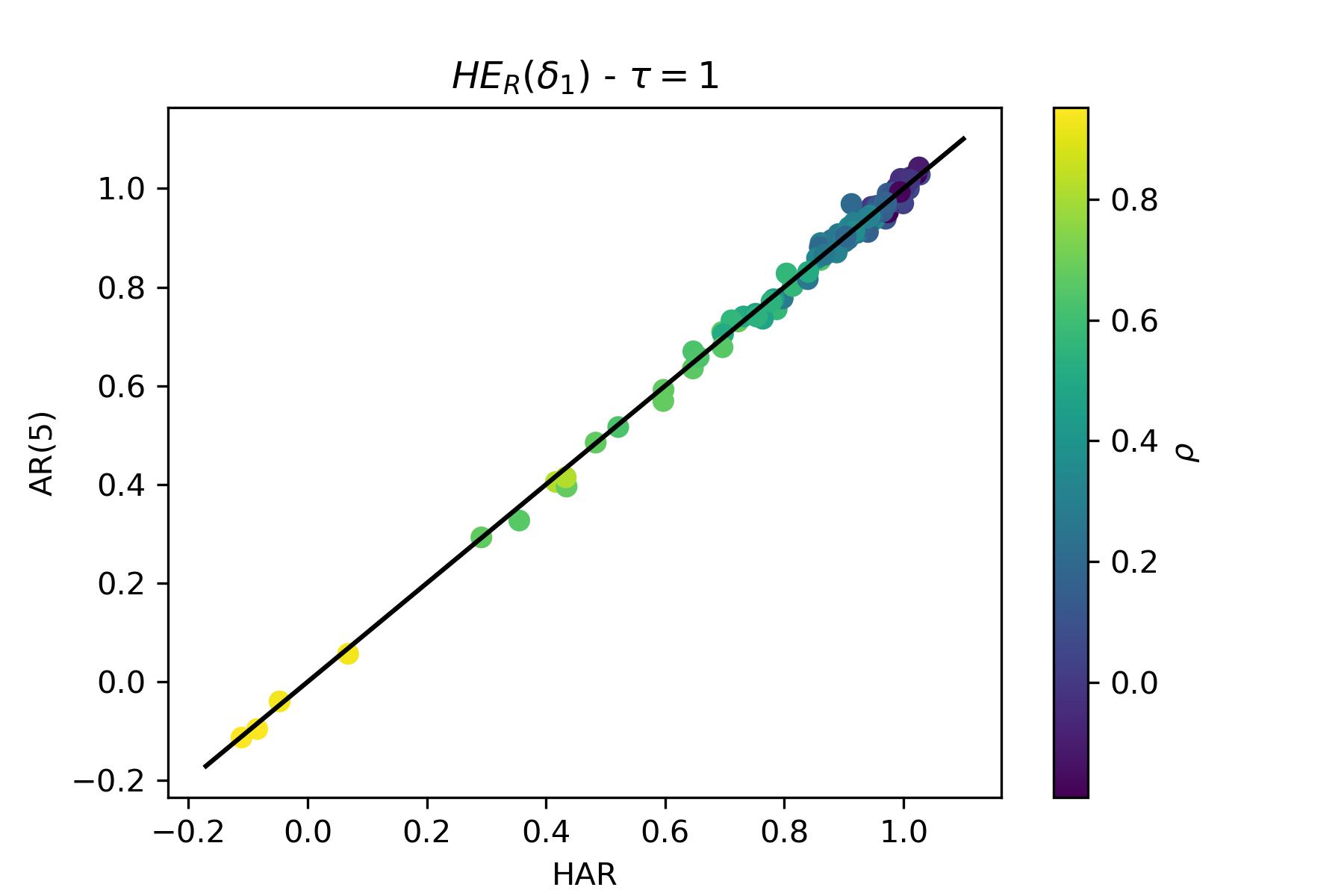}
\includegraphics[width=0.5\linewidth]{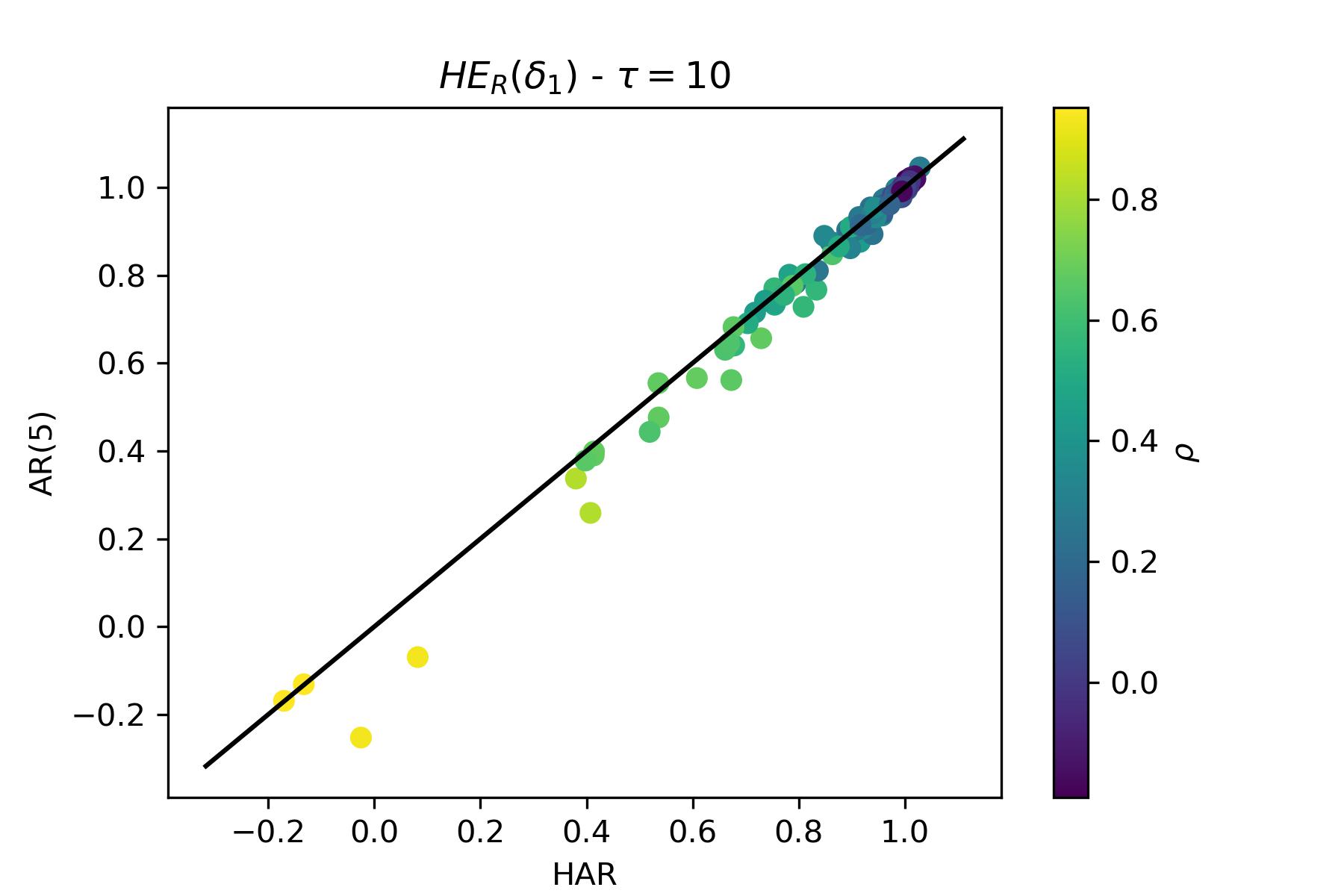}
\includegraphics[width=0.5\linewidth]{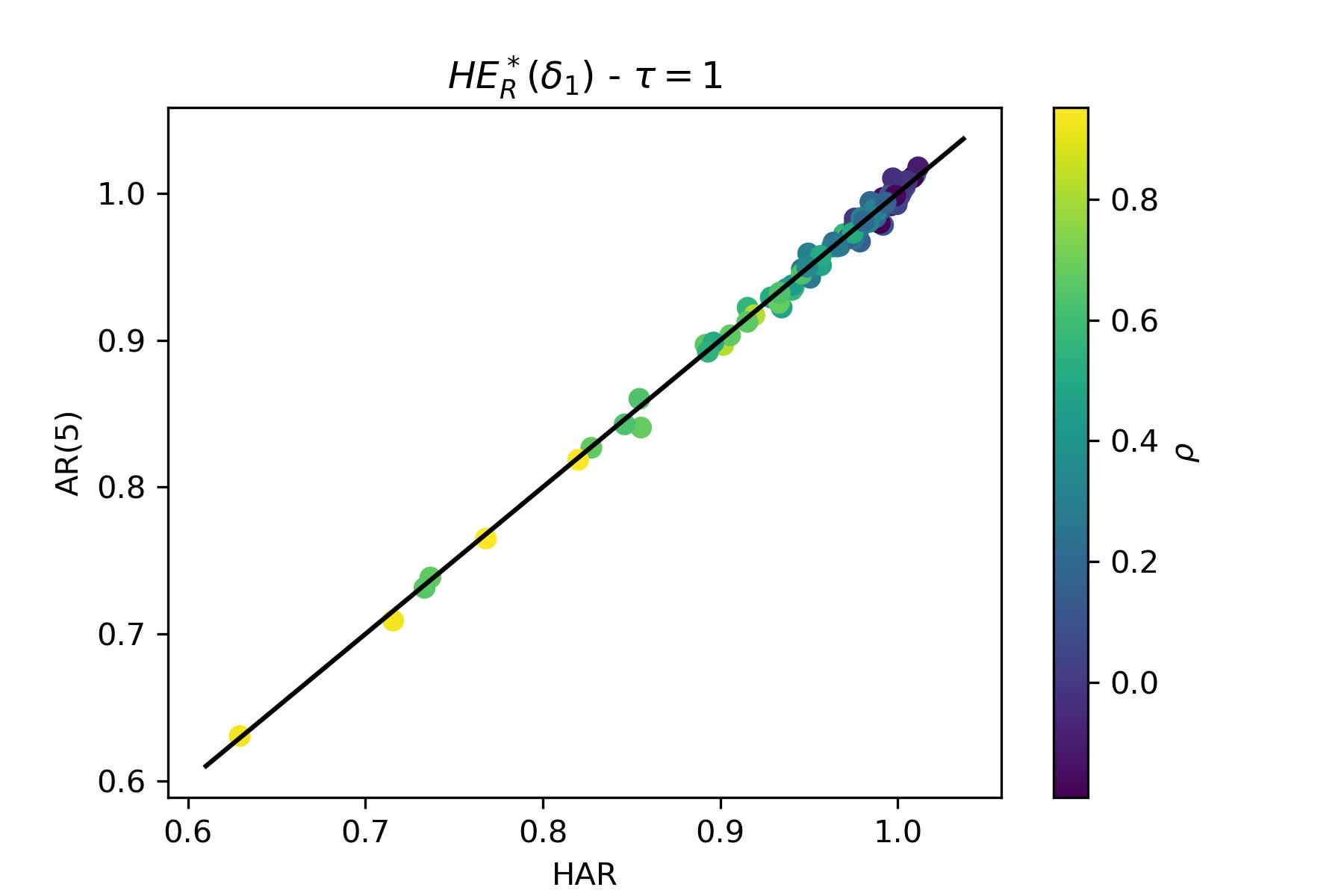}
\includegraphics[width=0.5\linewidth]{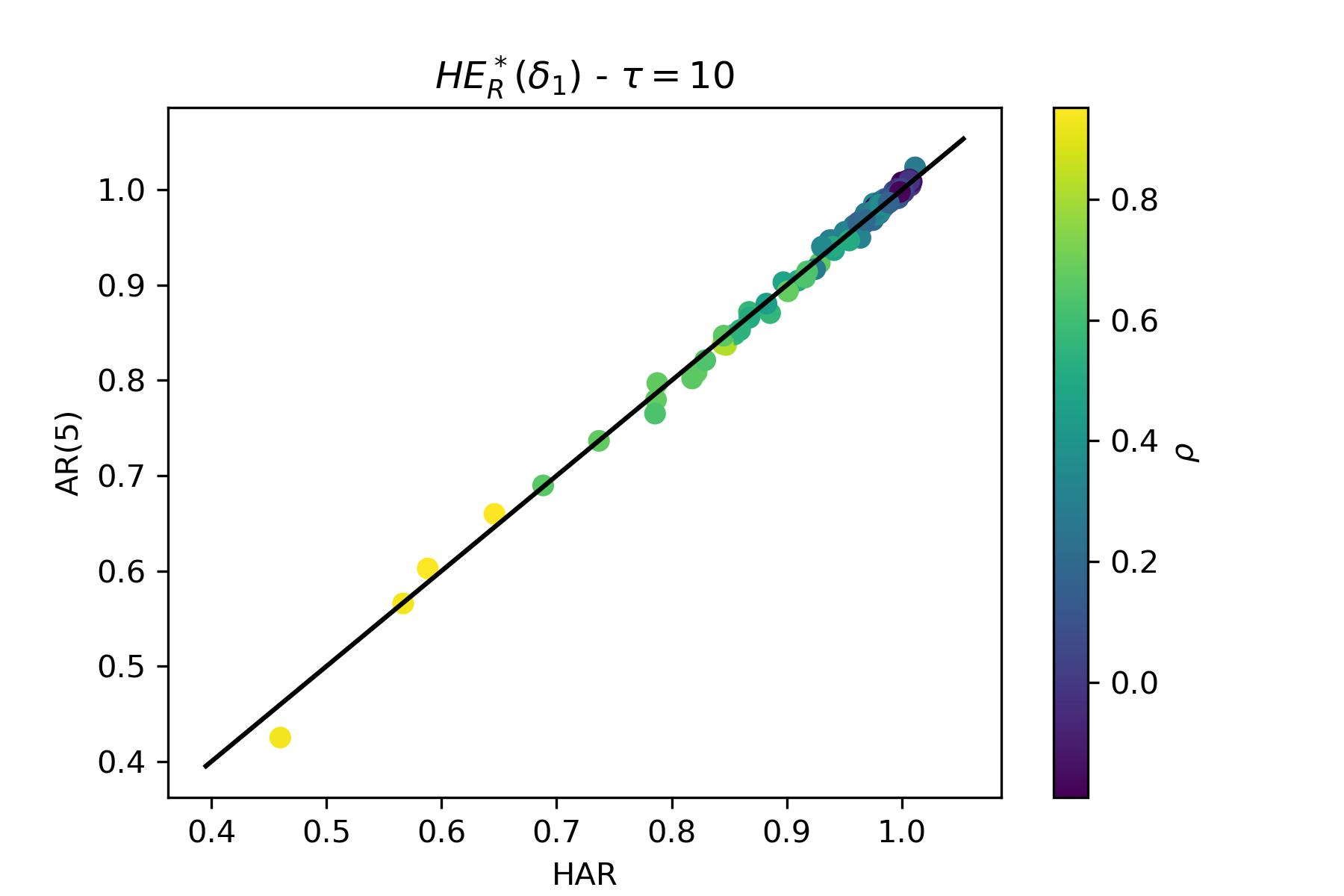}
    \caption{Comparison between the results that are obtained when realized variances and covariances are fitted by an HAR-type model and an AR(5) process. Two different prediction horizons are considered, i.e., $\tau = 1$ and $\tau = 10$. Each point is associated with a pair of instruments and its color is the return correlation $\rho$ of the pair. The dark line corresponds to the bisector.}
    \label{fig:HAR_2}
\end{figure}

\section{Performance and risk-adjusted metrics}
\label{app: performance}

In this appendix, we report Tables \ref{tab: transaction_costs_aaau_no_return}–\ref{tab:transaction_costs_ung}, which provide a detailed breakdown of the scatter plot results presented in Section \ref{performance}. For each hedged asset, the tables display the differences in performance metrics obtained under the robust and standard hedge ratios.

\begin{sidewaystable}
\begin{table}[H] \centering 
\scalebox{0.5}{
\begin{tabular}{@{\extracolsep{5pt}} ccccccccccccccccccccccccccc} 
\\[-1.8ex]\hline 
\hline \\[-1.8ex] 
Instrument & P\&L & P\&L(5bp) & P\&L(10bp) & SR & SR(5bp) & SR(10bp) & $\Omega$ & $\Omega$(5bp) & $\Omega$(10bp) & DD & DD(5bp) & DD(10bp) & VaR & VaR(5bp) & VaR(10bp) & ES & ES(5bp) & ES(10bp) \\ 
\hline \\[-1.8ex] 
AAAU/ASHR & 0.398 & 0.840 & 1.222 & 0.010 & 0.044 & 0.180 & -0.556 & -0.020 & 0.451 &
-0.149 & -0.089 & -0.029 & -0.035 & -0.035 & -0.035 & -0.058 & -0.056 & -0.054 \\
AAAU/BNO  & 1.545 & 2.002 & 2.442 & 0.216 & 0.370 & 0.524 & 0.767 & 1.302 & 1.828 &
-0.252 & -0.238 & -0.224 & 0.040 & 0.040 & 0.042 & -0.057 & -0.056 & -0.055 \\
AAAU/CORP & 7.589 & 13.601 & 19.706 & 2.600 & 4.658 & 6.743 & 8.839 & 15.524 & 21.863 &
-0.320 & -0.148 & 0.491 & 0.059 & 0.067 & 0.079 & -0.022 & 0.001 & 0.024 \\
AAAU/EWH  & 0.788 & 1.195 & 1.541 & 0.011 & 0.153 & 0.276 & -0.002 & 0.496 & 0.922 &
-0.282 & -0.269 & -0.257 & 0.008 & 0.009 & 0.009 & -0.068 & -0.066 & -0.064 \\
AAAU/GOVT & 2.235 & 6.597 & 10.698 & 0.577 & 2.137 & 3.601 & 2.097 & 7.287 & 11.871 &
-0.483 & -0.370 & -0.257 & 0.010 & 0.013 & 0.020 & -0.085 & -0.067 & -0.049 \\
AAAU/GSG  & 2.253 & 3.153 & 3.996 & 0.230 & 0.551 & 0.859 & 0.818 & 1.916 & 2.954 &
-0.425 & -0.401 & -0.377 & 0.018 & 0.018 & 0.019 & -0.102 & -0.099 & -0.097 \\
AAAU/ICLN & -1.996 & -1.199 & -0.453 & -0.982 & -0.715 & -0.451 & -3.400 & -2.470 & -1.519 &
-0.211 & -0.175 & -0.139 & -0.091 & -0.091 & -0.090 & -0.060 & -0.058 & -0.055 \\
AAAU/IEV  & 5.053 & 6.671 & 8.096 & 1.328 & 1.905 & 2.415 & 4.367 & 6.335 & 8.007 &
-0.235 & -0.197 & -0.160 & -0.097 & -0.083 & -0.074 & -0.100 & -0.096 & -0.092 \\
AAAU/IGOV & 1.925 & 3.619 & 5.110 & 0.444 & 1.054 & 1.590 & 1.505 & 3.610 & 5.396 &
-0.269 & -0.213 & -0.157 & 0.012 & 0.018 & 0.021 & -0.030 & -0.026 & -0.021 \\
AAAU/IVV  & 6.954 & 8.305 & 9.689 & 2.181 & 2.653 & 3.136 & 7.281 & 8.837 & 10.365 &
-0.019 & 0.038 & 0.095 & -0.031 & -0.018 & -0.007 & -0.066 & -0.061 & -0.056 \\
AAAU/QQQM & 4.338 & 4.854 & 5.588 & 1.842 & 2.169 & 2.550 & 6.215 & 7.105 & 8.283 &
-0.244 & -0.200 & -0.156 & -0.088 & -0.087 & -0.082 & -0.083 & -0.079 & -0.074 \\
AAAU/UNG  & -0.507 & -0.446 & -0.378 & -0.147 & -0.127 & -0.104 & -0.544 & -0.478 & -0.397 &
-0.035 & -0.032 & -0.029 & -0.007 & -0.007 & -0.006 & 0.006 & 0.005 & 0.005 \\
\hline \\[-1.8ex] 
\end{tabular}} 
  \caption{Hedged asset: Gold. We compute the difference, multiplied by 100, of various performance metrics obtained using the robust Hedge Ratios relative to those obtained using the standard Hedge Ratios. P\&L, SR, $\Omega$, DD, VaR, ES correspond to profit and loss, Sharpe ratio, Omega ratio, maximum drawdown, 95\% Value at Risk, and 95\% Expected Shortfall, respectively. Transaction costs are indicated in curly brackets.} 
  \label{tab: transaction_costs_aaau_no_return} 
\end{table}
\end{sidewaystable}

\begin{sidewaystable}
\begin{table}[H] \centering 
\scalebox{0.5}{
\begin{tabular}{@{\extracolsep{5pt}} ccccccccccccccccccccccccccc} 
\\[-1.8ex]\hline 
\hline \\[-1.8ex] 
Instrument & P\&L & P\&L(5bp) & P\&L(10bp) & SR & SR(5bp) & SR(10bp) & $\Omega$ & $\Omega$(5bp) & $\Omega$(10bp) & DD & DD(5bp) & DD(10bp) & VaR & VaR(5bp) & VaR(10bp) & ES & ES(5bp) & ES(10bp) \\ 
\hline \\[-1.8ex] 
ASHR/AAAU & 3.941 & 4.721 & 5.407 & 0.666 & 0.801 & 0.920 & 2.321 & 2.801 & 3.212 &
0.012 & 0.031 & 0.125 & 0.009 & 0.010 & 0.012 & -0.043 & -0.039 & -0.035 \\
ASHR/BNO  & 7.836 & 9.710 & 11.602 & 1.184 & 1.467 & 1.753 & 3.543 & 4.363 & 5.179 &
1.529 & 2.003 & 2.473 & -0.006 & -0.001 & 0.003 & -0.192 & -0.185 & -0.179 \\
ASHR/CORP & 2.255 & 5.327 & 8.349 & 0.300 & 0.709 & 1.111 & 0.968 & 2.265 & 3.521 &
0.632 & 1.521 & 2.717 & 0.013 & 0.015 & 0.018 & 0.000 & 0.008 & 0.015 \\
ASHR/EWH  & 0.384 & 4.217 & 8.071 & 0.164 & 0.701 & 1.242 & -0.038 & 2.584 & 5.100 &
-4.253 & -1.289 & 1.542 & -0.098 & -0.091 & -0.082 & -0.136 & -0.118 & -0.099 \\
ASHR/GOVT & -1.796 & 4.307 & 10.520 & -0.249 & 0.616 & 1.490 & -0.823 & 1.991 & 4.657 &
-0.860 & -0.181 & 1.354 & -0.041 & -0.024 & -0.008 & -0.071 & -0.033 & 0.006 \\
ASHR/GSG  & 7.456 & 10.039 & 12.651 & 1.123 & 1.513 & 1.907 & 3.374 & 4.501 & 5.618 &
1.501 & 2.227 & 2.942 & -0.002 & 0.012 & 0.022 & -0.157 & -0.148 & -0.139 \\
ASHR/ICLN & -13.147 & -9.661 & -6.234 & -1.935 & -1.427 & -0.926 & -6.716 & -4.893 & -3.139 &
-2.425 & -2.379 & -2.333 & -0.200 & -0.191 & -0.183 & -0.181 & -0.173 & -0.165 \\
ASHR/IEV  & 18.128 & 24.665 & 31.783 & 2.614 & 3.560 & 4.587 & 8.262 & 11.042 & 13.927 &
8.853 & 12.485 & 15.890 & -0.037 & -0.032 & -0.029 & -0.125 & -0.105 & -0.084 \\
ASHR/IGOV & -0.098 & 0.734 & 1.564 & -0.012 & 0.100 & 0.212 & -0.034 & 0.329 & 0.689 &
-0.443 & -0.100 & 0.241 & -0.017 & -0.016 & -0.015 & -0.021 & -0.019 & -0.017 \\
ASHR/IVV  & 48.212 & 58.874 & 70.199 & 6.664 & 8.125 & 9.657 & 18.583 & 22.102 & 25.531 &
23.605 & 27.653 & 31.300 & 0.132 & 0.163 & 0.193 & 0.138 & 0.172 & 0.208 \\
ASHR/QQQM & 10.641 & 11.987 & 13.904 & 2.409 & 2.712 & 3.138 & 8.930 & 9.832 & 11.211 &
1.912 & 2.404 & 2.892 & 0.169 & 0.180 & 0.189 & 0.064 & 0.081 & 0.097 \\
ASHR/UNG  & -0.525 & -0.380 & -0.206 & -0.067 & -0.048 & -0.024 & -0.208 & -0.147 & -0.073 &
-0.121 & -0.069 & -0.017 & -0.066 & -0.066 & -0.066 & -0.009 & -0.008 & -0.008 \\
\hline \\[-1.8ex] 
\end{tabular}} 
  \caption{Hedged asset: CSI 300. We compute the difference, multiplied by 100, of various performance metrics obtained using the robust Hedge Ratios relative to those obtained using the standard Hedge Ratios. P\&L, SR, $\Omega$, DD, VaR, ES correspond to profit and loss, Sharpe ratio, Omega ratio, maximum drawdown, 95\% Value at Risk, and 95\% Expected Shortfall, respectively. Transaction costs are indicated in curly brackets.} 
  \label{tab: transaction_costs_ashr_no_return} 
\end{table}
\end{sidewaystable}

\begin{sidewaystable}
\begin{table}[H] \centering 
\scalebox{0.5}{
\begin{tabular}{@{\extracolsep{5pt}} ccccccccccccccccccccccccccc} 
\\[-1.8ex]\hline 
\hline \\[-1.8ex] 
Instrument & P\&L & P\&L(5bp) & P\&L(10bp) & SR & SR(5bp) & SR(10bp) & $\Omega$ & $\Omega$(5bp) & $\Omega$(10bp) & DD & DD(5bp) & DD(10bp) & VaR & VaR(5bp) & VaR(10bp) & ES & ES(5bp) & ES(10bp) \\
\midrule 
BNO/AAAU & 7.343 & 8.434 & 9.374 & 1.433 & 1.645 & 1.828 & 3.584 & 4.098 & 4.539 &
2.865 & 3.463 & 4.057 & -0.156 & -0.153 & -0.151 & -0.003 & 0.000 & 0.003 \\
BNO/ASHR & -16.740 & -13.116 & -9.508 & -1.649 & -1.225 & -0.799 & -6.695 & -5.170 & -3.684 &
-1.294 & -1.107 & -0.921 & 0.216 & 0.240 & 0.252 & 0.022 & 0.032 & 0.041 \\
BNO/CORP & 1.222 & 3.182 & 5.085 & 0.191 & 0.440 & 0.682 & 0.524 & 1.220 & 1.892 &
0.517 & 0.688 & 0.858 & 0.042 & 0.047 & 0.051 & 0.060 & 0.066 & 0.072 \\
BNO/EWH & -3.537 & 0.916 & 5.400 & -0.424 & 0.168 & 0.766 & -1.459 & 0.307 & 2.035 &
0.795 & 2.020 & 4.361 & 0.078 & 0.109 & 0.139 & -0.038 & -0.029 & -0.021 \\
BNO/GOVT & -0.225 & 10.321 & 21.591 & -0.218 & 1.352 & 3.038 & -0.591 & 3.660 & 7.812 &
-0.491 & 0.886 & 5.827 & -0.244 & -0.199 & -0.142 & -0.246 & -0.220 & -0.193 \\
BNO/GSG & 34.491 & 44.006 & 53.692 & 8.227 & 11.712 & 15.213 & 20.990 & 28.503 & 35.406 &
2.061 & 7.431 & 14.599 & -1.032 & -0.999 & -0.955 & -1.442 & -1.426 & -1.409 \\
BNO/ICLN & -12.487 & -8.532 & -4.691 & -1.678 & -1.155 & -0.646 & -4.865 & -3.308 & -1.822 &
-0.724 & -0.508 & 0.646 & -0.003 & 0.007 & 0.019 & -0.112 & -0.105 & -0.098 \\
BNO/IEV & 28.368 & 36.482 & 45.600 & 3.716 & 4.786 & 5.983 & 9.910 & 12.561 & 15.435 &
14.599 & 17.922 & 21.066 & -0.062 & -0.045 & -0.014 & -0.028 & -0.009 & 0.010 \\
BNO/IGOV & -0.068 & 0.362 & 0.794 & -0.020 & 0.036 & 0.092 & -0.051 & 0.106 & 0.264 &
-0.101 & -0.050 & 0.000 & 0.001 & 0.001 & 0.002 & -0.025 & -0.024 & -0.022 \\
BNO/IVV & 63.646 & 76.336 & 90.450 & 8.074 & 9.668 & 11.415 & 19.843 & 23.248 & 26.770 &
29.598 & 36.073 & 41.543 & 0.141 & 0.169 & 0.217 & 0.278 & 0.313 & 0.346 \\
BNO/QQQM & 5.773 & 6.610 & 7.893 & 1.809 & 2.090 & 2.496 & 4.627 & 5.231 & 6.186 &
2.449 & 2.898 & 3.341 & -0.051 & -0.039 & -0.019 & -0.158 & -0.151 & -0.144 \\
BNO/UNG & -2.515 & -1.957 & -1.325 & -0.378 & -0.304 & -0.220 & -1.058 & -0.844 & -0.606 &
-1.013 & -0.881 & -0.749 & 0.032 & 0.026 & 0.027 & -0.088 & -0.087 & -0.086 \\
\hline \\[-1.8ex] 
\end{tabular}} 
  \caption{Hedged asset: Brent oil. We compute the difference, multiplied by 100, of various performance metrics obtained using the robust Hedge Ratios relative to those obtained using the standard Hedge Ratios. P\&L, SR, $\Omega$, DD, VaR, ES correspond to profit and loss, Sharpe ratio, Omega ratio, maximum drawdown, 95\% Value at Risk, and 95\% Expected Shortfall, respectively. Transaction costs are indicated in curly brackets.} 
  \label{tab: transaction_costs_bno_no_return} 
\end{table}
\end{sidewaystable}

\begin{sidewaystable}[htbp]
\centering
\scalebox{0.55}{
\begin{tabular}{@{\extracolsep{5pt}} ccccccccccccccccccccccccccc}
\toprule
Instrument & P\&L & P\&L(5bp) & P\&L(10bp) & SR & SR(5bp) & SR(10bp) & $\Omega$ & $\Omega$(5bp) & $\Omega$(10bp) & DD & DD(5bp) & DD(10bp) & VaR & VaR(5bp) & VaR(10bp) & ES & ES(5bp) & ES(10bp) \\
\midrule
CORP/AAAU & 5.260 & 6.319 & 7.281 & 4.434 & 5.314 & 6.111 & 11.623 & 13.673 & 15.517 & 2.361 & 3.112 & 3.856 & 0.030 & 0.033 & 0.036 & 0.022 & 0.025 & 0.028 \\
CORP/ASHR & -0.902 & -0.447 & 0.007 & -0.520 & -0.260 & 0.001 & -1.326 & -0.660 & 0.009 & -0.733 & -0.631 & -0.530 & -0.034 & -0.032 & -0.030 & -0.014 & -0.013 & -0.012 \\
CORP/BNO & 0.327 & 0.459 & 0.592 & 0.188 & 0.264 & 0.339 & 0.479 & 0.676 & 0.866 & 0.169 & 0.208 & 0.247 & 0.006 & 0.006 & 0.007 & 0.009 & 0.009 & 0.010 \\
CORP/EWH & 0.934 & 1.896 & 2.862 & 0.471 & 0.964 & 1.444 & 1.188 & 2.416 & 3.586 & 0.290 & 0.558 & 0.825 & 0.048 & 0.048 & 0.049 & 0.079 & 0.083 & 0.088 \\
CORP/GOVT & 0.210 & 4.475 & 8.675 & -0.003 & 3.951 & 7.878 & -0.154 & 10.555 & 19.892 & -2.271 & -1.039 & 2.497 & -0.155 & -0.147 & -0.132 & -0.113 & -0.099 & -0.085 \\
CORP/GSG & 1.199 & 1.566 & 1.934 & 0.689 & 0.901 & 1.113 & 1.726 & 2.262 & 2.785 & 0.518 & 0.618 & 0.717 & -0.013 & -0.013 & -0.013 & 0.000 & 0.000 & 0.001 \\
CORP/ICLN & -3.273 & -2.023 & -0.903 & -1.916 & -1.192 & -0.534 & -5.163 & -3.140 & -1.393 & -1.353 & -1.039 & -0.726 & -0.051 & -0.050 & -0.049 & -0.026 & -0.024 & -0.021 \\
CORP/IEV & 5.729 & 7.660 & 9.722 & 3.440 & 4.607 & 5.847 & 8.389 & 11.121 & 13.879 & 0.379 & 0.959 & 2.412 & -0.035 & -0.030 & -0.026 & -0.038 & -0.034 & -0.030 \\
CORP/IGOV & -0.076 & 1.020 & 2.113 & -0.094 & 0.650 & 1.413 & -0.242 & 1.650 & 3.540 & -1.259 & -0.845 & -0.393 & -0.066 & -0.059 & -0.054 & -0.064 & -0.061 & -0.058 \\
CORP/IVV & 17.006 & 20.920 & 25.063 & 9.316 & 11.506 & 13.763 & 20.544 & 24.751 & 28.680 & 9.395 & 12.834 & 16.120 & 0.006 & 0.014 & 0.020 & 0.032 & 0.041 & 0.053 \\
CORP/QQQM & 4.297 & 4.808 & 5.521 & 6.254 & 7.032 & 8.085 & 14.051 & 15.619 & 17.630 & 1.500 & 2.201 & 2.897 & -0.027 & -0.025 & -0.024 & 0.015 & 0.017 & 0.020 \\
CORP/UNG & 0.235 & 0.347 & 0.458 & 0.153 & 0.216 & 0.274 & 0.405 & 0.573 & 0.728 & 0.247 & 0.285 & 0.323 & 0.015 & 0.016 & 0.016 & 0.018 & 0.018 & 0.019 \\
\bottomrule
\end{tabular}}
\caption{Hedged asset: Corporate bonds. We compute the difference, multiplied by 100, of various performance metrics obtained using the robust Hedge Ratios relative to those obtained using the standard Hedge Ratios. P\&L, SR, $\Omega$, DD, VaR, ES correspond to profit and loss, Sharpe ratio, Omega ratio, maximum drawdown, 95\% Value at Risk, and 95\% Expected Shortfall, respectively. Transaction costs are indicated in curly brackets.}
\label{tab:transaction_costs_corp}
\end{sidewaystable}

\begin{sidewaystable}[htbp]
\centering
\scalebox{0.55}{
\begin{tabular}{@{\extracolsep{5pt}} ccccccccccccccccccccccccccc}
\toprule
Instrument & P\&L & P\&L(5bp) & P\&L(10bp) & SR & SR(5bp) & SR(10bp) & $\Omega$ & $\Omega$(5bp) & $\Omega$(10bp) & DD & DD(5bp) & DD(10bp) & VaR & VaR(5bp) & VaR(10bp) & ES & ES(5bp) & ES(10bp) \\
\midrule
EWH/AAAU & 2.614 & 3.126 & 3.567 & 0.563 & 0.669 & 0.763 & 0.475 & 0.601 & 0.727 & -0.074 & -0.071 & -0.069 & -0.035 & -0.037 & -0.039 & 1.498 & 1.792 & 2.042 \\
EWH/ASHR & -4.902 & -1.921 & 1.047 & -1.161 & -0.797 & -0.429 & -5.418 & -3.390 & -1.433 & -0.191 & -0.188 & -0.186 & 0.144 & 0.156 & 0.168 & -1.928 & -0.303 & 1.269 \\
EWH/BNO & 5.307 & 7.107 & 8.931 & 1.028 & 1.341 & 1.656 & 3.038 & 3.669 & 4.315 & 0.000 & 0.003 & 0.007 & -0.085 & -0.091 & -0.096 & 2.526 & 3.278 & 4.030 \\
EWH/CORP & 6.678 & 12.457 & 18.134 & 0.989 & 1.881 & 2.756 & 3.819 & 5.964 & 8.420 & 0.120 & 0.126 & 0.134 & 0.114 & 0.136 & 0.158 & 2.476 & 4.643 & 6.710 \\
EWH/GOVT & -6.029 & 0.810 & 7.809 & -0.969 & 0.237 & 1.467 & -3.843 & -1.046 & 1.518 & -0.109 & -0.083 & -0.040 & -0.095 & -0.129 & -0.159 & -2.550 & 0.605 & 3.597 \\
EWH/GSG & 4.870 & 7.257 & 9.679 & 0.966 & 1.384 & 1.807 & 2.771 & 3.607 & 4.426 & -0.119 & -0.114 & -0.110 & -0.060 & -0.069 & -0.078 & 2.367 & 3.367 & 4.364 \\
EWH/ICLN & -12.166 & -9.476 & -6.852 & -2.117 & -1.594 & -1.083 & -7.184 & -6.050 & -4.938 & -0.247 & -0.248 & -0.248 & -0.368 & -0.375 & -0.382 & -5.630 & -4.201 & -2.838 \\
EWH/IEV & 15.802 & 21.978 & 28.802 & 3.373 & 4.584 & 5.917 & 6.248 & 9.711 & 12.960 & -0.275 & -0.266 & -0.260 & -0.325 & -0.337 & -0.348 & 8.560 & 11.352 & 14.279 \\
EWH/IGOV & -0.296 & 0.340 & 0.974 & -0.033 & 0.069 & 0.171 & -0.364 & -0.074 & 0.215 & 0.001 & 0.001 & 0.003 & -0.011 & -0.012 & -0.014 & -0.075 & 0.186 & 0.445 \\
EWH/IVV & 49.291 & 59.581 & 70.644 & 8.743 & 10.512 & 12.381 & 22.454 & 26.305 & 30.106 & 0.108 & 0.122 & 0.139 & -0.019 & -0.051 & -0.084 & 19.527 & 22.868 & 26.120 \\
EWH/QQQM & 10.174 & 11.649 & 13.592 & 3.195 & 3.668 & 4.276 & 2.805 & 3.053 & 3.299 & 0.106 & 0.117 & 0.136 & 0.014 & 0.026 & 0.039 & 8.762 & 10.036 & 11.580 \\
EWH/UNG & -0.086 & 0.177 & 0.441 & 0.015 & 0.054 & 0.093 & -0.541 & -0.451 & -0.362 & -0.029 & -0.029 & -0.029 & -0.005 & -0.005 & -0.006 & 0.149 & 0.250 & 0.351 \\
\bottomrule
\end{tabular}}
\caption{Hedged asset: Hong Kong index. We compute the difference, multiplied by 100, of various performance metrics obtained using the robust Hedge Ratios relative to those obtained using the standard Hedge Ratios. P\&L, SR, $\Omega$, DD, VaR, ES correspond to profit and loss, Sharpe ratio, Omega ratio, maximum drawdown, 95\% Value at Risk, and 95\% Expected Shortfall, respectively. Transaction costs are indicated in curly brackets.}
\label{tab:transaction_costs_ewh}
\end{sidewaystable}

\begin{sidewaystable}
\begin{table}[H] 
\centering 
\scalebox{0.5}{
\begin{tabular}{@{\extracolsep{5pt}} ccccccccccccccccccccccccccc} 
\\[-1.8ex]\hline 
\hline \\[-1.8ex] 
Instrument & P\&L & P\&L(5bp) & P\&L(10bp) & SR & SR(5bp) & SR(10bp) & $\Omega$ & $\Omega$(5bp) & $\Omega$(10bp) & DD & DD(5bp) & DD(10bp) & VaR & VaR(5bp) & VaR(10bp) & ES & ES(5bp) & ES(10bp) \\
\midrule
GOVT/AAAU & 2.454 & 2.866 & 3.224 & 2.207 & 2.581 & 2.908 & 0.360 & 0.494 & 0.759 & -0.008 & -0.001 & 0.000 & -0.033 & -0.034 & -0.035 & 5.626 & 6.652 & 7.459 \\
GOVT/ASHR & -0.599 & -0.211 & 0.179 & -0.329 & -0.096 & 0.150 & 0.161 & 0.267 & 0.372 & -0.019 & -0.018 & -0.012 & -0.024 & -0.025 & -0.027 & -0.818 & -0.242 & 0.367 \\
GOVT/BNO & -0.116 & 0.346 & 0.798 & 0.053 & 0.332 & 0.626 & -0.204 & -0.072 & 0.061 & -0.030 & -0.025 & -0.024 & -0.024 & -0.025 & -0.026 & 0.135 & 0.810 & 1.512 \\
GOVT/CORP & 1.614 & 3.865 & 6.084 & 1.388 & 2.888 & 4.408 & -0.974 & -0.077 & 1.890 & -0.067 & -0.076 & -0.070 & -0.041 & -0.048 & -0.052 & 3.621 & 7.231 & 10.745 \\
GOVT/EWH & -0.556 & 0.048 & 0.649 & -0.257 & 0.123 & 0.522 & 0.085 & 0.258 & 0.431 & -0.057 & -0.052 & -0.052 & -0.027 & -0.029 & -0.029 & -0.618 & 0.314 & 1.284 \\
GOVT/GSG & 0.026 & 0.671 & 1.302 & 0.123 & 0.523 & 0.931 & 0.055 & 0.245 & 0.435 & -0.012 & -0.011 & -0.010 & -0.021 & -0.023 & -0.024 & 0.310 & 1.273 & 2.236 \\
GOVT/ICLN & 0.728 & 1.140 & 1.567 & 0.558 & 0.816 & 1.096 & 0.035 & 0.166 & 0.297 & -0.007 & -0.001 & -0.001 & -0.023 & -0.024 & -0.025 & 1.351 & 1.974 & 2.630 \\
GOVT/IEV & -1.229 & -0.164 & 0.794 & -0.693 & -0.039 & 0.565 & -0.081 & 0.179 & 0.438 & -0.028 & -0.021 & -0.014 & -0.025 & -0.028 & -0.031 & -1.726 & -0.140 & 1.299 \\
GOVT/IGOV & -0.154 & 0.365 & 0.881 & 0.075 & 0.420 & 0.796 & -0.965 & -0.803 & -0.622 & -0.066 & -0.068 & -0.067 & -0.061 & -0.062 & -0.061 & 0.181 & 1.011 & 1.899 \\
GOVT/IVV & -3.220 & -1.536 & -0.041 & -1.961 & -0.907 & 0.046 & -0.695 & -0.246 & 0.200 & 0.011 & 0.013 & 0.015 & -0.013 & -0.018 & -0.022 & -4.862 & -2.246 & 0.073 \\
GOVT/QQQM & 1.365 & 1.543 & 1.777 & 2.087 & 2.367 & 2.721 & 0.225 & 0.316 & 0.406 & 0.000 & 0.002 & 0.002 & -0.004 & -0.005 & -0.005 & 5.132 & 5.746 & 6.554 \\
GOVT/UNG & -0.043 & -0.035 & -0.027 & -0.021 & -0.017 & -0.012 & -0.006 & -0.003 & -0.000 & 0.002 & 0.001 & 0.001 & -0.002 & -0.002 & -0.002 & -0.052 & -0.044 & -0.031 \\
\hline \\[-1.8ex] 
\end{tabular}} 
\caption{Hedged asset: US Treasury bonds. We compute the difference, multiplied by 100, of various performance metrics obtained using the robust Hedge Ratios relative to those obtained using the standard Hedge Ratios. P\&L, SR, $\Omega$, DD, VaR, ES correspond to profit and loss, Sharpe ratio, Omega ratio, maximum drawdown, 95\% Value at Risk, and 95\% Expected Shortfall, respectively. Transaction costs are indicated in curly brackets.} 
\label{tab:transaction_costs_govt} 
\end{table}
\end{sidewaystable}

\begin{sidewaystable}
\begin{table}[H] 
\centering 
\scalebox{0.5}{
\begin{tabular}{@{\extracolsep{5pt}} ccccccccccccccccccccccccccc} 
\\[-1.8ex]\hline 
\hline \\[-1.8ex] 
Instrument & P\&L & P\&L(5bp) & P\&L(10bp) & SR & SR(5bp) & SR(10bp) & $\Omega$ & $\Omega$(5bp) & $\Omega$(10bp) & DD & DD(5bp) & DD(10bp) & VaR & VaR(5bp) & VaR(10bp) & ES & ES(5bp) & ES(10bp) \\
\midrule
GSG/AAAU & 6.694 & 7.661 & 8.487 & 2.210 & 2.522 & 2.790 & 4.120 & 4.692 & 5.259 & -0.123 & -0.115 & -0.110 & -0.003 & -0.007 & -0.010 & 5.283 & 6.007 & 6.618 \\
GSG/ASHR & -9.092 & -6.787 & -4.493 & -1.556 & -1.106 & -0.653 & -1.412 & -1.281 & -1.068 & 0.146 & 0.161 & 0.171 & 0.032 & 0.037 & 0.043 & -5.812 & -4.307 & -2.839 \\
GSG/BNO & 25.992 & 30.265 & 34.600 & 10.643 & 13.179 & 15.681 & 3.154 & 5.820 & 8.539 & -0.717 & -0.708 & -0.696 & -0.887 & -0.896 & -0.905 & 25.702 & 30.812 & 35.566 \\
GSG/CORP & 1.565 & 3.982 & 6.392 & 0.314 & 0.820 & 1.323 & 0.657 & 1.542 & 2.670 & 0.008 & 0.020 & 0.029 & 0.017 & 0.011 & 0.005 & 0.846 & 2.212 & 3.555 \\
GSG/EWH & -1.194 & 1.480 & 4.172 & -0.258 & 0.328 & 0.920 & 0.516 & 1.156 & 2.638 & -0.028 & -0.012 & 0.002 & -0.010 & -0.016 & -0.021 & -0.991 & 0.713 & 2.381 \\
GSG/GOVT & -0.694 & 6.302 & 13.655 & -0.346 & 1.326 & 3.085 & -0.481 & 0.385 & 2.653 & -0.056 & -0.037 & -0.017 & -0.109 & -0.128 & -0.147 & -0.924 & 3.500 & 7.732 \\
GSG/ICLN & -8.462 & -5.882 & -3.368 & -1.885 & -1.327 & -0.782 & -0.795 & -0.563 & 0.436 & -0.033 & -0.024 & -0.015 & -0.015 & -0.019 & -0.023 & -5.450 & -3.806 & -2.234 \\
GSG/IEV & 18.833 & 24.431 & 30.649 & 4.008 & 5.213 & 6.544 & 10.523 & 13.025 & 15.441 & 0.123 & 0.131 & 0.142 & 0.061 & 0.071 & 0.081 & 10.406 & 13.307 & 16.411 \\
GSG/IGOV & -0.062 & 0.165 & 0.392 & -0.019 & 0.028 & 0.076 & 0.022 & 0.062 & 0.103 & 0.003 & 0.003 & 0.004 & -0.002 & -0.003 & -0.004 & -0.053 & 0.077 & 0.207 \\
GSG/IVV & 41.902 & 50.203 & 59.324 & 8.604 & 10.293 & 12.125 & 21.174 & 26.453 & 31.238 & 0.144 & 0.162 & 0.177 & 0.282 & 0.264 & 0.246 & 20.735 & 24.244 & 27.844 \\
GSG/QQQM & 4.953 & 5.606 & 6.557 & 2.539 & 2.899 & 3.394 & 2.287 & 2.635 & 2.981 & -0.102 & -0.101 & -0.102 & -0.047 & -0.051 & -0.056 & 6.351 & 7.109 & 8.237 \\
GSG/UNG & -2.584 & -2.051 & -1.409 & -0.601 & -0.485 & -0.346 & -0.867 & -0.799 & -0.732 & -0.008 & -0.008 & -0.006 & -0.052 & -0.053 & -0.054 & -1.613 & -1.290 & -0.908 \\
\hline \\[-1.8ex] 
\end{tabular}} 
\caption{Hedged asset: GSCI Commodity index. We compute the difference, multiplied by 100, of various performance metrics obtained using the robust Hedge Ratios relative to those obtained using the standard Hedge Ratios. P\&L, SR, $\Omega$, DD, VaR, ES correspond to profit and loss, Sharpe ratio, Omega ratio, maximum drawdown, 95\% Value at Risk, and 95\% Expected Shortfall, respectively. Transaction costs are indicated in curly brackets.} 
\label{tab:transaction_costs_gsg} 
\end{table}
\end{sidewaystable}

\begin{sidewaystable}
\begin{table}[H] 
\centering 
\scalebox{0.5}{
\begin{tabular}{@{\extracolsep{5pt}} ccccccccccccccccccccccccccc} 
\\[-1.8ex]\hline 
\hline \\[-1.8ex] 
Instrument & P\&L & P\&L(5bp) & P\&L(10bp) & SR & SR(5bp) & SR(10bp) & DD & DD(5bp) & DD(10bp) & VaR & VaR(5bp) & VaR(10bp) & ES & ES(5bp) & ES(10bp) & $\Omega$ & $\Omega$(5bp) & $\Omega$(10bp) \\ 
\hline \\[-1.8ex] 
ICLN/AAAU & 6.453 & 7.606 & 8.537 & 1.439 & 1.673 & 1.862 & 0.009 & 0.011 & 0.012 & -0.114 & -0.111 & -0.108 & 3.513 & 3.977 & 4.438 & 3.244 & 3.768 & 4.188 \\
ICLN/ASHR & -11.550 & -7.395 & -3.254 & -1.530 & -0.866 & -0.200 & -7.485 & -4.941 & -2.522 & -0.145 & -0.136 & -0.126 & -0.049 & -0.047 & -0.045 & -3.567 & -2.046 & -0.555 \\
ICLN/BNO & 8.728 & 10.923 & 13.148 & 1.420 & 1.749 & 2.083 & 3.939 & 5.038 & 6.107 & -0.123 & -0.117 & -0.111 & 0.038 & 0.041 & 0.046 & 3.028 & 3.717 & 4.402 \\
ICLN/CORP & 8.504 & 17.608 & 27.121 & 1.362 & 2.734 & 4.152 & 4.104 & 8.589 & 12.647 & -0.022 & -0.015 & 0.010 & -0.014 & 0.016 & 0.043 & 3.020 & 5.939 & 8.817 \\
ICLN/EWH & -7.136 & -3.234 & 0.693 & -0.513 & 0.233 & 0.989 & -5.635 & -3.144 & -0.810 & -0.328 & -0.317 & -0.307 & -0.290 & -0.276 & -0.262 & -1.005 & 0.660 & 2.293 \\
ICLN/GOVT & -4.964 & 4.373 & 14.032 & -0.409 & 1.234 & 2.931 & -3.409 & 2.020 & 6.410 & -0.194 & -0.161 & -0.129 & -0.134 & -0.097 & -0.115 & -0.973 & 2.472 & 5.702 \\
ICLN/GSG & 7.237 & 10.329 & 13.469 & 1.279 & 1.751 & 2.230 & 2.919 & 4.477 & 5.973 & -0.181 & -0.173 & -0.164 & 0.019 & 0.025 & 0.036 & 2.708 & 3.690 & 4.665 \\
ICLN/IEV & 27.040 & 35.567 & 45.051 & 6.658 & 8.333 & 10.181 & 11.998 & 15.487 & 19.153 & -0.492 & -0.472 & -0.451 & -0.309 & -0.294 & -0.278 & 13.285 & 16.262 & 19.329 \\
ICLN/IGOV & -0.124 & 1.612 & 3.342 & 0.101 & 0.362 & 0.624 & -0.380 & 0.529 & 1.417 & -0.060 & -0.057 & -0.054 & -0.046 & -0.041 & -0.037 & 0.237 & 0.813 & 1.384 \\
ICLN/IVV & 73.399 & 87.629 & 103.078 & 13.862 & 16.426 & 19.122 & 27.456 & 30.572 & 33.463 & 0.004 & 0.044 & 0.090 & 0.058 & 0.095 & 0.135 & 25.593 & 29.172 & 32.529 \\
ICLN/QQQM & 21.426 & 24.363 & 28.429 & 6.636 & 7.568 & 8.898 & 11.507 & 13.858 & 16.093 & 0.232 & 0.280 & 0.331 & 0.511 & 0.540 & 0.569 & 13.719 & 15.030 & 17.174 \\
ICLN/UNG & -2.160 & -1.930 & -1.667 & -0.253 & -0.219 & -0.180 & -1.163 & -1.014 & -0.866 & -0.062 & -0.062 & -0.061 & 0.019 & 0.020 & 0.020 & -0.590 & -0.515 & -0.428 \\
\hline \\[-1.8ex] 
\end{tabular}} 
\caption{Hedged asset: Clean energy index. We compute the difference, multiplied by 100, of various performance metrics obtained using the robust Hedge Ratios relative to those obtained using the standard Hedge Ratios. P\&L, SR, $\Omega$, DD, VaR, ES correspond to profit and loss, Sharpe ratio, Omega ratio, maximum drawdown, 95\% Value at Risk, and 95\% Expected Shortfall, respectively. Transaction costs are indicated in curly brackets.} 
\label{tab:transaction_costs_icln} 
\end{table}
\end{sidewaystable}

\begin{sidewaystable}
\begin{table}[H] 
\centering 
\scalebox{0.5}{
\begin{tabular}{@{\extracolsep{5pt}} ccccccccccccccccccccccccccc} 
\\[-1.8ex]\hline 
\hline \\[-1.8ex] 
Instrument & P\&L & P\&L(5bp) & P\&L(10bp) & SR & SR(5bp) & SR(10bp) & $\Omega$ & $\Omega$(5bp) & $\Omega$(10bp) & DD & DD(5bp) & DD(10bp) & VaR & VaR(5bp) & VaR(10bp) & ES & ES(5bp) & ES(10bp) \\
\midrule
IEV/AAAU & 4.648 & 5.474 & 6.156 & 1.693 & 2.012 & 2.284 & 0.970 & 1.261 & 1.551 & -0.093 & -0.090 & -0.088 & -0.109 & -0.106 & -0.103 & 4.587 & 5.494 & 6.216 \\
IEV/ASHR & -6.818 & -3.953 & -1.098 & -2.139 & -1.294 & -0.458 & 0.955 & 1.314 & 1.670 & -0.185 & -0.179 & -0.174 & -0.120 & -0.115 & -0.110 & -5.848 & -3.392 & -1.069 \\
IEV/BNO & 8.614 & 10.382 & 12.169 & 2.262 & 2.742 & 3.215 & 3.538 & 3.889 & 4.238 & -0.037 & -0.033 & -0.029 & -0.048 & -0.042 & -0.036 & 6.111 & 7.269 & 8.454 \\
IEV/CORP & 4.279 & 10.054 & 15.737 & 1.031 & 2.620 & 4.167 & -0.880 & 0.512 & 1.880 & -0.017 & -0.006 & 0.006 & -0.040 & -0.028 & -0.017 & 2.809 & 7.000 & 10.937 \\
IEV/EWH & -6.318 & -3.135 & 0.064 & -2.685 & -1.558 & -0.445 & 1.353 & 2.007 & 2.654 & -0.194 & -0.188 & -0.188 & -0.107 & -0.100 & -0.093 & -8.539 & -4.901 & -1.558 \\
IEV/GOVT & -1.958 & 3.611 & 9.336 & -0.672 & 0.943 & 2.580 & 0.655 & 0.988 & 2.057 & -0.152 & -0.122 & -0.084 & -0.077 & -0.059 & -0.041 & -1.769 & 2.554 & 6.556 \\
IEV/GSG & 7.492 & 10.096 & 12.732 & 1.980 & 2.694 & 3.404 & 2.760 & 3.255 & 3.746 & -0.044 & -0.037 & -0.032 & -0.066 & -0.058 & -0.049 & 5.384 & 7.131 & 8.888 \\
IEV/ICLN & -13.227 & -9.933 & -6.706 & -5.704 & -4.485 & -3.315 & -2.312 & -1.958 & -1.605 & -0.344 & -0.335 & -0.327 & -0.316 & -0.309 & -0.303 & -17.815 & -13.545 & -9.818 \\
IEV/IGOV & -0.165 & 1.408 & 2.976 & -0.157 & 0.294 & 0.733 & -0.756 & -0.439 & -0.124 & -0.082 & -0.077 & -0.076 & -0.032 & -0.029 & -0.027 & -0.439 & 0.815 & 2.009 \\
IEV/IVV & 73.343 & 87.386 & 102.371 & 21.758 & 25.806 & 30.006 & 41.157 & 47.296 & 52.601 & 0.211 & 0.274 & 0.328 & 0.300 & 0.341 & 0.392 & 45.532 & 51.091 & 56.574 \\
IEV/QQQM & 14.455 & 16.265 & 18.803 & 9.243 & 10.418 & 12.041 & 10.059 & 11.518 & 12.939 & -0.023 & -0.012 & -0.001 & 0.109 & 0.126 & 0.143 & 20.501 & 22.482 & 25.273 \\
IEV/UNG & -0.470 & -0.289 & -0.076 & -0.112 & -0.069 & -0.019 & 0.106 & 0.159 & 0.212 & -0.052 & -0.051 & -0.050 & -0.014 & -0.013 & -0.012 & -0.339 & -0.205 & -0.061 \\
\hline \\[-1.8ex] 
\end{tabular}} 
\caption{Hedged asset: S\&P Europe 350. We compute the difference, multiplied by 100, of various performance metrics obtained using the robust Hedge Ratios relative to those obtained using the standard Hedge Ratios. P\&L, SR, $\Omega$, DD, VaR, ES correspond to profit and loss, Sharpe ratio, Omega ratio, maximum drawdown, 95\% Value at Risk, and 95\% Expected Shortfall, respectively. Transaction costs are indicated in curly brackets.} 
\label{tab:transaction_costs_iev} 
\end{table}
\end{sidewaystable}

\begin{sidewaystable}
\begin{table}[H] 
\centering 
\scalebox{0.5}{
\begin{tabular}{@{\extracolsep{5pt}} ccccccccccccccccccccccccccc} 
\\[-1.8ex]\hline 
\hline \\[-1.8ex] 
Instrument & P\&L & P\&L(5bp) & P\&L(10bp) & SR & SR(5bp) & SR(10bp) & $\Omega$ & $\Omega$(5bp) & $\Omega$(10bp) & DD & DD(5bp) & DD(10bp) & VaR & VaR(5bp) & VaR(10bp) & ES & ES(5bp) & ES(10bp) \\
\midrule
IGOV/AAAU & 6.039 & 7.069 & 7.949 & 3.402 & 3.990 & 4.497 & 4.538 & 5.117 & 5.691 & -0.023 & -0.021 & -0.019 & -0.096 & -0.093 & -0.091 & 8.659 & 10.121 & 11.320 \\
IGOV/ASHR & -1.039 & -0.648 & -0.258 & -0.388 & -0.237 & -0.086 & -0.404 & -0.364 & -0.325 & -0.023 & -0.021 & -0.020 & -0.024 & -0.024 & -0.023 & -0.995 & -0.606 & -0.219 \\
IGOV/BNO & 0.019 & 0.112 & 0.204 & 0.010 & 0.045 & 0.080 & -0.109 & -0.101 & -0.093 & 0.001 & 0.002 & 0.002 & -0.001 & 0.000 & 0.000 & 0.029 & 0.118 & 0.207 \\
IGOV/CORP & 2.706 & 6.159 & 9.553 & 1.246 & 2.715 & 4.158 & -0.517 & 1.266 & 3.932 & -0.143 & -0.141 & -0.138 & -0.112 & -0.102 & -0.091 & 3.208 & 6.810 & 10.212 \\
IGOV/EWH & -0.329 & 0.035 & 0.400 & -0.113 & 0.029 & 0.171 & -0.408 & -0.368 & -0.327 & -0.026 & -0.025 & -0.025 & -0.039 & -0.038 & -0.037 & -0.291 & 0.070 & 0.433 \\
IGOV/GOVT & -0.359 & 2.567 & 5.387 & -0.069 & 1.544 & 3.092 & -0.114 & 1.830 & 3.438 & -0.179 & -0.177 & -0.171 & -0.186 & -0.179 & -0.171 & -0.150 & 3.897 & 7.469 \\
IGOV/GSG & 0.414 & 0.530 & 0.647 & 0.159 & 0.203 & 0.247 & -0.021 & 0.034 & 0.050 & -0.012 & -0.011 & -0.011 & -0.005 & -0.004 & -0.004 & 0.402 & 0.513 & 0.624 \\
IGOV/ICLN & -2.164 & -1.474 & -0.793 & -0.838 & -0.566 & -0.295 & -1.138 & -0.989 & -0.801 & -0.042 & -0.042 & -0.040 & -0.055 & -0.053 & -0.052 & -2.178 & -1.462 & -0.757 \\
IGOV/IEV & 5.189 & 6.854 & 8.631 & 2.193 & 2.879 & 3.610 & 1.377 & 2.286 & 3.887 & -0.018 & -0.015 & -0.010 & -0.094 & -0.090 & -0.086 & 5.452 & 7.087 & 8.790 \\
IGOV/IVV & 6.531 & 7.971 & 9.467 & 2.569 & 3.135 & 3.721 & 2.532 & 3.826 & 5.202 & 0.047 & 0.051 & 0.055 & -0.016 & -0.014 & -0.011 & 6.278 & 7.587 & 8.917 \\
IGOV/QQQM & 4.824 & 5.405 & 6.205 & 4.155 & 4.667 & 5.360 & 2.172 & 2.932 & 3.685 & 0.044 & 0.047 & 0.048 & -0.014 & -0.010 & -0.006 & 9.114 & 10.063 & 11.434 \\
IGOV/UNG & -0.229 & -0.144 & -0.061 & -0.066 & -0.037 & -0.009 & 0.084 & 0.098 & 0.113 & 0.005 & 0.006 & 0.006 & 0.020 & 0.020 & 0.021 & -0.339 & -0.205 & -0.034 \\
\hline \\[-1.8ex] 
\end{tabular}} 
\caption{Hedged asset: Treasury bonds. We compute the difference, multiplied by 100, of various performance metrics obtained using the robust Hedge Ratios relative to those obtained using the standard Hedge Ratios. P\&L, SR, $\Omega$, DD, VaR, ES correspond to profit and loss, Sharpe ratio, Omega ratio, maximum drawdown, 95\% Value at Risk, and 95\% Expected Shortfall, respectively. Transaction costs are indicated in curly brackets.} 
\label{tab:transaction_costs_igov} 
\end{table}
\end{sidewaystable}

\begin{sidewaystable}
\begin{table}[H] 
\centering 
\scalebox{0.5}{
\begin{tabular}{@{\extracolsep{5pt}} ccccccccccccccccccccccccccc} 
\\[-1.8ex]\hline 
\hline \\[-1.8ex] 
Instrument & P\&L & P\&L(5bp) & P\&L(10bp) & SR & SR(5bp) & SR(10bp) & $\Omega$ & $\Omega$(5bp) & $\Omega$(10bp) & DD & DD(5bp) & DD(10bp) & VaR & VaR(5bp) & VaR(10bp) & ES & ES(5bp) & ES(10bp) \\
\midrule
IVV/AAAU & 2.985 & 3.439 & 3.837 & 0.907 & 1.099 & 1.264 & 0.050 & -0.080 & -0.110 & -0.019 & -0.018 & -0.017 & -0.046 & -0.044 & -0.042 & 3.527 & 4.187 & 4.766 \\
IVV/ASHR & 0.831 & 3.227 & 5.615 & 0.037 & 0.773 & 1.498 & -0.323 & -0.388 & -0.453 & -0.129 & -0.128 & -0.127 & 0.031 & 0.036 & 0.041 & 0.524 & 3.020 & 5.429 \\
IVV/BNO & 10.219 & 11.616 & 13.033 & 2.925 & 3.348 & 3.758 & 1.814 & 1.621 & 1.427 & 0.003 & 0.004 & 0.006 & -0.010 & -0.007 & -0.005 & 10.794 & 11.949 & 13.202 \\
IVV/CORP & 4.191 & 10.368 & 16.425 & 1.116 & 2.897 & 4.644 & 1.077 & 0.326 & -0.432 & 0.012 & 0.021 & 0.030 & 0.030 & 0.044 & 0.059 & 3.769 & 9.819 & 15.490 \\
IVV/EWH & 2.350 & 5.375 & 8.415 & -0.165 & 0.911 & 1.966 & 3.021 & 2.547 & 2.070 & -0.161 & -0.153 & -0.147 & -0.156 & -0.147 & -0.138 & 1.121 & 4.742 & 8.191 \\
IVV/GOVT & 3.098 & 8.859 & 14.822 & 0.672 & 2.505 & 4.380 & 1.021 & 0.028 & -0.942 & -0.077 & -0.063 & -0.051 & -0.003 & 0.012 & 0.026 & 2.354 & 8.333 & 13.790 \\
IVV/GSG & 8.432 & 10.382 & 12.365 & 2.348 & 2.943 & 3.529 & 1.254 & 0.983 & 0.710 & -0.015 & -0.013 & -0.011 & -0.048 & -0.044 & -0.041 & 8.903 & 10.623 & 12.407 \\
IVV/ICLN & -7.308 & -4.722 & -2.201 & -4.865 & -3.874 & -2.905 & 0.365 & -0.482 & -0.600 & -0.200 & -0.197 & -0.196 & -0.257 & -0.253 & -0.249 & -19.810 & -15.214 & -11.227 \\
IVV/IEV & 26.224 & 32.757 & 39.992 & 7.106 & 10.154 & 13.291 & 0.620 & -2.607 & -3.988 & -0.250 & -0.227 & -0.216 & -0.462 & -0.446 & -0.428 & 24.481 & 32.611 & 40.559 \\
IVV/IGOV & 0.439 & 1.184 & 1.927 & -0.008 & 0.217 & 0.437 & -0.125 & -0.264 & -0.404 & -0.022 & -0.021 & -0.019 & -0.019 & -0.016 & -0.014 & -0.065 & 0.732 & 1.500 \\
IVV/QQQM & 14.904 & 16.359 & 18.604 & 16.569 & 19.243 & 22.692 & 6.323 & 4.344 & 2.420 & -0.169 & -0.158 & -0.137 & -0.356 & -0.346 & -0.335 & 47.563 & 52.442 & 58.892 \\
IVV/UNG & -0.247 & -0.062 & 0.158 & -0.148 & -0.092 & -0.030 & 0.084 & 0.056 & 0.029 & -0.044 & -0.043 & -0.044 & -0.014 & -0.013 & -0.013 & -0.421 & -0.225 & -0.008 \\
\hline \\[-1.8ex] 
\end{tabular}} 
\caption{Hedged asset: S\&P 500. We compute the difference, multiplied by 100, of various performance metrics obtained using the robust Hedge Ratios relative to those obtained using the standard Hedge Ratios. P\&L, SR, $\Omega$, DD, VaR, ES correspond to profit and loss, Sharpe ratio, Omega ratio, maximum drawdown, 95\% Value at Risk, and 95\% Expected Shortfall, respectively. Transaction costs are indicated in curly brackets.} 
\label{tab:transaction_costs_ivv} 
\end{table}
\end{sidewaystable}

\begin{sidewaystable}
\begin{table}[H] 
\centering 
\scalebox{0.5}{
\begin{tabular}{@{\extracolsep{5pt}} ccccccccccccccccccccccccccc} 
\\[-1.8ex]\hline 
\hline \\[-1.8ex] 
Instrument & P\&L & P\&L(5bp) & P\&L(10bp) & SR & SR(5bp) & SR(10bp) & $\Omega$ & $\Omega$(5bp) & $\Omega$(10bp) & DD & DD(5bp) & DD(10bp) & VaR & VaR(5bp) & VaR(10bp) & ES & ES(5bp) & ES(10bp) \\
\midrule
QQQM/AAAU & 6.948 & 7.363 & 8.180 & 2.860 & 3.107 & 3.490 & -0.022 & -0.204 & -0.387 & -0.230 & -0.219 & -0.212 & -0.078 & -0.076 & -0.074 & 8.604 & 9.121 & 10.114 \\
QQQM/ASHR & 5.095 & 5.447 & 6.108 & 2.261 & 2.388 & 2.649 & 1.348 & 1.154 & 0.959 & 0.249 & 0.250 & 0.252 & 0.043 & 0.046 & 0.049 & 7.033 & 7.220 & 7.904 \\
QQQM/BNO & 1.734 & 1.953 & 2.325 & 0.574 & 0.693 & 0.856 & -0.472 & -0.506 & -0.541 & -0.024 & -0.021 & -0.021 & -0.069 & -0.068 & -0.066 & 1.930 & 2.253 & 2.745 \\
QQQM/CORP & 1.281 & 2.798 & 4.350 & 0.419 & 1.089 & 1.760 & 0.534 & 0.357 & 0.180 & 0.046 & 0.048 & 0.052 & -0.054 & -0.043 & -0.032 & 1.329 & 3.389 & 5.417 \\
QQQM/EWH & 3.282 & 4.139 & 4.896 & 1.300 & 1.707 & 2.048 & 1.032 & 0.812 & 0.591 & -0.095 & -0.089 & -0.060 & -0.156 & -0.151 & -0.146 & 4.539 & 5.531 & 6.452 \\
QQQM/GOVT & 0.461 & 2.000 & 3.588 & 0.099 & 0.771 & 1.454 & 0.445 & 0.288 & 0.130 & -0.057 & -0.043 & -0.038 & 0.002 & 0.012 & 0.022 & 0.250 & 2.356 & 4.431 \\
QQQM/GSG & 1.861 & 2.290 & 2.902 & 0.599 & 0.816 & 1.086 & -0.581 & -0.647 & -0.714 & -0.075 & -0.060 & -0.060 & -0.066 & -0.063 & -0.059 & 2.051 & 2.663 & 3.468 \\
QQQM/ICLN & -4.063 & -3.245 & -1.735 & -2.235 & -1.810 & -1.099 & 0.030 & -0.088 & -0.207 & -0.177 & -0.176 & -0.166 & -0.329 & -0.322 & -0.314 & -6.515 & -5.058 & -2.763 \\
QQQM/IEV & 7.449 & 8.895 & 11.152 & 2.753 & 3.723 & 5.007 & -0.530 & -0.945 & -1.361 & -0.233 & -0.230 & -0.215 & -0.446 & -0.436 & -0.426 & 9.003 & 11.523 & 14.946 \\
QQQM/IGOV & -0.025 & 0.569 & 1.227 & -0.057 & 0.202 & 0.481 & 0.231 & 0.160 & 0.089 & -0.108 & -0.100 & -0.094 & 0.003 & 0.005 & 0.008 & -0.164 & 0.653 & 1.525 \\
QQQM/IVV & 27.443 & 30.925 & 35.923 & 21.710 & 25.335 & 29.843 & 16.010 & 12.397 & 9.080 & -0.238 & -0.215 & -0.183 & -0.484 & -0.447 & -0.410 & 52.807 & 58.027 & 64.753 \\
QQQM/UNG & -0.091 & -0.034 & 0.009 & -0.037 & -0.014 & 0.004 & -0.047 & -0.051 & -0.055 & -0.015 & -0.013 & -0.007 & -0.007 & -0.007 & -0.007 & -0.032 & 0.034 & 0.090 \\
\hline \\[-1.8ex] 
\end{tabular}} 
\caption{Hedged asset: NASDAQ 100. We compute the difference, multiplied by 100, of various performance metrics obtained using the robust Hedge Ratios relative to those obtained using the standard Hedge Ratios. P\&L, SR, $\Omega$, DD, VaR, ES correspond to profit and loss, Sharpe ratio, Omega ratio, maximum drawdown, 95\% Value at Risk, and 95\% Expected Shortfall, respectively. Transaction costs are indicated in curly brackets.} 
\label{tab:transaction_costs_qqqm} 
\end{table}
\end{sidewaystable}

\begin{sidewaystable}
\begin{table}[H] 
\centering 
\scalebox{0.5}{
\begin{tabular}{@{\extracolsep{5pt}} ccccccccccccccccccccccccccc} 
\\[-1.8ex]\hline 
\hline \\[-1.8ex] 
Instrument & P\&L & P\&L(5bp) & P\&L(10bp) & SR & SR(5bp) & SR(10bp) & $\Omega$ & $\Omega$(5bp) & $\Omega$(10bp) & DD & DD(5bp) & DD(10bp) & VaR & VaR(5bp) & VaR(10bp) & ES & ES(5bp) & ES(10bp) \\
\midrule
UNG/AAAU & 3.427 & 4.029 & 4.555 & 0.062 & 0.072 & 0.081 & 0.262 & 0.311 & 0.367 & 0.020 & 0.024 & 0.026 & 0.006 & 0.007 & 0.008 & 0.981 & 1.138 & 1.280 \\
UNG/ASHR & 1.143 & 2.324 & 3.504 & -0.035 & -0.018 & -0.001 & -0.651 & -0.860 & -1.070 & -0.075 & -0.071 & -0.066 & -0.025 & -0.022 & -0.019 & -0.081 & 0.141 & 0.360 \\
UNG/BNO & 4.954 & 6.750 & 8.551 & 0.066 & 0.092 & 0.118 & 0.824 & 0.474 & 0.121 & 0.082 & 0.082 & 0.084 & -0.098 & -0.093 & -0.089 & 0.236 & 0.578 & 0.915 \\
UNG/CORP & 6.366 & 12.860 & 19.355 & 0.073 & 0.160 & 0.248 & 3.933 & 2.804 & 1.640 & 0.077 & 0.103 & 0.152 & 0.096 & 0.117 & 0.146 & 1.558 & 2.702 & 3.839 \\
UNG/EWH & 13.162 & 16.160 & 19.159 & 0.006 & 0.053 & 0.097 & -0.566 & -1.101 & -1.647 & -0.073 & -0.070 & -0.056 & -0.162 & -0.154 & -0.130 & 1.156 & 1.733 & 2.293 \\
UNG/GOVT & -0.461 & 0.557 & 1.568 & -0.004 & 0.010 & 0.024 & 0.204 & 0.013 & -0.179 & 0.000 & 0.002 & 0.004 & -0.015 & -0.011 & -0.007 & -0.165 & 0.023 & 0.208 \\
UNG/GSG & 12.605 & 16.581 & 20.582 & 0.150 & 0.208 & 0.265 & 1.988 & 1.219 & 0.435 & 0.013 & 0.027 & 0.038 & -0.065 & -0.056 & -0.048 & 1.230 & 1.985 & 2.730 \\
UNG/ICLN & -6.839 & -5.381 & -3.932 & -0.095 & -0.075 & -0.054 & -1.619 & -1.929 & -2.241 & -0.042 & -0.042 & -0.039 & -0.094 & -0.090 & -0.086 & -1.569 & -1.290 & -1.019 \\
UNG/IEV & 8.834 & 13.187 & 17.677 & 0.047 & 0.109 & 0.173 & 0.810 & -0.077 & -0.984 & -0.045 & -0.044 & 0.032 & -0.115 & -0.103 & -0.086 & 1.013 & 1.820 & 2.636 \\
UNG/IGOV & -0.101 & 1.894 & 3.890 & 0.020 & 0.045 & 0.069 & 1.298 & 0.958 & 0.613 & 0.026 & 0.038 & 0.084 & 0.069 & 0.074 & 0.091 & 0.247 & 0.601 & 0.954 \\
UNG/IVV & 18.512 & 24.861 & 31.417 & 0.254 & 0.342 & 0.433 & 5.873 & 4.767 & 3.627 & -0.066 & -0.065 & -0.018 & -0.020 & -0.006 & 0.022 & 3.156 & 4.290 & 5.443 \\
UNG/QQQM & 3.222 & 3.645 & 4.224 & 0.426 & 0.482 & 0.559 & 1.422 & 1.269 & 1.115 & 0.062 & 0.065 & 0.068 & -0.002 & 0.000 & 0.002 & 1.113 & 1.261 & 1.457 \\
\hline \\[-1.8ex] 
\end{tabular}} 
\caption{Hedged asset: Natural gas. We compute the difference, multiplied by 100, of various performance metrics obtained using the robust Hedge Ratios relative to those obtained using the standard Hedge Ratios. P\&L, SR, $\Omega$, DD, VaR, ES correspond to profit and loss, Sharpe ratio, Omega ratio, maximum drawdown, 95\% Value at Risk, and 95\% Expected Shortfall, respectively. Transaction costs are indicated in curly brackets.} 
\label{tab:transaction_costs_ung} 
\end{table}
\end{sidewaystable}

\end{document}